\documentclass[a4paper]{article}
\usepackage{enumitem}
\usepackage{calc}

\usepackage{lmodern} 

\usepackage{graphicx} 

\usepackage{amsmath, accents}
\usepackage{amssymb,tikz}
\usepackage{pict2e}
\usepackage{amsthm}
\usepackage{amscd}
\usepackage{mathrsfs}
\usepackage{todonotes}
\usetikzlibrary{calc} 

\usepackage{yfonts}

\usepackage{marvosym}
\usepackage[percent]{overpic}

\newtheorem{theorem}{Theorem} [section]
\newtheorem{proposition}[theorem]{Proposition}	
	
\newtheorem{lemma}[theorem]{Lemma}

\newtheorem{remark}[theorem]{Remark}
\theoremstyle{definition}

\let\oldbibliography\thebibliography
\renewcommand{\thebibliography}[1]{\oldbibliography{#1}
\setlength{\itemsep}{-0.5pt}}

\def\XXint#1#2#3{{\setbox0=\hbox{$#1{#2#3}{\int}$}
\vcenter{\hbox{$#2#3$}}\kern-.5\wd0}}

\usepackage{tikz}
\usetikzlibrary{arrows}
\usetikzlibrary{decorations.pathmorphing}
\usetikzlibrary{decorations.markings}
\usetikzlibrary{patterns}
\usetikzlibrary{automata}
\usetikzlibrary{positioning}
\usepackage{tikz-cd}
\tikzset{->-/.style={decoration={
				markings,
				mark=at position #1 with {\arrow{latex}}},postaction={decorate}}}
	
	\tikzset{-<-/.style={decoration={
				markings,
				mark=at position #1 with {\arrowreversed{latex}}},postaction={decorate}}}

\usetikzlibrary{shapes.misc}\tikzset{cross/.style={cross out, draw, 
         minimum size=2*(#1-\pgflinewidth), 
         inner sep=0pt, outer sep=0pt}}

\usepackage{pgfplots}
	
\allowdisplaybreaks	

\numberwithin{equation}{section}

\def\bigO{{\cal O}}

\usepackage[colorlinks=true]{hyperref}
\hypersetup{urlcolor=blue, citecolor=red, linkcolor=blue}

\begin{document}
\title{Large gap asymptotics on annuli \\ in the random normal matrix model}
\author{Christophe Charlier}

\maketitle

\begin{abstract}
We consider a two-dimensional determinantal point process arising in the random normal matrix model and which is a two-parameter generalization of the complex Ginibre point process. In this paper, we prove that the probability that no points lie on any number of annuli centered at $0$ satisfies large $n$ asymptotics of the form
\begin{align*}
\exp \bigg( C_{1} n^{2} + C_{2} n \log n + C_{3} n +  C_{4} \sqrt{n} + C_{5}\log n + C_{6} + \mathcal{F}_{n} + \bigO\big( n^{-\frac{1}{12}}\big)\bigg),
\end{align*}
where $n$ is the number of points of the process. We determine the constants $C_{1},\ldots,C_{6}$ explicitly, as well as the oscillatory term $\mathcal{F}_{n}$ which is of order $1$. We also allow one annulus to be a disk, and one annulus to be unbounded. For the complex Ginibre point process, we improve on the best known results: (i) when the hole region is a disk, only $C_{1},\ldots,C_{4}$ were previously known, (ii) when the hole region is an unbounded annulus, only $C_{1},C_{2},C_{3}$ were previously known, and (iii) when the hole region is a regular annulus in the bulk, only $C_{1}$ was previously known. For general values of our parameters, even $C_{1}$ is new. A main discovery of this work is that $\mathcal{F}_{n}$ is given in terms of the Jacobi theta function. As far as we know this is the first time this function appears in a large gap problem of a two-dimensional point process. 


\end{abstract}



\section{Introduction and statement of results}\label{section: introduction}
Consider the probability density function
\begin{align}\label{def of point process}
\frac{1}{n!Z_{n}} \prod_{1 \leq j < k \leq n} |z_{k} -z_{j}|^{2} \prod_{j=1}^{n}|z_{j}|^{2\alpha}e^{-n |z_{j}|^{2b}}, \qquad  b>0, \; \alpha > -1,
\end{align}
where $z_{1},\ldots,z_{n} \in \mathbb{C}$ and $Z_{n}$ is the normalization constant. We are interested in the gap probability
\begin{align}\label{def of Pn}
\mathcal{P}_{n} & := \mathbb{P}\bigg( \#\Big\{z_{j}:|z_{j}|\in [r_{1},r_{2}]\cup [r_{3},r_{4}]\cup ... \cup [r_{2g-1},r_{2g}] \Big\}=0 \bigg),
\end{align}
where $0 \leq r_{1} < r_{2} < \ldots < r_{2g} \leq +\infty$. Thus $\mathcal{P}_{n}$ is the probability that no points lie on $g$ annuli centered at $0$ and whose radii are given by $r_{1},\ldots,r_{2g}$. One annulus is a disk if $r_{1}=0$, and one annulus is unbounded if $r_{2g}=+\infty$. In this paper we obtain the large $n$ asymptotics of $\mathcal{P}_{n}$, up to and including the term of order 1. 

\newpage The particular case $b=1$ and $\alpha=0$ of \eqref{def of point process} is known as the complex Ginibre point process \cite{Ginibre} (or simply \textit{Ginibre process}, for short) and is the most well-studied two-dimensional determinantal point process of the theory of random matrices. It describes the distribution of the eigenvalues of an $n \times n$ random matrix whose entries are independent complex centered Gaussian random variables with variance $\frac{1}{n}$. For general values of $b>0$ and $\alpha>-1$, \eqref{def of point process} is the joint eigenvalue density of a normal matrix $M$ taken with respect to the probability measure \cite{Mehta}
\begin{align}\label{def of random normal matrix model}
\frac{1}{\mathcal{Z}_{n}}|\det(M)|^{2\alpha}e^{-n \, \mathrm{tr}((MM^{*})^{b})}dM.
\end{align}
Here $dM$ denotes the measure induced by the flat Euclidian metric of $\mathbb{C}^{n\times n}$ on the set of normal $n \times n$ matrices (see e.g. \cite{CZ1998,ElbauFedler} for details), $M^{*}$ is the conjugate transpose of $M$, ``$\mbox{tr}$" denotes the trace, and $\mathcal{Z}_{n}$ is the normalization constant.

\medskip The limiting mean density (with respect to $d^{2}z$) as $n \to + \infty$ of the points $z_{1},\ldots,z_{n}$ is given by \cite{SaTo, Charlier 2d jumps}
\begin{align}\label{limiting density}
\frac{b^{2}}{\pi}|z|^{2b-2},
\end{align}
and is supported on the disk centered at $0$ and of radius $b^{-\frac{1}{2b}}$. In particular, for $b=1$, the limiting density is uniform over the unit disk; this is a well-known result of Ginibre \cite{Ginibre}. Since the points accumulate on a compact set as $n \to + \infty$, this means that for large $n$, $\mathcal{P}_{n}$ is the probability of a rare event, namely that there are $g$ ``large gaps/holes" in the form of annuli. 

\medskip The probability to observe a hole on a disk centered at $0$ and of radius $r<1$ in the Ginibre process was first studied by Forrester, who obtained \cite[eq (27)]{ForresterHoleProba}
\begin{align}\label{JLM result for finite Ginibre}
\mathcal{P}_{n} = \exp \bigg( C_{1}n^{2} + C_{2}n \log n + C_{3} n + C_{4}\sqrt{n} + o(\sqrt{n}) \bigg), \qquad \mbox{as } n \to + \infty,
\end{align}
where
\begin{align*}
& C_{1} = -\frac{r^{4}}{4}, \qquad C_{2} = -\frac{r^{2}}{2}, \qquad C_{3} = r^{2}\Big( 1-\log(r\sqrt{2\pi}) \Big), \\
& C_{4} = \sqrt{2} \, r \bigg\{ \int_{-\infty}^{0}\log \bigg( \frac{1}{2}\mathrm{erfc}(y) \bigg)dy + \int_{0}^{+\infty} \bigg[\log \bigg( \frac{1}{2}\mathrm{erfc}(y) \bigg) +y^{2} +\log y + \log(2\sqrt{\pi})\bigg]dy \bigg\},
\end{align*}
and $\mathrm{erfc}$ is the complementary error function defined by
\begin{align}\label{def of erfc}
\mathrm{erfc} (z) = \frac{2}{\sqrt{\pi}}\int_{z}^{\infty} e^{-t^{2}}dt.
\end{align}
The constant $C_{1}$ was also given independently by Jancovici, Lebowitz and Manificat in \cite[eq (2.7)]{JLM1993}. As noticed in \cite[eq (13)]{ForresterHoleProba}, $C_{1}$ and $C_{2}$ also follow from the asymptotic expansion obtained in an equivalent problem considered in the earlier work \cite{GHS1988}. The constants $C_{1},C_{2},C_{3}$ have also been obtained in the more recent work \cite{APS2009} using a different method; see also \cite[eq (49)]{L et al 2019} for another derivation of $C_{1}$. Although Forrester's result \eqref{JLM result for finite Ginibre} is $30$ years old, as far as we know it is the most precise result available in the literature prior to this work. 

\medskip When the hole region is an unbounded annulus centered at $0$ and of inner radius $r<1$, the following third order asymptotics for $\mathcal{P}_{n}$ were obtained by Cunden, Mezzadri and Vivo in \cite[eq (51)]{CMV2016} for the Ginibre process:
\begin{align}\label{CMV result for unbounded annulus}
\mathcal{P}_{n} = \exp \bigg( C_{1}n^{2} + C_{2}n \log n + C_{3} n + o(n) \bigg), \qquad \mbox{as } n \to + \infty,
\end{align}
where $C_{1} = \frac{r^{4}}{4}-r^{2}+\frac{3}{4}+\log r$, $ C_{2} = \frac{r^{2}-1}{2}$, $C_{3} = (1-r^{2}) \big( 1 - \log \frac{\sqrt{2\pi}(1-r^{2})}{r} \big)$.

\medskip Hole probabilities of more general domains have been considered in \cite{AR Infinite Ginibre} for the Ginibre process. In particular, for a large class of open sets $U$ lying in the unit disk, Adhikari and Reddy in \cite{Adhikari} proved that
\begin{align*}
\mathbb{P}\bigg( \#\Big\{z_{j}:z_{j}\in U \Big\}=0 \bigg) = \exp \bigg( C_{1} n^{2}+ o(n^{2}) \bigg), \qquad \mbox{as } n \to + \infty,
\end{align*}
where the constant $C_{1}=C_{1}(U)$ is given in terms of a certain equilibrium measure related to a problem of potential theory. When $U$ is either a disk, an annulus, an ellipse, a cardioid, an equilateral triangle or a half-disk, $C_{1}$ has been computed explicitly. Some of these results have then been generalized for a wide class of point processes by Adhikari in \cite{Adhikari}. For the point process \eqref{def of point process} (with arbitrary $b>0$ but $\alpha=0$), he obtained
\begin{align}
& \mathbb{P}\bigg( \#\Big\{z_{j}:|z_{j}|\in [0,r] \Big\}=0 \bigg) = \exp \bigg( - \frac{br^{4b}}{4} n^{2}+ o(n^{2}) \bigg), \nonumber \\
& \mathbb{P}\bigg( \#\Big\{z_{j}:|z_{j}|\in [r_{1},r_{2}] \Big\}=0 \bigg) = \exp \bigg( -\bigg( \frac{b}{4}(r_{2}^{4b}-r_{1}^{4b}) - \frac{(r_{2}^{2b}-r_{1}^{2b})^{2}}{4 \log(\frac{r_{2}}{r_{1}})} \bigg) n^{2}+ o(n^{2}) \bigg), \label{one annuli result of AR}
\end{align}
as $n \to + \infty$ with $0<r_{1}<r_{2}<b^{-\frac{1}{2b}}$ and $r \in (0,b^{-\frac{1}{2b}})$ fixed, see \cite[Theorem 1.2 and eqs (3.5)--(3.6)]{Adhikari}. 

\medskip These are the only works which we are aware of and which fall exactly in our setting. There are however several other works that fall just outside. In \cite{Shirai}, Shirai considered the infinite Ginibre process, which, as its name suggests, is the limiting point process arising in the bulk of the (finite) Ginibre process. He proved, among other things, that the probability of the hole event $\#\big\{z_{j}:|z_{j}| \leq r \big\}=0$ behaves as $\exp \big( \frac{-r^{4}}{4} + o(r^{4}) \big)$ as $r \to + \infty$ (see also \cite[Proposition 7.2.1]{HKPV2009} for a different proof). This result can be seen as a less precise analogue of \eqref{JLM result for finite Ginibre} for the infinite Ginibre process, and was later generalized for more general shapes of holes in \cite{AR Infinite Ginibre} and then for more general point processes in \cite{Adhikari}. Hole probabilities for product random matrices have been investigated in \cite{AS2013, AIS2014}. The existing literature on large gap problems in dimension $\geq 2$ goes beyond random matrix theory. The random zeros of the hyperbolic analytic function $\sum_{k=0}^{+\infty}\xi_{k} z^{k}$ --- here the $\xi_{k}$'s are independent standard complex Gaussians --- form a determinantal point process \cite{PV2005}, and the associated large gap problem on a centered disk has been solved in \cite[Corollary 3 (i)]{PV2005}. Another well studied two-dimensional point process is the random zeros of the standard Gaussian entire function. This function is given by $\sum_{k=0}^{+\infty}\xi_{k} \frac{z^{k}}{\sqrt{k!}}$, where the $\xi_{k}$'s are independent standard complex Gaussians. In \cite{ST2003}, the probability for this function to have no zeros in a centered disk of radius $r$ was shown to be, for all sufficiently large $r$, bounded from below by $\exp(-Cr^{4})$ and bounded from above by $\exp(-cr^{4})$ for some positive constants $c$ and $C$. This result was later improved by Nishry in \cite{Nishry}, who proved that this probability is $\exp(-\frac{e^{2}}{4}r^{4}+o(r^{4}))$ as $r \to + \infty$. A similar result as in \cite{ST2003} was obtained in \cite{Hough} for a different kind of random functions with diffusing coefficients. Also, for a $d$-dimensional process of noninteracting fermions, it is shown in \cite{GDS2021} that the hole probability on a spherical domain of radius $r$ behaves as $\exp(-cr^{d+1}+o(r^{3}))$ as $r \to + \infty$, and an explicit expression for $c>0$ is also given. 

\medskip In its full generality, the random normal matrix model is associated with a given confining potential $Q:\mathbb{C}\to \mathbb{R}\cup\{+\infty\}$ and is defined by a probability measure proportional to $e^{-n\mathrm{tr}Q(M)}dM$, where $dM$ is as in \eqref{def of random normal matrix model}. The random normal matrix model has been extensively studied over the years, see e.g. \cite{CZ1998, ElbauFedler} for early works, \cite{Forrester, RiderVirag, AHM2011, LebleSerfaty} for smooth linear statistics, \cite{Rider, AKS2018, WebbWong, DeanoSimm, FenzlLambert, Charlier 2d jumps} for non-smooth linear statistics, and \cite{BBLM2015, HedWenn2017, LeeYang, LeeYang3, AC2021} for recent investigations on planar orthogonal polynomials. Despite such progress, the problem of determining large gap asymptotics in this model has remained an outstanding problem. In this work we focus on $Q(z) = |z|^{2b}+\frac{2\alpha}{n} \log |z|$, which is a generalization of the Gaussian potential $|z|^{2}$ known as the Mittag–Leffler potential \cite{AK2013}.

\medskip Let us now explain our results in more detail. We obtain the large $n$ asymptotics of $\mathcal{P}_{n}$ for general values of $b>0$ and $\alpha > -1$ in four different cases:
\begin{enumerate}
\item \vspace{-0.0cm} The case $0 < r_{1} < \ldots < r_{2g} < b^{-\frac{1}{2b}}$ \hspace{3cm} is stated in Theorem \ref{thm:main thm},
\item \vspace{-0.2cm} The case $0 < r_{1} < \ldots < r_{2g-1} < b^{-\frac{1}{2b}} < r_{2g}=+\infty$ \hspace{0.64cm} is stated in Theorem \ref{thm:main thm 2},
\item \vspace{-0.2cm} The case $0 = r_{1} < r_{2} < \ldots < r_{2g} < b^{-\frac{1}{2b}}$ \hspace{2.26cm} is stated in Theorem \ref{thm:main thm 3},
\item \vspace{-0.2cm} The case $0 = r_{1} < r_{2} < \ldots < r_{2g-1} < b^{-\frac{1}{2b}} < r_{2g}=+\infty$ is stated in Theorem \ref{thm:main thm 4}.
\end{enumerate}
In other words, we cover the situations where the hole region consists of
\begin{enumerate}
\item \vspace{-0.1cm} $g$ annuli inside the disk of radius $b^{-\frac{1}{2b}}$ (``the bulk"), 
\item \vspace{-0.2cm} $g-1$ annuli in the bulk and one unbounded annulus ($g \geq 1$),
\item \vspace{-0.2cm} $g-1$ annuli in the bulk and one disk ($g \geq 1$),
\item \vspace{-0.2cm} $g-2$ annuli in the bulk, one unbounded annulus, and one disk ($g \geq 2$).
\end{enumerate}
For each of these four cases, we prove that
\begin{align}\label{asymp of Pn in intro}
\mathcal{P}_{n} = \exp \bigg( C_{1} n^{2} + C_{2} n \log n + C_{3} n +  C_{4} \sqrt{n} + C_{5}\log n + C_{6} + \mathcal{F}_{n} + \bigO\big( n^{-\frac{1}{12}}\big)\bigg),
\end{align}
as $n \to + \infty$, and we give explicit expressions for the constants $C_{1},\ldots,C_{6}$. 

\medskip The quantity $\mathcal{F}_{n}$ fluctuates around $0$ as $n$ increases, is of order $1$, and is given in terms of the Jacobi theta function (see e.g. \cite[Chapter 20]{NIST})
\begin{align}\label{def of Jacobi theta}
\theta(z|\tau) := \sum_{\ell=-\infty}^{+\infty} e^{2 \pi i \ell z}e^{\pi i \ell^{2} \tau}, \qquad z \in \mathbb{C}, \quad \tau \in i(0,+\infty).
\end{align}
Note that $\theta(z|\tau)=\theta(z+1|\tau)$ for all $z \in \mathbb{C}$ and $\tau \in i(0,+\infty)$; in particular $\mathbb{R}\ni x\mapsto \theta(x|\tau)$ is periodic of period $1$. To our knowledge, this is the first time the Jacobi theta function appears in a large gap problem of a two-dimensional point process. 

\medskip The presence of oscillations in these asymptotics can be explained by the following heuristics. It is easy to see (using Bayes' formula) that $\mathcal{P}_{n}$ is also equal to the partition function (= normalization constant) of the point process \eqref{def of point process} \textit{conditioned} on the event that $\#\{z_{j}:|z_{j}|\in [r_{1},r_{2}]\cup [r_{3},r_{4}]\cup ... \cup [r_{2g-1},r_{2g}] \}=0$. As is typically the case in the asymptotic analysis of partition functions, an important role is played by the $n$-tuples $(z_{1},\ldots,z_{n})$ which maximize the density of this conditional process. One is then left to understand the configurations of such $n$-tuples when $n$ is large. To be more concrete, suppose for example that $g=1$ and $0<r_{1}<r_{2}<b^{-\frac{1}{2b}}$. Since the support of \eqref{limiting density} is the centered disk of radius $b^{-\frac{1}{2b}}$, it is natural to expect that the points in the conditional process will accumulate as $n \to + \infty$ on two separated components (namely the centered disk of radius $r_{1}$, and an annulus whose small radius is $r_{2}$). The $n$-tuples $(z_{1},\ldots,z_{n})$ maximizing the conditional density may differ from each other by the number of $z_{j}$'s lying on a given component. This, in turn, produces some oscillations in the behavior of $\mathcal{P}_{n}$. More generally, if the points in the conditional process accumulate on several components (``the multi-component regime"), then one expects some oscillations in the asymptotics of $\mathcal{P}_{n}$. (There exist several interesting studies on conditional processes in dimension two, see e.g. \cite{GN2019, NishryWennman, Seo}.) In the setting of this paper, there are three cases for which there is no oscillation (i.e. $\mathcal{F}_{n}=0$): when the hole region consists of only one disk (the case $g=1$ of Theorem \ref{thm:main thm 2}), only one unbounded annulus (the case $g=1$ of Theorem \ref{thm:main thm 3}), or one disk and one unbounded annulus (the case $g=2$ of Theorem \ref{thm:main thm 4}). This is consistent with our above discussion since in each of these three cases the points of the conditional process will accumulate on a single connected component.

\medskip It has already been observed that the Jacobi theta function (and more generally, the Riemann theta function) typically describes the oscillations of various statistics of one-dimensional point processes in ``the multi-cut regime". Indeed, Widom in \cite{Widom1995} discovered that the large gap asymptotics of the one-dimensional sine process, when the gaps consist of \textit{several} intervals, contain oscillations of order 1 given in terms of the solution to a Jacobi inversion problem. These oscillations were then substantially simplified by Deift, Its and Zhou in \cite{DIZ1997}, who expressed them in terms of the Riemann theta function. Since then, there has been other works of this vein, see \cite{BG2} for $\beta$-ensembles, \cite{ClaeysGravaMcLaughlin} for partition functions of random matrix models, \cite{FK2020} for the sine process, \cite{BCL1, BCL2, KM2021} for the Airy process, and \cite{BCL3} for the Bessel process. In all these works, the Riemann theta function describes the fluctuations in the asymptotics, thereby suggesting that this function is a universal object related to the multi-cut regime of one-dimensional point processes. Our results show that, perhaps surprisingly, the universality of the Jacobi theta function goes beyond dimension $1$. 

\medskip Another function that plays a predominant role in the description of the large $n$ asymptotics of $\mathcal{P}_{n}$ is the complementary error function (defined in \eqref{def of erfc}). This function already emerged in the constant $C_{4}|_{(b=1,\alpha=0)}$ of Forrester, see \eqref{JLM result for finite Ginibre}. Interestingly, the constant $C_{4}$ of Theorem \ref{thm:main thm} involves the same integrals (which are independent of $b$ and $\alpha$), namely
\begin{align}\label{int1 intro}
\int_{-\infty}^{0}\log \bigg( \frac{1}{2}\mathrm{erfc}(y) \bigg)dy, \qquad \int_{0}^{+\infty} \bigg[\log \bigg( \frac{1}{2}\mathrm{erfc}(y) \bigg) +y^{2} +\log y + \log(2\sqrt{\pi})\bigg]dy,
\end{align}
and the constants $C_{6}$ of Theorems \ref{thm:main thm 2} and \ref{thm:main thm 3} involve
\begin{align}
& \int_{-\infty}^{0} \bigg\{ 2y\log \bigg( \frac{1}{2}\mathrm{erfc}(y)\bigg) + \frac{e^{-y^{2}}(1-5y^{2})}{3\sqrt{\pi}\mathrm{erfc}(y)} \bigg\}dy, \label{int2 intro easy} \\
& \int_{0}^{+\infty} \bigg\{ 2y\log \bigg( \frac{1}{2}\mathrm{erfc}(y)\bigg) + \frac{e^{-y^{2}}(1-5y^{2})}{3\sqrt{\pi}\mathrm{erfc}(y)} + \frac{11}{3}y^{3} + 2y \log y + \bigg( \frac{1}{2} + 2 \log(2\sqrt{\pi}) \bigg)y \bigg\}dy. \label{int2 intro}
\end{align}
Using the well-known large $y$ asymptotics of $\mathrm{erfc}(y)$ \cite[7.12.1]{NIST}
\begin{align}\label{large y asymp of erfc}
& \mathrm{erfc}(y) = \frac{e^{-y^{2}}}{\sqrt{\pi}}\bigg( \frac{1}{y}-\frac{1}{2y^{3}}+\frac{3}{4y^{5}}-\frac{15}{8y^{7}} + \bigO(y^{-9}) \bigg), & & \mbox{as } y \to + \infty,
\end{align} 
and $\mathrm{erfc}(-y) = 2-\mathrm{erfc}(y)$, it is easy to check that
\begin{align*}
\log \bigg( \frac{1}{2}\mathrm{erfc}(y) \bigg) = \bigO(e^{-\frac{y^{2}}{2}}),  \quad 2y\log \bigg( \frac{1}{2}\mathrm{erfc}(y)\bigg) + \frac{e^{-y^{2}}(1-5y^{2})}{3\sqrt{\pi}\mathrm{erfc}(y)} = \bigO(e^{-\frac{y^{2}}{2}}), \qquad \mbox{as } y \to - \infty,
\end{align*}
and that
\begin{align*}
& \log \bigg( \frac{1}{2}\mathrm{erfc}(y) \bigg) +y^{2} +\log y + \log(2\sqrt{\pi}) = \bigO(y^{-2}), \\
& 2y\log \bigg( \frac{1}{2}\mathrm{erfc}(y)\bigg) + \frac{e^{-y^{2}}(1-5y^{2})}{3\sqrt{\pi}\mathrm{erfc}(y)} + \frac{11}{3}y^{3} + 2y \log y + \bigg( \frac{1}{2} + 2 \log(2\sqrt{\pi}) \bigg)y = \bigO(y^{-3}), 
\end{align*}
as $y \to + \infty$, which implies that the integrals in \eqref{int1 intro}, \eqref{int2 intro easy} and \eqref{int2 intro} are finite, as it must.

\medskip We expect that the estimate $\bigO\big(n^{-\frac{1}{12}}\big)$ for the error term in \eqref{asymp of Pn in intro} is not optimal and could be reduced to $\bigO\big(n^{-\frac{1}{2}}\big)$. However, proving this is a very technical task, and we will not pursue that here. We now state our main results, and discuss our method of proof afterwards.

\begin{figure}[h!]
\begin{center}
\begin{tikzpicture}
\node at (0,0) {\includegraphics[width=5cm]{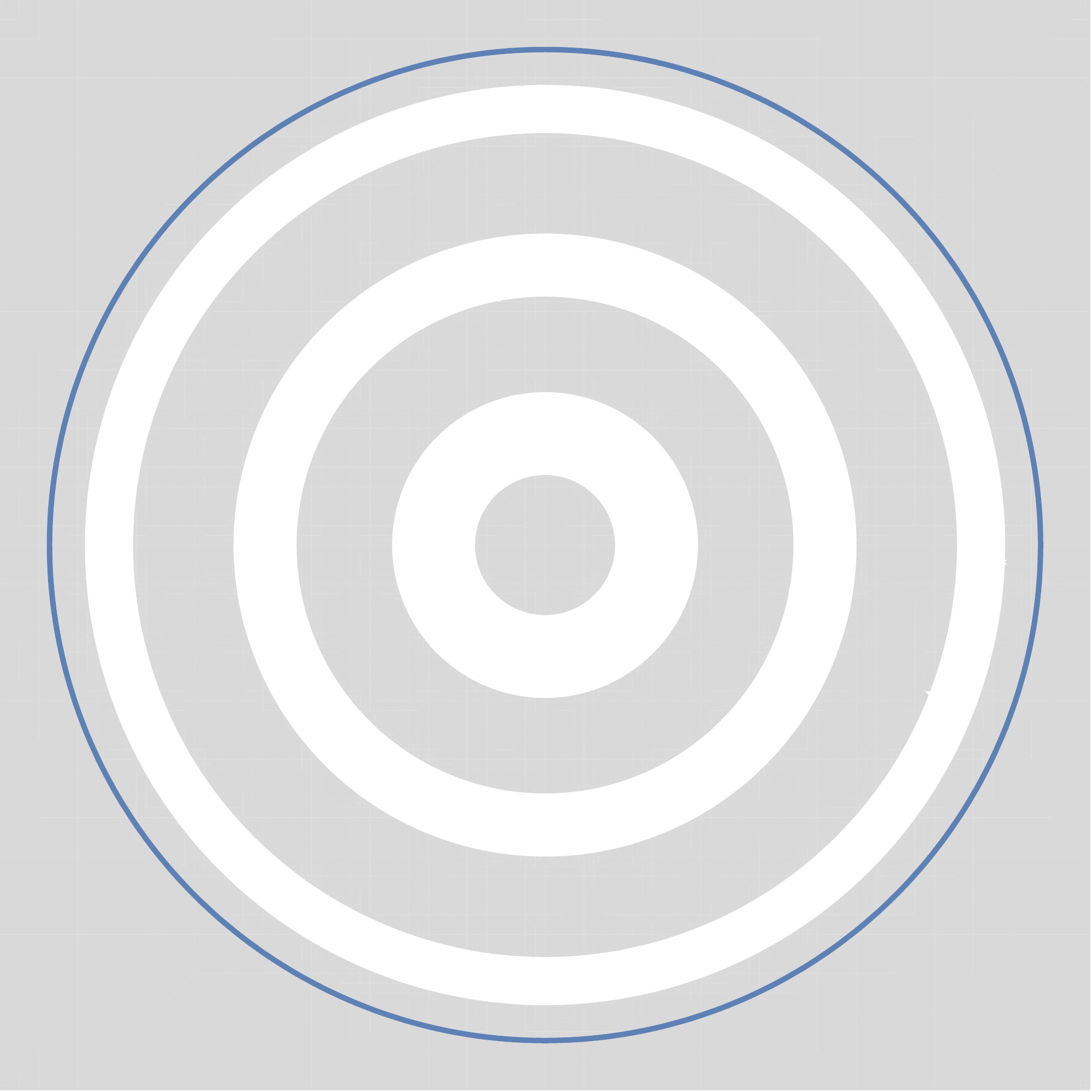}};
\node at (0,0.5) {\footnotesize gap};
\node at (0,1.27) {\footnotesize gap};
\node at (0,1.975) {\footnotesize gap};
\draw[-<-=0,->-=1] (0,0)--(2.28,0);
\node at (1.2,0.2) {\footnotesize $b^{-\frac{1}{2b}}$};
\end{tikzpicture}
\end{center}
\vspace{-0.4cm}\caption{This situation is covered by Theorem \ref{thm:main thm} with $g=3$. 
}
\end{figure}
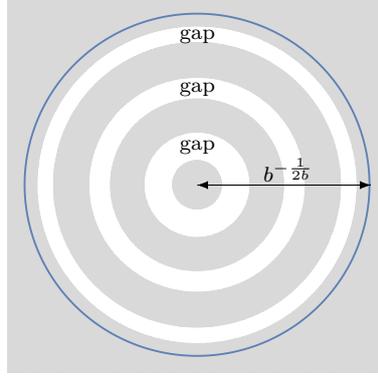
\begin{theorem}\label{thm:main thm}
($g$ annuli in the bulk) 

Let 
\begin{align*}
g \in \{1,2,\ldots\}, \qquad \alpha > -1, \qquad b>0, \qquad 0 < r_{1} < \ldots < r_{2g} < b^{-\frac{1}{2b}}
\end{align*}
be fixed parameters. As $n \to + \infty$, we have
\begin{align}\label{asymp in main thm}
\mathcal{P}_{n} = \exp \bigg( C_{1} n^{2} + C_{2} n \log n + C_{3} n +  C_{4} \sqrt{n} + C_{5}\log n + C_{6} + \mathcal{F}_{n} + \bigO\big( n^{-\frac{1}{12}}\big)\bigg), 
\end{align}
where
\begin{align*}
& C_{1} = \sum_{k=1}^{g} \bigg\{ \frac{(r_{2k}^{2b}-r_{2k-1}^{2b})^{2}}{4 \log(\frac{r_{2k}}{r_{2k-1}})} - \frac{b}{4}(r_{2k}^{4b}-r_{2k-1}^{4b}) \bigg\}, \\
& C_{2} = - \sum_{k=1}^{g} \frac{b(r_{2k}^{2b}-r_{2k-1}^{2b})}{2}, \\
& C_{3} = \sum_{k=1}^{g} \bigg\{ b(r_{2k}^{2b}-r_{2k-1}^{2b}) \bigg( \frac{1}{2}+\log \frac{b}{\sqrt{2\pi}} \bigg) + b^{2} \Big( r_{2k}^{2b}\log (r_{2k}) - r_{2k-1}^{2b}\log (r_{2k-1}) \Big) \\
& \hspace{1.65cm} -(t_{2k}-br_{2k-1}^{2b})\log(t_{2k}-br_{2k-1}^{2b})-(br_{2k}^{2b}-t_{2k})\log(br_{2k}^{2b}-t_{2k}) \bigg\}, \\
& C_{4} = \sqrt{2}b \bigg\{ \int_{-\infty}^{0}\log \bigg( \frac{1}{2}\mathrm{erfc}(y) \bigg)dy + \int_{0}^{+\infty} \bigg[\log \bigg( \frac{1}{2}\mathrm{erfc}(y) \bigg) +y^{2} +\log y + \log(2\sqrt{\pi})\bigg]dy \bigg\} \sum_{k=1}^{2g}r_{k}^{b}, \\
& C_{5} = 0, \\
& C_{6} = \frac{g}{2}\log (\pi)+ \sum_{k=1}^{g} \bigg\{ \frac{1-2b^{2}}{12}\log \bigg(\frac{r_{2k}}{r_{2k-1}}\bigg) + \frac{b^{2}r_{2k}^{2b}}{br_{2k}^{2b}-t_{2k}}+ \frac{b^{2}r_{2k-1}^{2b}}{t_{2k}-br_{2k-1}^{2b}} - \frac{1}{2}\log \log \bigg( \frac{r_{2k}}{r_{2k-1}} \bigg) \\
& + \frac{\big[ \log \big(\frac{br_{2k}^{2b}-t_{2k}}{t_{2k}-br_{2k-1}^{2b}} \big) \big]^{2}}{4 \log \big( \frac{r_{2k}}{r_{2k-1}} \big)} - \sum_{j=1}^{+\infty} \log \bigg( 1- \bigg( \frac{r_{2k-1}}{r_{2k}} \bigg)^{2j} \bigg) \bigg\}, \\
& \mathcal{F}_{n} = \sum_{k=1}^{g} \log \theta \Bigg(t_{2k}n + \frac{1}{2} - \alpha + \frac{\log \big(\frac{br_{2k}^{2b}-t_{2k}}{t_{2k}-br_{2k-1}^{2b}} \big)}{2 \log \big( \frac{r_{2k}}{r_{2k-1}} \big)} \Bigg| \frac{\pi i}{\log(\frac{r_{2k}}{r_{2k-1}})} \Bigg),
\end{align*}
$\theta$ is given by \eqref{def of Jacobi theta}, and for $k \in \{1,\ldots,g\}$
\begin{align}
t_{2k} & := \frac{1}{2} \frac{r_{2k}^{2b}-r_{2k-1}^{2b}}{\log \big( \frac{r_{2k}}{r_{2k-1}} \big)} \in (br_{2k-1}^{2b},br_{2k}^{2b}). \label{def of t2kstar in thm}
\end{align}
\end{theorem}
\begin{remark}
By setting $\alpha=0$ and $g=1$ in Theorem \ref{thm:main thm}, we obtain $C_{1} = \frac{(r_{2}^{2b}-r_{1}^{2b})^{2}}{4 \log(\frac{r_{2}}{r_{1}})} - \frac{b}{4}(r_{2}^{4b}-r_{1}^{4b})$, which agrees with \eqref{one annuli result of AR}.
\end{remark}
\begin{remark}
The constant $C_{5}=0$ has been included in \eqref{asymp in main thm} to ease the comparison with Theorems \ref{thm:main thm 2}, \ref{thm:main thm 3} and \ref{thm:main thm 4} below.
\end{remark}

\begin{figure}[h!]
\begin{center}
\begin{tikzpicture}
\node at (0,0) {\includegraphics[width=5cm]{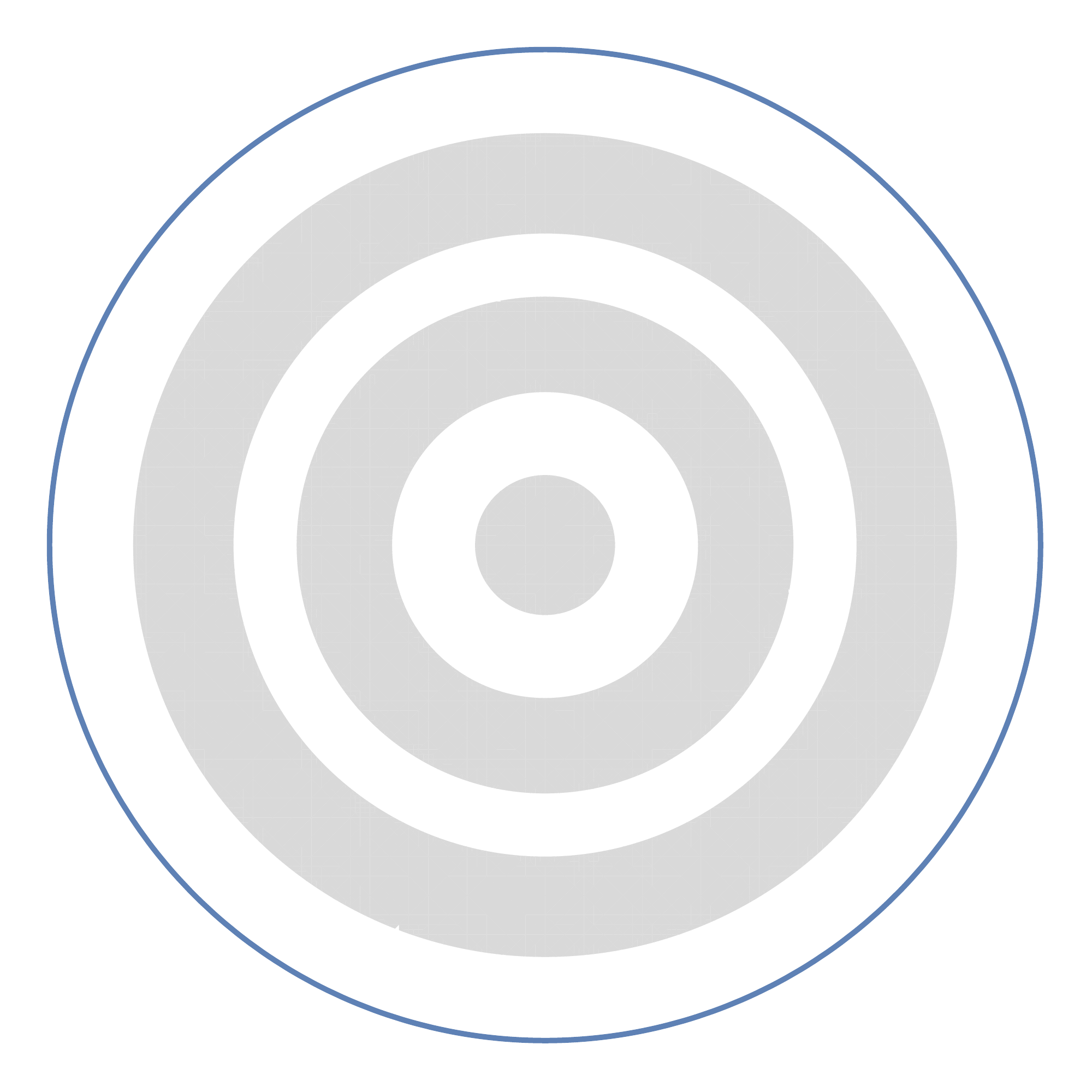}};
\node at (0,0.5) {\footnotesize gap};
\node at (0,1.27) {\footnotesize gap};
\node at (0,2.05) {\footnotesize gap};
\draw[-<-=0,->-=1] (0,0)--(2.28,0);
\node at (1.2,0.2) {\footnotesize $b^{-\frac{1}{2b}}$};
\end{tikzpicture}
\end{center}
\caption{This situation is covered by Theorem \ref{thm:main thm 2} with $g=3$.}
\end{figure}
\begin{theorem}\label{thm:main thm 2}
($g-1$ annuli in the bulk and one unbounded annulus)

Let 
\begin{align*}
g \in \{1,2,\ldots\}, \qquad \alpha > -1, \qquad b>0, \qquad 0 < r_{1} < \ldots < r_{2g-1} < b^{-\frac{1}{2b}} < r_{2g}=+\infty
\end{align*}
be fixed parameters. As $n \to + \infty$, we have
\begin{align}\label{asymp in main thm 2}
\mathcal{P}_{n} = \exp \bigg( C_{1} n^{2} + C_{2} n \log n + C_{3} n +  C_{4} \sqrt{n} + C_{5} \log n + C_{6} + \mathcal{F}_{n} + \bigO\big( n^{-\frac{1}{12}}\big)\bigg), 
\end{align}
where
\begin{align*}
& C_{1} = \sum_{k=1}^{g-1} \bigg\{ \frac{(r_{2k}^{2b}-r_{2k-1}^{2b})^{2}}{4 \log(\frac{r_{2k}}{r_{2k-1}})} - \frac{b}{4}(r_{2k}^{4b}-r_{2k-1}^{4b}) \bigg\} + \frac{br_{2g-1}^{4b}}{4}-r_{2g-1}^{2b} + \frac{1}{2b}\log \big( br_{2g-1}^{2b} \big) + \frac{3}{4b}, \\
& C_{2} = - \sum_{k=1}^{g-1} \frac{b(r_{2k}^{2b}-r_{2k-1}^{2b})}{2} + \frac{br_{2g-1}^{2b}}{2}-\frac{1}{2}, \\
& C_{3} = \sum_{k=1}^{g-1} \bigg\{ b(r_{2k}^{2b}-r_{2k-1}^{2b}) \bigg( \frac{1}{2}+\log \frac{b}{\sqrt{2\pi}} \bigg) + b^{2} \Big( r_{2k}^{2b}\log (r_{2k}) - r_{2k-1}^{2b}\log (r_{2k-1}) \Big) \\
&  -(t_{2k}-br_{2k-1}^{2b})\log(t_{2k}-br_{2k-1}^{2b})-(br_{2k}^{2b}-t_{2k})\log(br_{2k}^{2b}-t_{2k}) \bigg\} \\
&  -r_{2g-1}^{2b} \bigg( \alpha + \frac{b+1}{2}+ b \log \bigg( \frac{br_{2g-1}^{b}}{\sqrt{2\pi}} \bigg) \bigg) - (1-br_{2g-1}^{2b})\log \big( 1-br_{2g-1}^{2b} \big) + \frac{1+2\alpha}{2b}\log \big( br_{2g-1}^{2b}\big) \\
& + \frac{b+2\alpha+1}{2b}+\frac{1}{2}\log \bigg( \frac{b}{2\pi} \bigg), \\
& C_{4} = \sqrt{2}b \bigg\{ \int_{-\infty}^{0}\log \bigg( \frac{1}{2}\mathrm{erfc}(y) \bigg)dy + \int_{0}^{+\infty} \bigg[\log \bigg( \frac{1}{2}\mathrm{erfc}(y) \bigg) +y^{2} +\log y + \log(2\sqrt{\pi})\bigg]dy \bigg\} \sum_{k=1}^{2g-1}r_{k}^{b}, \\
& C_{5} = - \frac{1+2\alpha}{4}, \\
& C_{6} = \frac{g-1}{2}\log (\pi)+ \sum_{k=1}^{g-1} \bigg\{ \frac{1-2b^{2}}{12}\log \bigg(\frac{r_{2k}}{r_{2k-1}}\bigg) + \frac{b^{2}r_{2k}^{2b}}{br_{2k}^{2b}-t_{2k}}+ \frac{b^{2}r_{2k-1}^{2b}}{t_{2k}-br_{2k-1}^{2b}} - \frac{1}{2}\log \log \bigg( \frac{r_{2k}}{r_{2k-1}} \bigg) \\
& + \frac{\big[ \log \big(\frac{br_{2k}^{2b}-t_{2k}}{t_{2k}-br_{2k-1}^{2b}} \big) \big]^{2}}{4 \log \big( \frac{r_{2k}}{r_{2k-1}} \big)} - \sum_{j=1}^{+\infty} \log \bigg( 1- \bigg( \frac{r_{2k-1}}{r_{2k}} \bigg)^{2j} \bigg) \bigg\} \\
& - \frac{2\alpha +1}{4}\log(2\pi) - \frac{1+2\alpha}{2}\log(1-br_{2g-1}^{2b}) + \frac{b^{2}r_{2g-1}^{2b}}{1-br_{2g-1}^{2b}}+b + \frac{b^{2}+6b\alpha + 6\alpha^{2} + 6\alpha +3b+1}{12b}\log(b) \\
&  + \frac{b^{2}+6\alpha^{2}+6\alpha+1}{6}\log(r_{2g-1}) +2b \int_{-\infty}^{0} \bigg\{ 2y\log \bigg( \frac{1}{2}\mathrm{erfc}(y)\bigg) + \frac{e^{-y^{2}}(1-5y^{2})}{3\sqrt{\pi}\mathrm{erfc}(y)} \bigg\}dy \\
& +2b \int_{0}^{+\infty} \bigg\{ 2y\log \bigg( \frac{1}{2}\mathrm{erfc}(y)\bigg) + \frac{e^{-y^{2}}(1-5y^{2})}{3\sqrt{\pi}\mathrm{erfc}(y)} + \frac{11}{3}y^{3} + 2y \log y + \bigg( \frac{1}{2} + 2 \log(2\sqrt{\pi}) \bigg)y \bigg\}dy, \\
& \mathcal{F}_{n} = \sum_{k=1}^{g-1} \log \theta \Bigg(t_{2k}n + \frac{1}{2} - \alpha + \frac{\log \big(\frac{br_{2k}^{2b}-t_{2k}}{t_{2k}-br_{2k-1}^{2b}} \big)}{2 \log \big( \frac{r_{2k}}{r_{2k-1}} \big)} \Bigg| \frac{\pi i}{\log(\frac{r_{2k}}{r_{2k-1}})} \Bigg),
\end{align*}
$\theta$ is given by \eqref{def of Jacobi theta}, and $t_{2k}$ is given by \eqref{def of t2kstar in thm} for $k \in \{1,\ldots,g-1\}$.
\end{theorem}
\begin{remark}
It is easy to check that the constants $C_{1}, C_{2}, C_{3}$ of Theorem \ref{thm:main thm 2}, when specialized to $b=1$, $\alpha=0$ and $g=1$, are the same as the constants of Cunden, Mezzadri and Vivo in \eqref{CMV result for unbounded annulus}.
\end{remark}

The constants $C_{6}$ appearing in Theorems \ref{thm:main thm 3} and \ref{thm:main thm 4} below are notably different than in the previous two theorems, because they involve a new quantity $\mathcal{G}(b,\alpha)$ which is defined by
\begin{align}
& \mathcal{G}(b,\alpha) = \lim_{N\to + \infty} \Bigg[ \sum_{j=1}^{N} \log \Gamma\bigg( \frac{k+\alpha}{b} \bigg) - \Bigg\{ \frac{N^{2}}{2b}\log N - \frac{3+2\log b}{4b}N^{2} + \frac{1+2\alpha-b}{2b}N \log N \nonumber \\
& + \bigg( \frac{\log(2\pi)}{2} + \frac{b-2\alpha-1}{2b}(1+\log b) \bigg)N + \frac{1-3b+b^{2}+6\alpha-6b \alpha + 6\alpha^{2}}{12b} \log N \Bigg\} \Bigg], \label{lol17}
\end{align} 
where $\Gamma(z)=\int_{0}^{\infty} t^{z-1}e^{-t}dt$ is the Gamma function. Interestingly, this same object $\mathcal{G}(b,\alpha)$ also appears in the large gap asymptotics at the hard edge of the Muttalib-Borodin ensemble, see \cite[Theorem 1.1]{CLM1} ($\mathcal{G}(b,\alpha)$ here corresponds to $d(\frac{1}{b},\frac{\alpha}{b}-1)$ in \cite{CLM1}). It was also shown in \cite{CLM1} that if $b$ is a rational, then $\mathcal{G}(b,\alpha)$ can be expressed in terms of the Riemann $\zeta$-function and Barnes' $G$ function, two well-known special functions (see e.g. \cite[Chapters 5 and 25]{NIST} for the definitions of these functions). More precisely, we have the following.
\begin{proposition}(Taken from \cite[Proposition 1.4]{CLM1})
If $b = \frac{n_{1}}{n_{2}}$ for some positive integers $n_{1},n_{2}$, then $\mathcal{G}(b,\alpha)$ is explicitly given by
\begin{align*}
\mathcal{G}(b,\alpha) & = n_{1}n_{2}\zeta'(-1) + \frac{b(n_{2}-n_{1})+2n_{1}\alpha}{4b}\log(2\pi) \\
& - \frac{1-3b+b^{2}+6\alpha-6b\alpha+6\alpha^{2}}{12b}\log n_{1} - \sum_{j=1}^{n_{2}}\sum_{k=1}^{n_{1}} \log G \bigg( \frac{j+\frac{\alpha}{b}-1}{n_{2}} + \frac{k}{n_{1}} \bigg).
\end{align*}
\end{proposition}

\begin{figure}[h!]
\begin{center}
\begin{tikzpicture}
\node at (0,0) {\includegraphics[width=5cm]{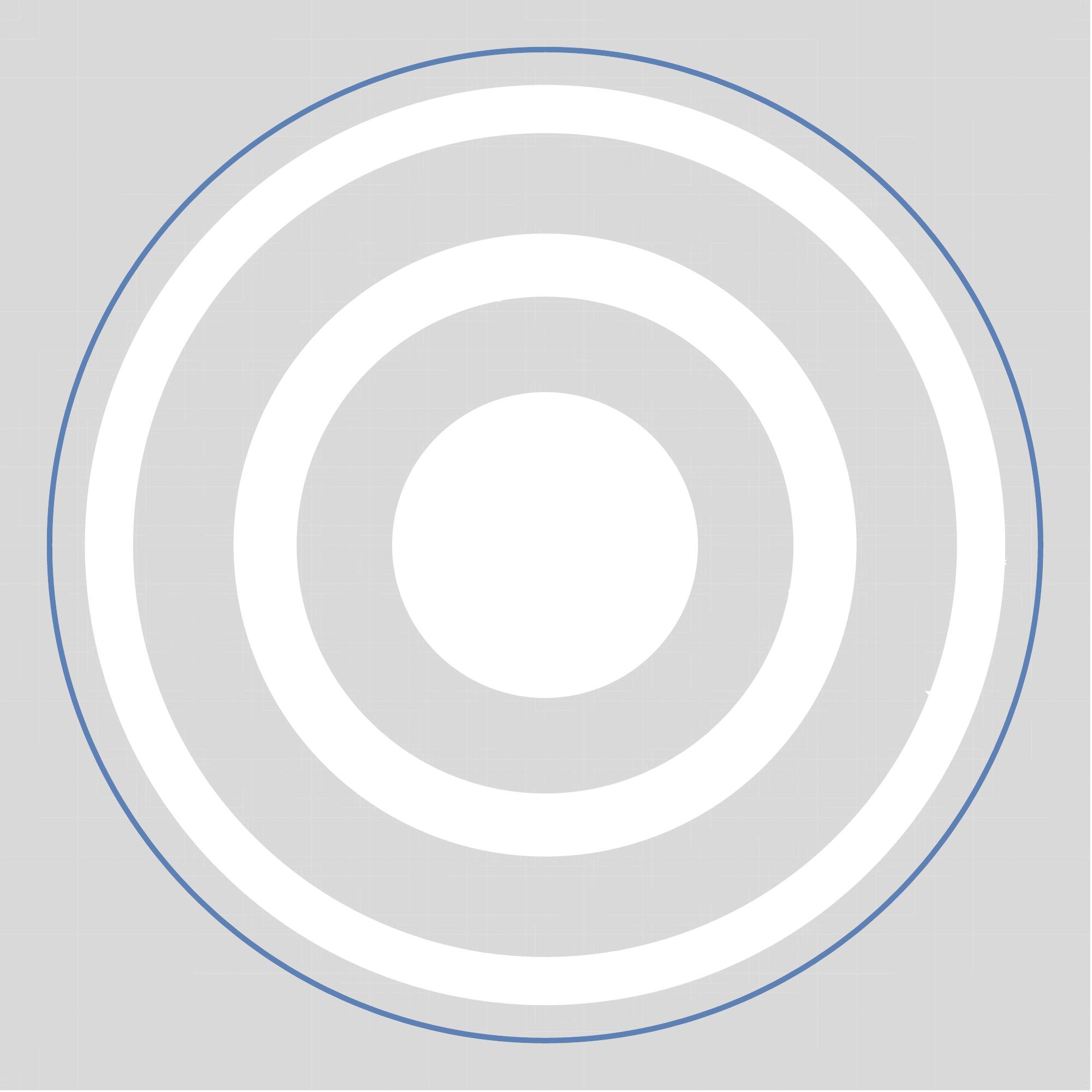}};
\node at (0,0.25) {\footnotesize gap};
\node at (0,1.27) {\footnotesize gap};
\node at (0,1.975) {\footnotesize gap};
\draw[-<-=0,->-=1] (0,0)--(2.28,0);
\node at (1.2,0.2) {\footnotesize $b^{-\frac{1}{2b}}$};
\end{tikzpicture}
\end{center}
\caption{This situation is covered by Theorem \ref{thm:main thm 3} with $g=3$.}
\end{figure}
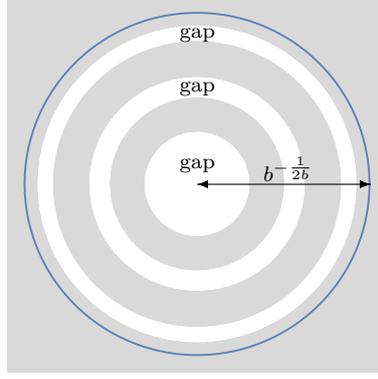
We now state our next theorem.
\begin{theorem}\label{thm:main thm 3}
($g-1$ annuli in the bulk and one disk)

Let 
\begin{align*}
g \in \{1,2,\ldots\}, \qquad \alpha > -1, \qquad b>0, \qquad 0 = r_{1} < r_{2} < \ldots < r_{2g} < b^{-\frac{1}{2b}} 
\end{align*}
be fixed parameters. As $n \to + \infty$, we have
\begin{align}\label{asymp in main thm 3}
\mathcal{P}_{n} = \exp \bigg( C_{1} n^{2} + C_{2} n \log n + C_{3} n +  C_{4} \sqrt{n} + C_{5} \log n + C_{6} + \mathcal{F}_{n} + \bigO\big( n^{-\frac{1}{12}}\big)\bigg), 
\end{align}
where
\begin{align*}
& C_{1} = \sum_{k=2}^{g} \bigg\{ \frac{(r_{2k}^{2b}-r_{2k-1}^{2b})^{2}}{4 \log(\frac{r_{2k}}{r_{2k-1}})} - \frac{b}{4}(r_{2k}^{4b}-r_{2k-1}^{4b}) \bigg\} - \frac{br_{2}^{4b}}{4}, \\
& C_{2} = - \sum_{k=2}^{g} \frac{b(r_{2k}^{2b}-r_{2k-1}^{2b})}{2} -\frac{br_{2}^{2b}}{2}, \\
& C_{3} = \sum_{k=2}^{g} \bigg\{ b(r_{2k}^{2b}-r_{2k-1}^{2b}) \bigg( \frac{1}{2}+\log \frac{b}{\sqrt{2\pi}} \bigg) + b^{2} \Big( r_{2k}^{2b}\log (r_{2k}) - r_{2k-1}^{2b}\log (r_{2k-1}) \Big) \\
&  -(t_{2k}-br_{2k-1}^{2b})\log(t_{2k}-br_{2k-1}^{2b})-(br_{2k}^{2b}-t_{2k})\log(br_{2k}^{2b}-t_{2k}) \bigg\} \\
&  +r_{2}^{2b}\bigg( \alpha + \frac{1}{2} + \frac{b}{2}\Big( 1-2\log \big( r_{2}^{b}\sqrt{2\pi} \big) \Big) \bigg), \\
& C_{4} = \sqrt{2}b \bigg\{ \int_{-\infty}^{0}\log \bigg( \frac{1}{2}\mathrm{erfc}(y) \bigg)dy + \int_{0}^{+\infty} \bigg[\log \bigg( \frac{1}{2}\mathrm{erfc}(y) \bigg) +y^{2} +\log y + \log(2\sqrt{\pi})\bigg]dy \bigg\} \sum_{k=2}^{2g}r_{k}^{b}, \\
& C_{5} = - \frac{1-6b + b^{2}+6\alpha + 6\alpha^{2}-12\alpha b}{12b}, \\
& C_{6} = \frac{g-1}{2}\log (\pi)+ \sum_{k=2}^{g} \bigg\{ \frac{1-2b^{2}}{12}\log \bigg(\frac{r_{2k}}{r_{2k-1}}\bigg) + \frac{b^{2}r_{2k}^{2b}}{br_{2k}^{2b}-t_{2k}}+ \frac{b^{2}r_{2k-1}^{2b}}{t_{2k}-br_{2k-1}^{2b}} - \frac{1}{2}\log \log \bigg( \frac{r_{2k}}{r_{2k-1}} \bigg) \\
& + \frac{\big[ \log \big(\frac{br_{2k}^{2b}-t_{2k}}{t_{2k}-br_{2k-1}^{2b}} \big) \big]^{2}}{4 \log \big( \frac{r_{2k}}{r_{2k-1}} \big)} - \sum_{j=1}^{+\infty} \log \bigg( 1- \bigg( \frac{r_{2k-1}}{r_{2k}} \bigg)^{2j} \bigg) \bigg\} \\
& + \frac{2\alpha +1}{4}\log(2\pi) + \bigg( b + 2\alpha b - \alpha -\alpha^{2} - \frac{1+b^{2}}{6} \bigg)\log r_{2} - \frac{b^{2}-6b\alpha + 6\alpha^{2} + 6\alpha -3b+1}{12b}\log(b) \\
& - \mathcal{G}(b,\alpha) -2b \int_{-\infty}^{0} \bigg\{ 2y\log \bigg( \frac{1}{2}\mathrm{erfc}(y)\bigg) + \frac{e^{-y^{2}}(1-5y^{2})}{3\sqrt{\pi}\mathrm{erfc}(y)} \bigg\}dy \\
& -2b \int_{0}^{+\infty} \bigg\{ 2y\log \bigg( \frac{1}{2}\mathrm{erfc}(y)\bigg) + \frac{e^{-y^{2}}(1-5y^{2})}{3\sqrt{\pi}\mathrm{erfc}(y)} + \frac{11}{3}y^{3} + 2y \log y + \bigg( \frac{1}{2} + 2 \log(2\sqrt{\pi}) \bigg)y \bigg\}dy, \\
& \mathcal{F}_{n} = \sum_{k=2}^{g} \log \theta \Bigg(t_{2k}n + \frac{1}{2} - \alpha + \frac{\log \big(\frac{br_{2k}^{2b}-t_{2k}}{t_{2k}-br_{2k-1}^{2b}} \big)}{2 \log \big( \frac{r_{2k}}{r_{2k-1}} \big)} \Bigg| \frac{\pi i}{\log(\frac{r_{2k}}{r_{2k-1}})} \Bigg),
\end{align*}
$\theta$ is given by \eqref{def of Jacobi theta}, $t_{2k}$ is given by \eqref{def of t2kstar in thm} for $k \in \{2,\ldots,g\}$, and $\mathcal{G}(b,\alpha)$ is given by \eqref{lol17}.
\end{theorem}
\begin{remark}
It is easy to check that the constants $C_{1}, C_{2}, C_{3}, C_{4}$ of Theorem \ref{thm:main thm 3}, when specialized to $b=1$, $\alpha=0$ and $g=1$, are the same as Forrester's constants in \eqref{JLM result for finite Ginibre}.
\end{remark}
\begin{figure}[h!]
\begin{center}
\begin{tikzpicture}
\node at (0,0) {\includegraphics[width=5cm]{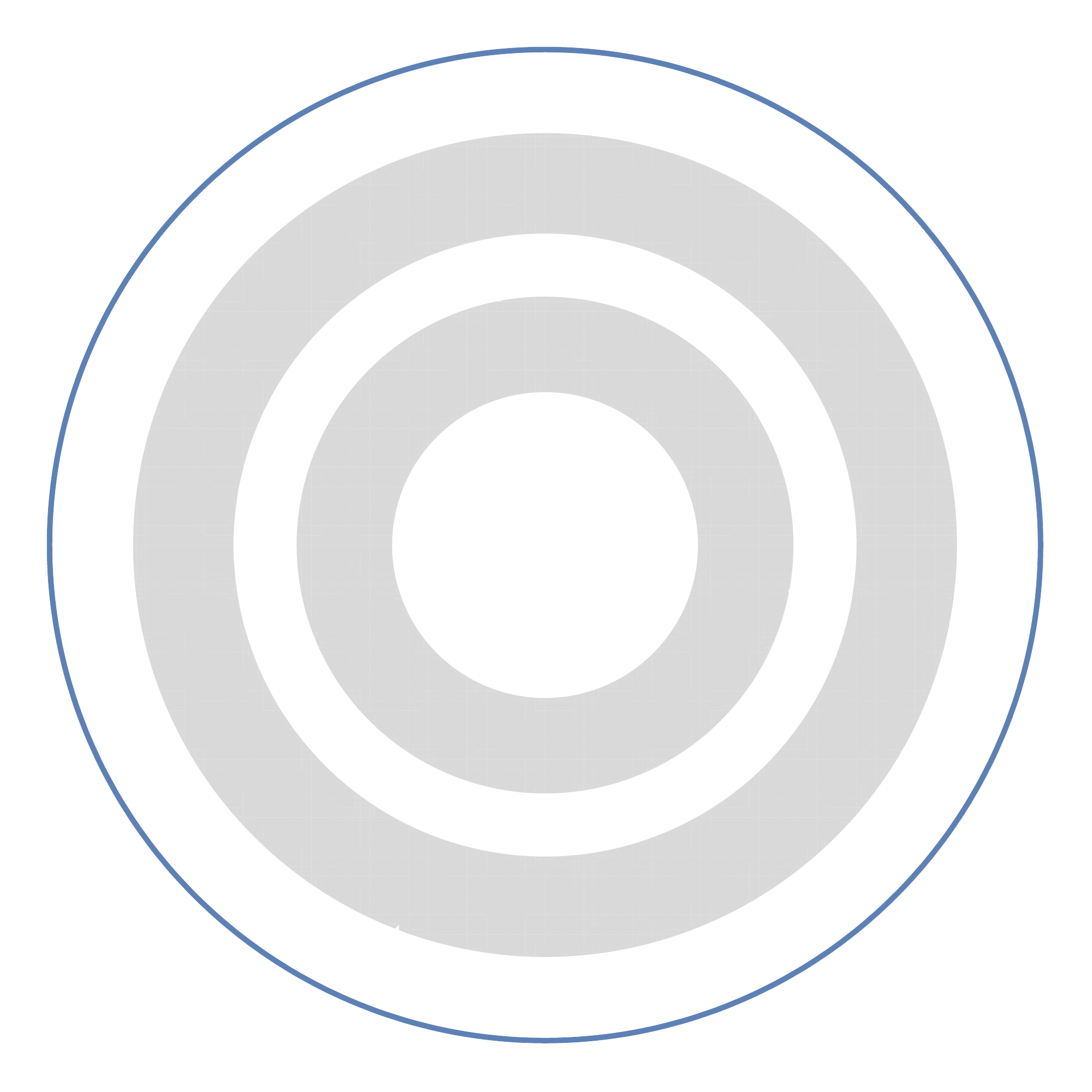}};
\node at (0,0.25) {\footnotesize gap};
\node at (0,1.27) {\footnotesize gap};
\node at (0,2.05) {\footnotesize gap};
\draw[-<-=0,->-=1] (0,0)--(2.28,0);
\node at (1.2,0.2) {\footnotesize $b^{-\frac{1}{2b}}$};
\end{tikzpicture}
\end{center}
\caption{This situation is covered by Theorem \ref{thm:main thm 4} with $g=3$.}
\end{figure}
\begin{theorem}\label{thm:main thm 4}
($g-2$ annuli in the bulk, one unbounded annulus, and one disk)

Let 
\begin{align*}
g \in \{2,3,\ldots\}, \qquad \alpha > -1, \qquad b>0, \qquad 0 = r_{1} < r_{2} < \ldots < r_{2g-1} < b^{-\frac{1}{2b}} < r_{2g}=+\infty
\end{align*}
be fixed parameters. As $n \to + \infty$, we have
\begin{align}\label{asymp in main thm 4}
\mathcal{P}_{n} = \exp \bigg( C_{1} n^{2} + C_{2} n \log n + C_{3} n +  C_{4} \sqrt{n} + C_{5} \log n + C_{6} + \mathcal{F}_{n} + \bigO\big( n^{-\frac{1}{12}}\big)\bigg), 
\end{align}
where
\begin{align*}
& C_{1} = \sum_{k=2}^{g-1} \bigg\{ \frac{(r_{2k}^{2b}-r_{2k-1}^{2b})^{2}}{4 \log(\frac{r_{2k}}{r_{2k-1}})} - \frac{b}{4}(r_{2k}^{4b}-r_{2k-1}^{4b}) \bigg\} + \frac{br_{2g-1}^{4b}}{4}-r_{2g-1}^{2b} + \frac{1}{2b}\log \big( br_{2g-1}^{2b} \big) + \frac{3}{4b} - \frac{br_{2}^{4b}}{4}, \\
& C_{2} = - \sum_{k=2}^{g-1} \frac{b(r_{2k}^{2b}-r_{2k-1}^{2b})}{2} + \frac{br_{2g-1}^{2b}}{2}-\frac{1}{2}-\frac{br_{2}^{2b}}{2}, \\
& C_{3} = \sum_{k=2}^{g-1} \bigg\{ b(r_{2k}^{2b}-r_{2k-1}^{2b}) \bigg( \frac{1}{2}+\log \frac{b}{\sqrt{2\pi}} \bigg) + b^{2} \Big( r_{2k}^{2b}\log (r_{2k}) - r_{2k-1}^{2b}\log (r_{2k-1}) \Big) \\
&  -(t_{2k}-br_{2k-1}^{2b})\log(t_{2k}-br_{2k-1}^{2b})-(br_{2k}^{2b}-t_{2k})\log(br_{2k}^{2b}-t_{2k}) \bigg\} \\
&  -r_{2g-1}^{2b} \bigg( \alpha + \frac{b+1}{2}+ b \log \bigg( \frac{br_{2g-1}^{b}}{\sqrt{2\pi}} \bigg) \bigg) - (1-br_{2g-1}^{2b})\log \big( 1-br_{2g-1}^{2b} \big) + \frac{1+2\alpha}{2b}\log \big( br_{2g-1}^{2b}\big) \\
& + \frac{b+2\alpha+1}{2b}+\frac{1}{2}\log \bigg( \frac{b}{2\pi} \bigg) +r_{2}^{2b}\bigg( \alpha + \frac{1}{2} + \frac{b}{2}\Big( 1-2\log \big( r_{2}^{b}\sqrt{2\pi} \big) \Big) \bigg), \\
& C_{4} = \sqrt{2}b \bigg\{ \int_{-\infty}^{0}\log \bigg( \frac{1}{2}\mathrm{erfc}(y) \bigg)dy + \int_{0}^{+\infty} \bigg[\log \bigg( \frac{1}{2}\mathrm{erfc}(y) \bigg) +y^{2} +\log y + \log(2\sqrt{\pi})\bigg]dy \bigg\} \sum_{k=2}^{2g-1}r_{k}^{b}, \\
& C_{5} = - \frac{1-3b + b^{2}+6\alpha + 6\alpha^{2}-6\alpha b}{12b}, \\
& C_{6} = \frac{g-2}{2}\log (\pi)+ \sum_{k=2}^{g-1} \bigg\{ \frac{1-2b^{2}}{12}\log \bigg(\frac{r_{2k}}{r_{2k-1}}\bigg) + \frac{b^{2}r_{2k}^{2b}}{br_{2k}^{2b}-t_{2k}}+ \frac{b^{2}r_{2k-1}^{2b}}{t_{2k}-br_{2k-1}^{2b}} - \frac{1}{2}\log \log \bigg( \frac{r_{2k}}{r_{2k-1}} \bigg) \\
& + \frac{\big[ \log \big(\frac{br_{2k}^{2b}-t_{2k}}{t_{2k}-br_{2k-1}^{2b}} \big) \big]^{2}}{4 \log \big( \frac{r_{2k}}{r_{2k-1}} \big)} - \sum_{j=1}^{+\infty} \log \bigg( 1- \bigg( \frac{r_{2k-1}}{r_{2k}} \bigg)^{2j} \bigg) \bigg\} - \frac{1+2\alpha}{2}\log(1-br_{2g-1}^{2b}) + \frac{b^{2}r_{2g-1}^{2b}}{1-br_{2g-1}^{2b}} \\
& +b + \frac{1+2\alpha}{2}\log(br_{2}^{2b}) + \frac{b^{2}+6\alpha^{2}+6\alpha+1}{6}\log\bigg(\frac{r_{2g-1}}{r_{2}}\bigg) - \mathcal{G}(b,\alpha), \\
& \mathcal{F}_{n} = \sum_{k=2}^{g-1} \log \theta \Bigg(t_{2k}n + \frac{1}{2} - \alpha + \frac{\log \big(\frac{br_{2k}^{2b}-t_{2k}}{t_{2k}-br_{2k-1}^{2b}} \big)}{2 \log \big( \frac{r_{2k}}{r_{2k-1}} \big)} \Bigg| \frac{\pi i}{\log(\frac{r_{2k}}{r_{2k-1}})} \Bigg),
\end{align*}
$\theta$ is given by \eqref{def of Jacobi theta}, $t_{2k}$ is given by \eqref{def of t2kstar in thm} for $k \in \{2,\ldots,g-1\}$, and $\mathcal{G}(b,\alpha)$ is given by \eqref{lol17}.
\end{theorem}

\paragraph{Method of proof.} The problem of determining large gap asymptotics is a notoriously difficult problem in random matrix theory with a long history \cite{K2009, F2014, GN2018}. There have been several methods that have proven successful to solve large gap problems of one-dimensional point processes, among which: the Deift--Zhou \cite{DeiftZhou} steepest descent method for Riemann--Hilbert problems \cite{K2003, DIKZ2007, DIK, BBD2008, DeiftKrasVasi, CGS2019, CLM1, CLM2, DXZ2020}, operator theoretical methods \cite{W1971, EhrSine, Ehr2010}, the ``loop equations" \cite{BG1, BG2, Marchal1, Marchal2}, and the Brownian carousel \cite{VV2009, VV2010, RRZ2011, DVAiry2013}. 

Our method of proof shows similarities with the method of Forrester in \cite{ForresterHoleProba}. It relies on the fact that \eqref{def of point process} is determinantal, rotation-invariant, and combines the uniform asymptotics of the incomplete gamma function with some precise Riemann sum approximations. Our method is less robust with respect to the shape of the hole region than the one of Adhikari and Reddy \cite{Adhikari, AR Infinite Ginibre}, but allows to give precise asymptotics. We also recently used this method of Riemann sum approximations in \cite{Charlier 2d jumps} to obtain precise asymptotics for the moment generating function of the disk counting statistics of \eqref{def of point process}. However, the problem considered here is more challenging and of a completely different nature than the one considered in \cite{Charlier 2d jumps}; most of the difficulties that we have to overcome here are not present in \cite{Charlier 2d jumps}. These differences will be discussed in more detail in Section \ref{section:proof}. 

\section{Preliminaries}
By definition of $Z_{n}$ and $\mathcal{P}_{n}$ (see \eqref{def of point process} and \eqref{def of Pn}), we have
\begin{align}
Z_{n} & = \frac{1}{n!} \int_{\mathbb{C}}\ldots \int_{\mathbb{C}} \prod_{1 \leq j < k \leq n} |z_{k} -z_{j}|^{2} \prod_{j=1}^{n} |z_{j}|^{2\alpha}e^{-n |z_{j}|^{2b}} d^{2}z_{j}, \label{def of Zn as n fold integral} \\
\mathcal{P}_{n} & = \frac{1}{n!Z_{n}} \int_{\mathbb{C}}\ldots \int_{\mathbb{C}} \prod_{1 \leq j < k \leq n} |z_{k} -z_{j}|^{2} \prod_{j=1}^{n} w(z_{j}) d^{2}z_{j}, \label{def of P1 as n fold integral}
\end{align}
where the weight $w$ is defined by
\begin{align*}
& w(z) = |z|^{2\alpha} e^{-n |z|^{2b}}\begin{cases}
0, & \mbox{if } |z| \in [r_{1},r_{2}]\cup [r_{3},r_{4}]\cup ... \cup [r_{2g-1},r_{2g}], \\
1, & \mbox{otherwise.}
\end{cases}
\end{align*}
We will use the following well-known formula to rewrite $Z_{n}$ and $\mathcal{P}_{n}$ in terms of one-fold integrals. 
\begin{lemma}\label{lemma: exact identity}
If $\mathsf{w}:\mathbb{C}\to [0,+\infty)$ is rotation invariant (i.e. $\mathsf{w}(z)=\mathsf{w}(|z|)$) and satisfies
\begin{align*}
\int_{\mathbb{C}} u^{j} \mathsf{w}(u)d u <+\infty, \qquad \mbox{for all } j \geq 0,
\end{align*}
then
\begin{align*}
\frac{1}{n!} \int_{\mathbb{C}}\ldots \int_{\mathbb{C}} \prod_{1 \leq j < k \leq n} |z_{k} -z_{j}|^{2} \prod_{j=1}^{n}\mathsf{w}(z_{j}) d^{2}z_{j} = (2\pi)^{n} \prod_{j=0}^{n-1} \int_{0}^{+\infty} u^{2j+1}\mathsf{w}(u)du.
\end{align*}
\end{lemma}
The proof of Lemma \ref{lemma: exact identity} is standard and we omit it, see e.g. \cite{WebbWong}, \cite[Lemma 1.9]{Charlier 2d jumps} and the references therein. The argument relies on the fact that the point process on $z_{1},\ldots,z_{n}\in \mathbb{C}$ with density proportional to $\prod_{1 \leq j < k \leq n} |z_{k} -z_{j}|^{2} \prod_{j=1}^{n}\mathsf{w}(z_{j})$ is determinantal and rotation-invariant. 

\medskip Using twice Lemma \ref{lemma: exact identity}, with $\mathsf{w}(z) = |z|^{2\alpha}e^{-n|z|^{2b}}$ and $\mathsf{w}(x)=w(x)$, we obtain
\begin{align}\label{explicit formula for Zn}
Z_{n} & = n^{-\frac{n^{2}}{2b}}n^{-\frac{1+2\alpha}{2b}n} \frac{\pi^{n}}{b^{n}} \prod_{j=1}^{n} \Gamma(\tfrac{j+\alpha}{b}), \\
Z_{n}\mathcal{P}_{n} & = (2\pi)^{n} \prod_{j=0}^{n-1} \sum_{\ell=0}^{g} \int_{r_{2\ell}}^{r_{2\ell+1}} u^{2j+1+2\alpha}e^{-n u^{2b}}du \nonumber \\
& = n^{-\frac{n^{2}}{2b}}n^{-\frac{1+2\alpha}{2b}n} \frac{\pi^{n}}{b^{n}} \prod_{j=1}^{n} \sum_{\ell=0}^{g} \bigg(\gamma(\tfrac{j+\alpha}{b},nr_{2\ell+1}^{2b})-\gamma(\tfrac{j+\alpha}{b},nr_{2\ell}^{2b}) \bigg), \label{main exact formula}
\end{align}
where $r_{0}:=0$, $r_{2g+1}:=+\infty$, we recall that $\Gamma(a)=\int_{0}^{\infty} t^{a-1}e^{-t}dt$ is the Gamma function, and $\gamma(a,z)$ is the incomplete gamma function
\begin{align*}
\gamma(a,z) = \int_{0}^{z}t^{a-1}e^{-t}dt.
\end{align*}
By combining \eqref{explicit formula for Zn} with \eqref{main exact formula}, we obtain
\begin{align}\label{exact formula for log Pn}
\log \mathcal{P}_{n} =  \sum_{j=1}^{n} \log \bigg(\sum_{\ell=1}^{2g+1} (-1)^{\ell+1}\frac{\gamma(\tfrac{j+\alpha}{b},nr_{\ell}^{2b})}{\Gamma(\tfrac{j+\alpha}{b})}\bigg).
\end{align}
This exact formula is the starting point of the proofs of our four theorems. To analyze the large $n$ behavior of $\log \mathcal{P}_{n}$, we will use the asymptotics of $\gamma(a,z)$ in various regimes of the parameters $a$ and $z$. These asymptotics are available in the literature and are summarized in the following lemmas. 
\begin{lemma}\label{lemma:various regime of gamma}(\cite[formula 8.11.2]{NIST}).
Let $a>0$ be fixed. As $z \to +\infty$,
\begin{align*}
\gamma(a,z) = \Gamma(a) + \bigO(e^{-\frac{z}{2}}).
\end{align*}
\end{lemma}
\begin{lemma}\label{lemma: uniform}(\cite[Section 11.2.4]{Temme}).
The following hold:
\begin{align}\label{Temme exact formula}
& \frac{\gamma(a,z)}{\Gamma(a)} = \frac{1}{2}\mathrm{erfc}(-\eta \sqrt{a/2}) - R_{a}(\eta), \qquad R_{a}(\eta) := \frac{e^{-\frac{1}{2}a \eta^{2}}}{2\pi i}\int_{-\infty}^{\infty}e^{-\frac{1}{2}a u^{2}}g(u)du,
\end{align}
where $\mathrm{erfc}$ is given by \eqref{def of erfc}, $g(u) := \frac{dt}{du}\frac{1}{\lambda-t}+\frac{1}{u+i\eta}$,
\begin{align}\label{lol8}
& \lambda = \frac{z}{a}, \quad \eta = (\lambda-1)\sqrt{\frac{2 (\lambda-1-\ln \lambda)}{(\lambda-1)^{2}}}, \quad  u = -i(t-1)\sqrt{\frac{2 (t-1-\ln t)}{(t-1)^{2}}},
\end{align}
and the principal branch is used for the roots. In particular, $\eta \in \mathbb{R}$ for $\lambda >0$, while $t \in \mathcal{L}:=\{\frac{\theta}{\sin \theta}e^{i\theta}:-\pi<\theta<\pi \}$ for $u\in \mathbb{R}$. Moreover, as $a \to + \infty$, uniformly for $z \in [0,\infty)$, 
\begin{align}\label{asymp of Ra}
& R_{a}(\eta) \sim \frac{e^{-\frac{1}{2}a \eta^{2}}}{\sqrt{2\pi a}}\sum_{j=0}^{\infty} \frac{c_{j}(\eta)}{a^{j}}.
\end{align}
All coefficients $c_{j}(\eta)$ are bounded functions of $\eta \in \mathbb{R}$ (i.e. bounded for $\lambda \in (0,\infty)$), and 
\begin{align}\label{def of c0 and c1}
c_{0}(\eta) = \frac{1}{\lambda-1}-\frac{1}{\eta}, \qquad c_{1}(\eta) = \frac{1}{\eta^{3}}-\frac{1}{(\lambda-1)^{3}}-\frac{1}{(\lambda-1)^{2}}-\frac{1}{12(\lambda-1)}.
\end{align}

\end{lemma}
By combining Lemma \ref{lemma: uniform} with the large $z$ asymptotics of $\mathrm{erfc}(z)$ given in \eqref{large y asymp of erfc}, we get the following.
\begin{lemma}\label{lemma: asymp of gamma for lambda bounded away from 1}
\item[(i)] Let $\delta>0$ be fixed. As $a \to +\infty$, uniformly for $\lambda \geq 1+\delta$,
\begin{align*}
\frac{\gamma(a,\lambda a)}{\Gamma(a)} = 1 + \frac{e^{-\frac{a\eta^{2}}{2}}}{\sqrt{2\pi}} \bigg( \frac{-1}{\lambda-1}\frac{1}{\sqrt{a}}+\frac{1+10\lambda+\lambda^{2}}{12(\lambda-1)^{3}} \frac{1}{a^{3/2}} + \bigO(a^{-5/2}) \bigg),
\end{align*}
where $\eta$ is as in \eqref{lol8} (in particular $e^{-\frac{a\eta^{2}}{2}} = e^{a-z}\frac{z^{a}}{a^{a}}$).  
\item[(ii)] As $a \to +\infty$, uniformly for $\lambda$ in compact subsets of $(0,1)$,
\begin{align*}
\frac{\gamma(a,\lambda a)}{\Gamma(a)} = \frac{e^{-\frac{a\eta^{2}}{2}}}{\sqrt{2\pi}} \bigg( \frac{-1}{\lambda-1}\frac{1}{\sqrt{a}}+\frac{1+10\lambda+\lambda^{2}}{12(\lambda-1)^{3}} \frac{1}{a^{3/2}} + \bigO(a^{-5/2}) \bigg),
\end{align*}
where $\eta$ is as in \eqref{lol8} (in particular $e^{-\frac{a\eta^{2}}{2}} = e^{a-z}\frac{z^{a}}{a^{a}}$).
\end{lemma}
\section{Proof of Theorem \ref{thm:main thm}: the case $r_{1}>0$ and $r_{2g}<b^{-\frac{1}{2b}}$}\label{section:proof}
In this paper, $\log$ always denotes the principal branch of the logarithm. Recall from \eqref{exact formula for log Pn} that 
\begin{align}\label{exact formula for Pn thm1}
\log \mathcal{P}_{n} =  \sum_{j=1}^{n} \log \bigg(\sum_{\ell=1}^{2g+1} (-1)^{\ell+1}\frac{\gamma(\tfrac{j+\alpha}{b},nr_{\ell}^{2b})}{\Gamma(\tfrac{j+\alpha}{b})}\bigg).
\end{align}
To analyze asymptotically as $n \to + \infty$ the sum on the right-hand side, we will split it into several smaller sums, which need to be handled in different ways. 

For $j=1,\ldots,n$ and $\ell=1,\ldots,2g$, we define
\begin{align}\label{def of aj lambdajl etajl}
a_{j} := \frac{j+\alpha}{b}, \qquad \lambda_{j,\ell} := \frac{bnr_{\ell}^{2b}}{j+\alpha}, \qquad \eta_{j,\ell} := (\lambda_{j,\ell}-1)\sqrt{\frac{2 (\lambda_{j,\ell}-1-\ln \lambda_{j,\ell})}{(\lambda_{j,\ell}-1)^{2}}}.
\end{align}
Let $M'$ be a large integer independent of $n$, and let $\epsilon > 0$ be a small constant independent of $n$. Define
\begin{align}
& j_{\ell,-}:=\lceil \tfrac{bnr_{\ell}^{2b}}{1+\epsilon} - \alpha \rceil, & & j_{\ell,+} := \lfloor  \tfrac{bnr_{\ell}^{2b}}{1-\epsilon} - \alpha \rfloor, & & \ell=1,\ldots,2g, \label{def of jk plus and minus} \\
& j_{0,-}:=1, & & j_{0,+}:=M', & & j_{2g+1,-}:=n+1, \nonumber
\end{align}
where $\lceil x \rceil$ denotes the smallest integer $\geq x$, and $\lfloor  x \rfloor$ denotes the largest integer $\leq x$.
We take $\epsilon$ sufficiently small such that
\begin{align}\label{cond on epsilon 1}
\frac{br_{\ell}^{2b}}{1-\epsilon} < \frac{br_{\ell+1}^{2b}}{1+\epsilon}, \qquad \mbox{for all } \ell \in \{1,\ldots,2g-1\}, \qquad \mbox{and} \qquad \frac{br_{2g}^{2b}}{1-\epsilon} < 1.
\end{align}
A natural quantity that will appear in our analysis is 
\begin{align}\label{def of t2k in proof}
t_{2k} & := \frac{1}{2} \frac{r_{2k}^{2b}-r_{2k-1}^{2b}}{\log \big( \frac{r_{2k}}{r_{2k-1}} \big)} = \frac{br_{2k}^{2b}-br_{2k-1}^{2b}}{\log(r_{2k}^{2b})-\log(r_{2k-1}^{2b})}, \qquad k=1,\ldots,g.
\end{align}
It is easy to check that for each $k\in \{1,\ldots,g\}$, $t_{2k}$  lies in the interval $(br_{2k-1}^{2b},br_{2k}^{2b})$. For reasons that will be apparent below, we also choose $\epsilon>0$ sufficiently small such that
\begin{align}\label{cond on epsilon 2}
\frac{br_{2k-1}^{2b}}{1-\epsilon} < t_{2k} < \frac{br_{2k}^{2b}}{1+\epsilon}, \qquad k=1,\ldots,g.
\end{align}
Using \eqref{main exact formula} and \eqref{cond on epsilon 1}, we split the $j$-sum in \eqref{exact formula for Pn thm1} into $4g+2$ sums
\begin{align}\label{log Dn as a sum of sums}
\log \mathcal{P}_{n} = S_{0} + \sum_{k=1}^{2g}(S_{2k-1}+S_{2k}) + S_{4g+1},
\end{align}
with 
\begin{align}
& S_{0} = \sum_{j=1}^{M'} \log \bigg( \sum_{\ell=1}^{2g+1} (-1)^{\ell+1}\frac{\gamma(\tfrac{j+\alpha}{b},nr_{\ell}^{2b})}{\Gamma(\tfrac{j+\alpha}{b})} \bigg), \label{def of S0}  \\
& S_{2k-1} = \sum_{j=j_{k-1,+}+1}^{j_{k,-}-1} \hspace{-0.3cm} \log \bigg( \sum_{\ell=1}^{2g+1} (-1)^{\ell+1}\frac{\gamma(\tfrac{j+\alpha}{b},nr_{\ell}^{2b})}{\Gamma(\tfrac{j+\alpha}{b})} \bigg), & & k=1,\ldots,2g+1, \label{def of S2kp1}  \\
& S_{2k} = \sum_{j=j_{k,-}}^{j_{k,+}} \log \bigg( \sum_{\ell=1}^{2g+1} (-1)^{\ell+1}\frac{\gamma(\tfrac{j+\alpha}{b},nr_{\ell}^{2b})}{\Gamma(\tfrac{j+\alpha}{b})} \bigg), & & k=1,\ldots,2g. \label{def of S2k} 
\end{align}
We first show that the sums $S_{0}$ and $S_{1},S_{5},S_{9},\ldots,S_{4g+1}$ are exponentially small as $n \to + \infty$. 
\begin{lemma}\label{lemma: S0}
There exists $c>0$ such that $S_{0} = \bigO(e^{-cn})$ as $n \to + \infty$.
\end{lemma}
\begin{proof}
Since $M'$ is fixed, by \eqref{def of S0} and Lemma \ref{lemma:various regime of gamma}, as $n \to +\infty$ we have
\begin{align*}
S_{0} & = \sum_{j=1}^{M'} \log \bigg( \sum_{\ell=1}^{2g+1} (-1)^{\ell+1} \big[1 + \bigO(e^{-\frac{1}{2}r_{\ell}^{2b}n}) \big] \bigg) = \bigO(e^{-\frac{1}{2}r_{1}^{2b}n}).
\end{align*}
\end{proof}

\begin{lemma}\label{lemma: S2km1 k odd}
Let $k \in \{1,3,5,\ldots,2g+1\}$. There exists $c>0$ such that $S_{2k-1} = \bigO(e^{-cn})$ as $n \to + \infty$.
\end{lemma}
\begin{proof}
The proof is similar to \cite[Lemma 2.2]{Charlier 2d jumps}. Let us consider first the case $k \in \{3,5,\ldots,2g+1\}$. By \eqref{def of aj lambdajl etajl} and \eqref{def of jk plus and minus}, for $j \in \{j_{k-1,+}+1,\ldots,j_{k,-}-1\}$ and $\ell\in \{1,\ldots,2g\}$ we have
\begin{align}\label{lol18}
& (1+\epsilon) \frac{r_{\ell}^{2b}}{r_{k}^{2b}+\frac{1+\epsilon}{b n}} \leq \lambda_{j,\ell} \leq (1-\epsilon) \frac{r_{\ell}^{2b}}{r_{k-1}^{2b} - \frac{1-\epsilon}{bn}}.
\end{align}
For $k=2g+1$, the left-hand side in \eqref{lol18} must be replaced by $\frac{r_{\ell}^{2b}}{b^{-1}+\frac{\alpha}{bn}}$. Since $\epsilon>0$ is fixed, $\lambda_{j,\ell}$ remains in a compact  subset of $(0,1)$ as $n \to + \infty$ with $j \in \{j_{k-1,+}+1,\ldots,j_{k,-}-1\}$ and $\ell\in \{1,\ldots,k-1\}$, while $\lambda_{j,\ell}$ remains in a compact  subset of $(1,\infty)$ as $n \to + \infty$ with $j \in \{j_{k-1,+}+1,\ldots,j_{k,-}-1\}$ and $\ell\in \{k,\ldots,2g\}$. Thus we can use Lemma \ref{lemma: asymp of gamma for lambda bounded away from 1} (i)--(ii) with $a$ and $\lambda$ replaced by $a_{j}$ and $\lambda_{j,\ell}$ respectively, where $j \in \{j_{k-1,+}+1,\ldots,j_{k,-}-1\}$ and $\ell\in \{1,\ldots,2g\}$. This yields
\begin{align}\label{lol9}
S_{2k-1} & = \sum_{j=j_{k-1,+}+1}^{j_{k,-}-1} \hspace{-0.3cm}
\log \bigg( \sum_{\ell=1}^{k-1} (-1)^{\ell+1}  \bigO(e^{-\frac{a_{j}\eta_{j,\ell}^{2}}{2}}) + \sum_{\ell=k}^{2g} (-1)^{\ell+1}  \big(1+\bigO(e^{-\frac{a_{j}\eta_{j,\ell}^{2}}{2}})\big) + 1 \bigg), 
\end{align}
as $n \to + \infty$. By \eqref{def of aj lambdajl etajl} and \eqref{lol18}, there exist constants $\{c_{j},c_{j}'\}_{j=1}^{3}$ such that $c_{1} n \leq  a_{j} \leq c_{1}'n$, $0<c_{1}$, $0<c_{2} \leq |\lambda_{j,\ell}-1| \leq c_{2}'$ and $0<c_{3} \leq \eta_{j,\ell}^{2} \leq c_{3}'$ hold for all large enough $n$, all $j \in \{j_{k-1,+}+1,\ldots,j_{k,-}-1\}$ and all $\ell\in \{1,\ldots,2g\}$. Thus $S_{2k-1}=\bigO(e^{-\frac{c_{1}c_{3}}{4}n})$ as $n \to + \infty$, which finishes the proof for $k=3,5,\ldots,2g+1$. Let us now consider the case $k=1$, which requires a slightly different argument. We infer from Lemma \ref{lemma: asymp of gamma for lambda bounded away from 1} (i) that for any $\epsilon'>0$ there exist $A=A(\epsilon'),C=C(\epsilon')>0$ such that $|\frac{\gamma(a,\lambda a)}{\Gamma(a)}-1| \leq Ce^{-\frac{a\eta^{2}}{2}}$ for all $a \geq A$ and all $\lambda \in [1+\epsilon',+\infty]$, where $\eta$ is given by \eqref{lol8}. Let us choose $\epsilon'=\frac{\epsilon}{2}$ and $M'$ sufficiently large such that $a_{j} = \frac{j+\alpha}{b} \geq A(\frac{\epsilon}{2})$ holds for all $j \in \{M'+1,\ldots,j_{1,-}-1\}$. In a similar way as in \eqref{lol9}, we obtain
\begin{align*}
S_{1} & = \sum_{j=M'+1}^{j_{1,-}-1} \hspace{-0.3cm}
\log \bigg( \sum_{\ell=1}^{2g} (-1)^{\ell+1} \big(1+\bigO(e^{-\frac{a_{j}\eta_{j,\ell}^{2}}{2}})\big) + 1 \bigg), \qquad \mbox{as } n \to + \infty.
\end{align*}
For each $\ell \in \{1,2,\ldots,2g\}$, $a_{j}\eta_{j,\ell}^{2}$ is decreasing as $j$ increases from $M'+1$ to $j_{1,-}-1$, and therefore
\begin{align*}
\frac{a_{j}\eta_{j,\ell}^{2}}{2} \geq \frac{a_{j_{1,-}-1}\eta_{j_{1,-}-1,\ell}^{2}}{2} \geq cn, \qquad \mbox{for all } j \in \{M'+1,\ldots,j_{1,-}-1\}, \; \ell \in \{1,\ldots,2g\},
\end{align*}
for a small enough constant $c>0$. It follows that $S_{1} = \bigO(e^{-cn})$ as $n \to + \infty$, which finishes the proof for $k=1$.
\end{proof}
Now, we analyze $S_{3},S_{7},\ldots,S_{4g-1}$. As it turns out, these are the sums responsible for the oscillations in the large $n$ asymptotics of $\log \mathcal{P}_{n}$. There is no such sums in \cite{Charlier 2d jumps}, so the analysis done here for $S_{3},S_{7},\ldots,S_{4g-1}$ is new. 

The next lemma makes apparent the terms that are not exponentially small.
\begin{lemma}\label{lemma: S2km1 k even splitting}
Let $k \in \{2,4,\ldots,2g\}$. There exists $c>0$ such that 
\begin{align}\label{lol30}
S_{2k-1} = S_{2k-1}^{(1)}+S_{2k-1}^{(2)}+\bigO(e^{-cn}), \qquad \mbox{as } n \to + \infty,
\end{align}
where 
\begin{align*}
& S_{2k-1}^{(1)} = \sum_{j=j_{k-1,+}+1}^{\lfloor j_{k,\star} \rfloor} \log \bigg(1+ \frac{\gamma(\tfrac{j+\alpha}{b},nr_{k-1}^{2b})}{\Gamma(\tfrac{j+\alpha}{b})} - \frac{\gamma(\tfrac{j+\alpha}{b},nr_{k}^{2b})}{\Gamma(\tfrac{j+\alpha}{b})} \bigg), \\
& S_{2k-1}^{(2)} = \sum_{j=\lfloor j_{k,\star} \rfloor+1}^{j_{k,-}-1} \log \bigg(1+ \frac{\gamma(\tfrac{j+\alpha}{b},nr_{k-1}^{2b})}{\Gamma(\tfrac{j+\alpha}{b})} - \frac{\gamma(\tfrac{j+\alpha}{b},nr_{k}^{2b})}{\Gamma(\tfrac{j+\alpha}{b})} \bigg),
\end{align*}
and 
\begin{align}\label{def of jkstar and tkstar}
j_{k,\star} := n t_{k} -\alpha,
\end{align}
where $t_{k}$ is defined in \eqref{def of t2k in proof}.
\end{lemma}
\begin{proof}
Note that \eqref{lol18} also holds for $k \in \{2,4,\ldots,2g\}$, which implies in particular that for each $\ell\in \{1,\ldots,2g\}$, $|\lambda_{j,\ell}-1|$ remains bounded away from $0$ as $n \to + \infty$ and simultaneously $j \in \{j_{k-1,+}+1,\ldots,j_{k,-}-1\}$. Thus we can use Lemma \ref{lemma: asymp of gamma for lambda bounded away from 1} (i)--(ii) with $a$ and $\lambda$ replaced by $a_{j}$ and $\lambda_{j,\ell}$ respectively, where $j \in \{j_{k-1,+}+1,\ldots,j_{k,-}-1\}$ and $\ell\in \{1,\ldots,2g\}$, and this gives
\begin{align*}
& \frac{\gamma(\tfrac{j+\alpha}{b},nr_{\ell}^{2b})}{\Gamma(\tfrac{j+\alpha}{b})} = \frac{e^{-\frac{a_{j}\eta_{j,\ell}^{2}}{2}}}{\sqrt{2\pi}}\bigg( \frac{1}{1-\lambda_{j,\ell}}\frac{1}{\sqrt{a_{j}}} + \bigO(n^{-3/2}) \bigg), & & \ell\in \{1,\ldots,k-1\}, \\
& \frac{\gamma(\tfrac{j+\alpha}{b},nr_{\ell}^{2b})}{\Gamma(\tfrac{j+\alpha}{b})} = 1+\frac{e^{-\frac{a_{j}\eta_{j,\ell}^{2}}{2}}}{\sqrt{2\pi}}\bigg( \frac{1}{1-\lambda_{j,\ell}}\frac{1}{\sqrt{a_{j}}} + \bigO(n^{-3/2}) \bigg), & & \ell\in \{k,\ldots,2g\},
\end{align*}
as $n \to + \infty$ uniformly for $j \in \{j_{k-1,+}+1,\ldots,j_{k,-}-1\}$. In a similar way as \eqref{lol18}, we derive 
\begin{align*}
& (1+\epsilon) \frac{r_{\ell}^{2b}-r_{\ell-1}^{2b}}{r_{k}^{2b}+\frac{1+\epsilon}{b n}} \leq \lambda_{j,\ell}-\lambda_{j,\ell-1} \leq (1-\epsilon) \frac{r_{\ell}^{2b}-r_{\ell-1}^{2b}}{r_{k-1}^{2b} - \frac{1-\epsilon}{bn}},
\end{align*}
for all $j \in \{j_{k-1,+}+1,\ldots,j_{k,-}-1\}$ and $\ell \in \{2,\ldots,2g\}$, which implies by \eqref{def of aj lambdajl etajl} that
\begin{align*}
& \min\Big\{\eta_{j,2}-\eta_{j,1}, \eta_{j,3}-\eta_{j,2}, \ldots, \eta_{j,k-1}-\eta_{j,k-2}, 0-\eta_{j,k-1}, \eta_{j,k}-0, \eta_{j,k+1}-\eta_{j,k}, \ldots, \eta_{j,2g}-\eta_{j,2g-1}\Big\}
\end{align*}
is positive and remains bounded away from $0$ for all $n$ sufficiently large and for all $j \in \{j_{k-1,+}+1,\ldots,j_{k,-}-1\}$. In particular,
\begin{align*}
& 1+\sum_{\ell=1}^{2g} (-1)^{\ell+1}\frac{\gamma(\tfrac{j+\alpha}{b},nr_{\ell}^{2b})}{\Gamma(\tfrac{j+\alpha}{b})} = \bigg(1+ \frac{\gamma(\tfrac{j+\alpha}{b},nr_{k-1}^{2b})}{\Gamma(\tfrac{j+\alpha}{b})} - \frac{\gamma(\tfrac{j+\alpha}{b},nr_{k}^{2b})}{\Gamma(\tfrac{j+\alpha}{b})} \bigg) (1+ \bigO(e^{-cn}))
\end{align*}
as $n \to + \infty$ uniformly for $j \in \{j_{k-1,+}+1,\ldots,j_{k,-}-1\}$, which implies
\begin{align}\label{lol19}
S_{2k-1} = \sum_{j=j_{k-1,+}+1}^{j_{k,-}-1} \log \bigg(1+ \frac{\gamma(\tfrac{j+\alpha}{b},nr_{k-1}^{2b})}{\Gamma(\tfrac{j+\alpha}{b})} - \frac{\gamma(\tfrac{j+\alpha}{b},nr_{k}^{2b})}{\Gamma(\tfrac{j+\alpha}{b})} \bigg)+\bigO(e^{-cn}), \qquad \mbox{as } n \to + \infty,
\end{align}
and the claim follows after splitting the above sum into two parts.
\end{proof}
The reason why we have split the sum in \eqref{lol19} into two parts (denoted $S_{2k-1}^{(1)}$ and $S_{2k-1}^{(2)}$) around the value $j=\lfloor j_{k,\star} \rfloor$ is the following. As can be seen from the proof of Lemma \ref{lemma: S2km1 k even splitting}, we have
\begin{align}
& \frac{\gamma(\tfrac{j+\alpha}{b},nr_{k-1}^{2b})}{\Gamma(\tfrac{j+\alpha}{b})} = \frac{e^{-\frac{a_{j}\eta_{j,k-1}^{2}}{2}}}{\sqrt{2\pi}}\bigg( \frac{1}{1-\lambda_{j,k-1}}\frac{1}{\sqrt{a_{j}}} + \bigO(n^{-3/2})\bigg), \label{lol20} \\
& 1-\frac{\gamma(\tfrac{j+\alpha}{b},nr_{k}^{2b})}{\Gamma(\tfrac{j+\alpha}{b})} = \frac{e^{-\frac{a_{j}\eta_{j,k}^{2}}{2}}}{\sqrt{2\pi}}\bigg( \frac{1}{\lambda_{j,k}-1}\frac{1}{\sqrt{a_{j}}} + \bigO(n^{-3/2})\bigg), \label{lol21}
\end{align}
as $n \to + \infty$ uniformly for $j \in \{j_{k-1,+}+1,\ldots,j_{k,-}-1\}$. The two above right-hand sides are exponentially small. To analyze their sum, it is relevant to know whether $\eta_{j,k-1}^{2}\geq \eta_{j,k}^{2}$ or $\eta_{j,k-1}^{2}< \eta_{j,k}^{2}$ holds. It is easy to check that the function $j \mapsto \eta_{j,k}^{2}-\eta_{j,k-1}^{2}$, when viewed as an analytic function of $j \in [j_{k-1,+}+1,j_{k,-}-1]$, has a simple zero at $j=j_{k,\star}$. In fact, we have
\begin{align}\label{asymp etajk-etajkm1}
\frac{a_{j}(\eta_{j,k}^{2}-\eta_{j,k-1}^{2})}{2} = 2(j_{k,\star}-j) \log\bigg( \frac{r_{k}}{r_{k-1}} \bigg),
\end{align}
which implies in particular that $\eta_{j,k}^{2}-\eta_{j,k-1}^{2}$ is positive for $j\in\{j_{k-1,+}+1,\ldots,\lfloor j_{k,\star} \rfloor\}$ and negative for $j \in \{\lfloor j_{k,\star} \rfloor+1,\ldots,j_{k,-}-1\}$. Note that $j_{k,\star}$ lies well within the interval $[j_{k-1,+}+1,j_{k,-}-1]$ for all sufficiently large $n$ by \eqref{def of jk plus and minus}, \eqref{def of t2k in proof} and \eqref{cond on epsilon 2}, which implies that the number of terms in each of the sums $S_{2k-1}^{(1)}$ and $S_{2k-1}^{(2)}$ is of order $n$. When $j$ is close to $\lfloor j_{k,\star} \rfloor$, the two terms \eqref{lol20} and \eqref{lol21} are of the same order, and this will produce the oscillations in the asymptotics of $\log \mathcal{P}_{n}$. We will evaluate $S_{2k-1}^{(1)}$ and $S_{2k-1}^{(2)}$ separately using some precise Riemann sum approximations. We first state a general lemma.
\begin{lemma}\label{lemma:Riemann sum NEW}
Let $A,a_{0}$, $B,b_{0}$ be bounded function of $n \in \{1,2,\ldots\}$, such that 
\begin{align*}
& a_{n} := An + a_{0} \qquad \mbox{ and } \qquad b_{n} := Bn + b_{0}
\end{align*}
are integers. Assume also that $B-A$ is positive and remains bounded away from $0$. Let $f$ be a function independent of $n$, and which is $C^{4}([\min\{\frac{a_{n}}{n},A\},\max\{\frac{b_{n}}{n},B\}])$ for all $n\in \{1,2,\ldots\}$. Then as $n \to + \infty$, we have
\begin{align}
&  \sum_{j=a_{n}}^{b_{n}}f(\tfrac{j}{n}) = n \int_{A}^{B}f(x)dx + \frac{(1-2a_{0})f(A)+(1+2b_{0})f(B)}{2}  \nonumber \\
& + \frac{(-1+6a_{0}-6a_{0}^{2})f'(A)+(1+6b_{0}+6b_{0}^{2})f'(B)}{12n}+ \frac{(-a_{0}+3a_{0}^{2}-2a_{0}^{3})f''(A)+(b_{0}+3b_{0}^{2}+2b_{0}^{3})f''(B)}{12n^{2}} \nonumber \\
& + \bigO \bigg( \frac{\mathfrak{m}_{A,n}(f''')+\mathfrak{m}_{B,n}(f''')}{n^{3}} + \sum_{j=a_{n}}^{b_{n}-1} \frac{\mathfrak{m}_{j,n}(f'''')}{n^{4}} \bigg), \label{sum f asymp gap NEW}
\end{align}
where, for a given function $g$ continuous on $[\min\{\frac{a_{n}}{n},A\},\max\{\frac{b_{n}}{n},B\}]$,
\begin{align*}
\mathfrak{m}_{A,n}(g) := \max_{x \in [\min\{\frac{a_{n}}{n},A\},\max\{\frac{a_{n}}{n},A\}]}|g(x)|, \quad \mathfrak{m}_{B,n}(g) := \max_{x \in [\min\{\frac{b_{n}}{n},B\},\max\{\frac{b_{n}}{n},B\}]}|g(x)|,
\end{align*}
and for $j \in \{a_{n},\ldots,b_{n}-1\}$, $\mathfrak{m}_{j,n}(g) := \max_{x \in [\frac{j}{n},\frac{j+1}{n}]}|g(x)|$.
\end{lemma}
\begin{remark}\label{remark:we can apply the lemma for functions with a pole}
To analyze the sums $S_{2k-1}^{(1)}$ and $S_{2k-1}^{(2)}$, we will use Lemma \ref{lemma:Riemann sum NEW} only with $A$ and $B$ fixed. However, we will also deal with other sums (denoted $\smash{S_{2k}^{(1)}}$ and $\smash{S_{2k}^{(3)}}$ in Lemma \ref{lemma: exact decomposition for the interpolating sums} below) that require the use of Lemma \ref{lemma:Riemann sum NEW} with varying $A$ and $B$. So it is worth to emphasize already here that the condition ``$f \in C^{4}([\min\{\frac{a_{n}}{n},A\},\max\{\frac{b_{n}}{n},B\}])$ for all $n\in \{1,2,\ldots\}$" allows to handle the situation where, for example, $A \searrow 0$ as $n \to + \infty$ and $f \in C^{4}((0,\max\{\frac{b_{n}}{n},B\}])$ but $f \notin C^{4}([0,\max\{\frac{b_{n}}{n},B\}])$.
\end{remark}
\begin{proof}
By Taylor's theorem,
\begin{align}
\int_{\frac{a_{n}}{n}}^{\frac{b_{n}}{n}}f(x)dx & = \sum_{j=a_{n}}^{b_{n}-1}\int_{\frac{j}{n}}^{\frac{j+1}{n}}f(x)dx  \nonumber \\
& = \sum_{j=a_{n}}^{b_{n}-1}\bigg\{ \frac{f(\tfrac{j}{n})}{n} + \frac{f'(\tfrac{j}{n})}{2n^{2}} + \frac{f''(\tfrac{j}{n})}{6n^{3}} + \frac{f'''(\tfrac{j}{n})}{24n^{4}} + \int_{\frac{j}{n}}^{\frac{j+1}{n}}  \frac{(x-\frac{j}{n})^{4}}{24}f''''(\xi_{j,n}(x)) dx\bigg\}, \label{lol22}
\end{align}
for some $\xi_{j,n}(x) \in [\frac{j}{n},x]$. Clearly,
\begin{align*}
\bigg| \int_{\frac{j}{n}}^{\frac{j+1}{n}}  \frac{(x-\frac{j}{n})^{4}}{24}f''''(\xi_{j,n}(x)) dx \bigg| \leq \frac{\mathfrak{m}_{j,n}(f'''')}{120n^{5}}.
\end{align*}
Therefore, by isolating the sum $\sum_{j=a_{n}}^{b_{n}-1}f(\tfrac{j}{n})$ in \eqref{lol22}, we get
\begin{align}\label{lol23}
\sum_{j=a_{n}}^{b_{n}-1}f(\tfrac{j}{n}) = n \int_{\frac{a_{n}}{n}}^{\frac{b_{n}}{n}}f(x)dx - \sum_{j=a_{n}}^{b_{n}-1}\bigg\{ \frac{f'(\tfrac{j}{n})}{2n} + \frac{f''(\tfrac{j}{n})}{6n^{2}} + \frac{f'''(\tfrac{j}{n})}{24n^{3}} \bigg\} + \bigO \bigg( \sum_{j=a_{n}}^{b_{n}-1} \frac{\mathfrak{m}_{j,n}(f'''')}{n^{4}} \bigg),
\end{align}
as $n \to + \infty$. In the same way as \eqref{lol23}, by replacing $f$ successively by $f'$, $f''$ and $f'''$, we also obtain
\begin{align}
& \sum_{j=a_{n}}^{b_{n}-1}f'(\tfrac{j}{n}) = n \int_{\frac{a_{n}}{n}}^{\frac{b_{n}}{n}}f'(x)dx - \sum_{j=a_{n}}^{b_{n}-1}\bigg\{ \frac{f''(\tfrac{j}{n})}{2n} + \frac{f'''(\tfrac{j}{n})}{6n^{2}} \bigg\} + \bigO \bigg( \sum_{j=a_{n}}^{b_{n}-1} \frac{\mathfrak{m}_{j,n}(f'''')}{n^{3}} \bigg), \label{lol24} \\
& \sum_{j=a_{n}}^{b_{n}-1}f''(\tfrac{j}{n}) = n \int_{\frac{a_{n}}{n}}^{\frac{b_{n}}{n}}f''(x)dx - \sum_{j=a_{n}}^{b_{n}-1} \frac{f'''(\tfrac{j}{n})}{2n} + \bigO \bigg( \sum_{j=a_{n}}^{b_{n}-1} \frac{\mathfrak{m}_{j,n}(f'''')}{n^{2}} \bigg), \label{lol25} \\
& \sum_{j=a_{n}}^{b_{n}-1}f'''(\tfrac{j}{n}) = n \int_{\frac{a_{n}}{n}}^{\frac{b_{n}}{n}}f'''(x)dx + \bigO \bigg( \sum_{j=a_{n}}^{b_{n}-1} \frac{\mathfrak{m}_{j,n}(f'''')}{n} \bigg), \label{lol26}
\end{align}
as $n \to + \infty$. After substituting \eqref{lol24}--\eqref{lol26} in \eqref{lol23}, we get
\begin{align}\label{lol27}
\sum_{j=a_{n}}^{b_{n}}f(\tfrac{j}{n}) = f(\tfrac{b_{n}}{n}) + \int_{\frac{a_{n}}{n}}^{\frac{b_{n}}{n}}\bigg\{ n f(x) - \frac{f'(x)}{2}+\frac{f''(x)}{12n} \bigg\}dx + \bigO \bigg( \sum_{j=a_{n}}^{b_{n}-1} \frac{\mathfrak{m}_{j,n}(f'''')}{n^{4}} \bigg),
\end{align}
as $n \to + \infty$. The integral on the right-hand side of \eqref{lol23} can be expanded using again Taylor's theorem; this gives
\begin{align*}
\int_{\frac{a_{n}}{n}}^{\frac{b_{n}}{n}}  f(x)dx = \int_{A}^{B}  f(x)dx -  \frac{a_{0}f(A)}{n} - \frac{a_{0}^{2}f'(A)}{2n^{2}} - \frac{a_{0}^{3}f''(A)}{6n^{3}} + \frac{b_{0}f(B)}{n} + \frac{b_{0}^{2}f'(B)}{2n^{2}} + \frac{b_{0}^{3}f''(B)}{6n^{3}} +\mathcal{E}_{n},
\end{align*}
for some $\mathcal{E}_{n}$ satisfying $|\mathcal{E}_{n}| \leq \frac{\mathfrak{m}_{A,n}(f''')+\mathfrak{m}_{B,n}(f''')}{n^{4}}$. 
The quantities $f(\tfrac{b_{n}}{n})$, $\int_{\frac{a_{n}}{n}}^{\frac{b_{n}}{n}}f'(x)dx$, $\int_{\frac{a_{n}}{n}}^{\frac{b_{n}}{n}}f''(x)dx$ can be expanded in a similar way using Taylor's Theorem. After substituting these expressions in \eqref{lol27} and using some elementary primitives, we find the claim.
\end{proof}
We introduce here a number of quantities that will appear in the large $n$ asymptotics of $S_{2k-1}^{(1)}$ and $S_{2k-1}^{(2)}$. For $k=2,4,\ldots,2g$, define
\begin{align}\label{lol10 start}
\theta_{k} = j_{k,\star}-\lfloor j_{k,\star} \rfloor, \qquad A_{k} = \frac{br_{k-1}^{2b}}{1-\epsilon}, \qquad B_{k} = \frac{b r_{k}^{2b}}{1+\epsilon},
\end{align}
and for $k=1,2,\ldots,2g$, define
\begin{align}
& f_{1,k}(x) = \frac{x}{b}\bigg( 1+\log \frac{br_{k}^{2b}}{x} \bigg) -r_{k}^{2b}, \nonumber \\
& f_{2,k}(x) = \bigg( \frac{1}{2}-\frac{\alpha}{b} \bigg) \log x + \frac{1}{2}\log b - \log \sqrt{2\pi} + \frac{\alpha}{b}\log(br_{k}^{2b})-\log|br_{k}^{2b}-x|, \\
& f_{3,k}(x) = - \bigg( \frac{b^{2}-6b\alpha + 6\alpha^{2}}{12b x} + \frac{b x}{(x-br_{k}^{2b})^{2}} + \frac{b-\alpha}{br_{k}^{2b}-x} \bigg), \\
& \theta_{k,+}^{(n,\epsilon)} = \bigg( \frac{b n r_{k}^{2b}}{1-\epsilon}-\alpha \bigg)-\bigg\lfloor \frac{b n r_{k}^{2b}}{1-\epsilon}-\alpha \bigg\rfloor, \qquad \theta_{k,-}^{(n,\epsilon)} = \bigg\lceil \frac{b n r_{k}^{2b}}{1+\epsilon}-\alpha \bigg\rceil-\bigg( \frac{b n r_{k}^{2b}}{1+\epsilon}-\alpha \bigg). \label{lol10 end}
\end{align}
\begin{lemma}\label{lemma:S2km1 part 1}
Let $k \in \{2,4,\ldots,2g\}$. As $n \to + \infty$, we have
\begin{align*}
& S_{2k-1}^{(1)} = n^{2} \int_{A_{k}}^{t_{k}} f_{1,k-1}(x)dx - \frac{t_{k}-A_{k}}{2}n \log n + n \bigg( (\alpha-1+\theta_{k-1,+}^{(n,\epsilon)})f_{1,k-1}(A_{k}) \\
& -(\alpha+\theta_{k})f_{1,k-1}(t_{k}) + \frac{f_{1,k-1}(t_{k})+f_{1,k-1}(A_{k})}{2} + \int_{A_{k}}^{t_{k}}f_{2,k-1}(x)dx \bigg) - \frac{\log n}{2}(\theta_{k-1,+}^{(n,\epsilon)}-\theta_{k})  \\
& + \frac{1-6(\alpha + \theta_{k})+6(\alpha + \theta_{k})^{2}}{12}(f_{1,k-1})'(t_{k}) - \frac{1+6(\alpha -1 + \theta_{k-1,+}^{(n,\epsilon)})+6(\alpha -1 + \theta_{k-1,+}^{(n,\epsilon)})^{2}}{12}(f_{1,k-1})'(A_{k}) \\
& -(\alpha + \theta_{k})f_{2,k-1}(t_{k}) + (\alpha-1+\theta_{k-1,+}^{(n,\epsilon)})f_{2,k-1}(A_{k}) + \frac{f_{2,k-1}(t_{k})+f_{2,k-1}(A_{k})}{2} \\
& + \int_{A_{k}}^{t_{k}}f_{3,k-1}(x)dx + \sum_{j=0}^{+\infty} \log \bigg\{ 1+\bigg( \frac{r_{k-1}}{r_{k}} \bigg)^{2(j+\theta_{k})} \frac{t_{k}-br_{k-1}^{2b}}{br_{k}^{2b}-t_{k}} \bigg\} + \bigO\bigg( \frac{(\log n)^{2}}{n} \bigg),
\end{align*}
where $t_{k}$ is given in \eqref{def of t2k in proof} and $f_{1,k-1},f_{2,k-1},f_{3,k-1}, A_{k}, \theta_{k}, \theta_{k-1,+}^{(n,\epsilon)}$ are given in \eqref{lol10 start}--\eqref{lol10 end}.
\end{lemma}
\begin{proof}
Recall from \eqref{lol18} that $\lambda_{j,k-1}$ remains in a compact subset of $(0,1)$ as $n \to + \infty$ uniformly for $j\in\{j_{k-1,+}+1,\ldots,j_{k,-}-1\}$, and that $\lambda_{j,k}$ remains in a compact subset of $(1,+\infty)$ as $n \to + \infty$ uniformly for $j\in\{j_{k-1,+}+1,\ldots,j_{k,-}-1\}$. Hence, by Lemma \ref{lemma: asymp of gamma for lambda bounded away from 1} (i)--(ii), as $n \to + \infty$ we have
\begin{align*}
S_{2k-1}^{(1)} &  = \sum_{j=j_{k-1,+}+1}^{\lfloor j_{k,\star} \rfloor} \log \Bigg\{ \frac{e^{-\frac{a_{j}\eta_{j,k-1}^{2}}{2}}}{\sqrt{2\pi}}\bigg( \frac{1}{1-\lambda_{j,k-1}}\frac{1}{\sqrt{a_{j}}} + \frac{1+10\lambda_{j,k-1}+\lambda_{j,k-1}^{2}}{12(\lambda_{j,k-1}-1)^{3}}\frac{1}{a_{j}^{3/2}} + \bigO(n^{-5/2})\bigg)  \\
& + \frac{e^{-\frac{a_{j}\eta_{j,k}^{2}}{2}}}{\sqrt{2\pi}}\bigg( \frac{1}{\lambda_{j,k}-1}\frac{1}{\sqrt{a_{j}}} - \frac{1+10\lambda_{j,k}+\lambda_{j,k}^{2}}{12(\lambda_{j,k}-1)^{3}}\frac{1}{a_{j}^{3/2}} + \bigO(n^{-5/2})\bigg)  \Bigg\}.
\end{align*}
Since the number of terms in $S_{2k-1}^{(1)}$, namely $\#\{j_{k-1,+}+1,\ldots,\lfloor j_{k,\star} \rfloor\}$, is of order $n$ as $n \to + \infty$, the above asymptotics can be rewritten as
\begin{align}
S_{2k-1}^{(1)} &  = \sum_{j=j_{k-1,+}+1}^{\lfloor j_{k,\star} \rfloor} \log \Bigg\{ \frac{e^{-\frac{a_{j}\eta_{j,k-1}^{2}}{2}}}{\sqrt{2\pi}}\bigg( \frac{1}{1-\lambda_{j,k-1}}\frac{1}{\sqrt{a_{j}}} + \frac{1+10\lambda_{j,k-1}+\lambda_{j,k-1}^{2}}{12(\lambda_{j,k-1}-1)^{3}}\frac{1}{a_{j}^{3/2}}\bigg) \nonumber \\
& + \frac{e^{-\frac{a_{j}\eta_{j,k}^{2}}{2}}}{\sqrt{2\pi}}\bigg( \frac{1}{\lambda_{j,k}-1}\frac{1}{\sqrt{a_{j}}} - \frac{1+10\lambda_{j,k}+\lambda_{j,k}^{2}}{12(\lambda_{j,k}-1)^{3}}\frac{1}{a_{j}^{3/2}}\bigg) \Bigg\} + \bigO(n^{-1}) \nonumber \\
& = \mathsf{S}_{n}^{(1)} n + \mathsf{S}_{n}^{(2)} \log n  + \mathsf{S}_{n}^{(3)} + \mathsf{S}_{n}^{(4)}\frac{1}{n} + \widetilde{\mathsf{S}}_{n} + \bigO(n^{-1}), \label{lol29}
\end{align}
where
\begin{align*}
& \mathsf{S}_{n}^{(1)} = \sum_{j=j_{k-1,+}+1}^{\lfloor j_{k,\star}\rfloor} \bigg\{ \frac{j/n}{b}\bigg( 1 + \log \frac{br_{k-1}^{2b}}{j/n} \bigg) - r_{k-1}^{2b} \bigg\}, \qquad \mathsf{S}_{n}^{(2)} = -\frac{1}{2}\sum_{j=j_{k-1,+}+1}^{\lfloor j_{k,\star}\rfloor} 1, \\
& \mathsf{S}_{n}^{(3)} = \sum_{j=j_{k-1,+}+1}^{\lfloor j_{k,\star} \rfloor} \bigg\{ \bigg( \frac{1}{2}-\frac{\alpha}{b} \bigg)\log(j/n) + \frac{1}{2}\log b - \log \sqrt{2\pi} +\frac{\alpha}{b}\log(br_{k-1}^{2b}) - \log \Big( j/n - br_{k-1}^{2b} \Big) \bigg\}, \\
& \mathsf{S}_{n}^{(4)} = -\sum_{j=j_{k-1,+}+1}^{\lfloor j_{k,\star} \rfloor} \bigg\{ \frac{b^{2}-6b\alpha + 6\alpha^{2}}{12 b j/n} + \frac{b j/n}{(j/n - br_{k-1}^{2b})^{2}} + \frac{\alpha-b}{j/n - br_{k-1}^{2b}} \bigg\}, \\
& \widetilde{\mathsf{S}}_{n} = \sum_{j=j_{k-1,+}+1}^{\lfloor j_{k,\star} \rfloor} \log \bigg\{ 1+e^{-\frac{a_{j}(\eta_{j,k}^{2}-\eta_{j,k-1}^{2})}{2}} \bigg( \frac{j/n-b r_{k-1}^{2b}}{br_{k}^{2b}-j/n} + \widetilde{\mathcal{E}}_{n} \bigg) \bigg\} \\
& = \sum_{j=j_{k-1,+}+1}^{\lfloor j_{k,\star} \rfloor} \log \bigg( 1+e^{-\frac{a_{j}(\eta_{j,k}^{2}-\eta_{j,k-1}^{2})}{2}}  \frac{j/n-b r_{k-1}^{2b}}{br_{k}^{2b}-j/n} \bigg) + \sum_{j=j_{k-1,+}+1}^{\lfloor j_{k,\star} \rfloor} \log \bigg( 1+e^{-\frac{a_{j}(\eta_{j,k}^{2}-\eta_{j,k-1}^{2})}{2}}  \mathcal{E}_{n} \bigg),
\end{align*}
where $\widetilde{\mathcal{E}}_{n}=\bigO(n^{-1})$ and $\mathcal{E}_{n}=\bigO(n^{-1})$ as $n \to + \infty$ uniformly for $j \in \{j_{k-1,+}+1,\ldots,\lfloor j_{k,\star} \rfloor\}$. The large $n$ asymptotics of $\mathsf{S}_{n}^{(1)}$, $\mathsf{S}_{n}^{(2)}$, $\mathsf{S}_{n}^{(3)}$ and $\mathsf{S}_{n}^{(4)}$ can be obtained using Lemma \ref{lemma:Riemann sum NEW} with
\begin{align*}
a_{n} = j_{k-1,+}+1, \quad b_{n} = \lfloor j_{k,\star} \rfloor, \quad A = \frac{br_{k-1}^{2b}}{1-\epsilon}, \quad a_{0} = 1-\alpha-\theta_{k-1,+}^{(n,\epsilon)}, \quad B = t_{k}, \quad b_{0} = -\alpha-\theta_{k},
\end{align*}
and with $f$ replaced by $f_{1,k-1}$, $-\frac{1}{2}$, $f_{2,k-1}$ and $f_{3,k-1}$ respectively. Thus it only remains to obtain the asymptotics of $\widetilde{\mathsf{S}}_{n}$. We can estimate the $\mathcal{E}_{n}$-part of $\widetilde{\mathsf{S}}_{n}$ using \eqref{asymp etajk-etajkm1} as follows:
\begin{align}
& \sum_{j=j_{k-1,+}+1}^{\lfloor j_{k,\star} \rfloor} \log \bigg( 1+e^{-\frac{a_{j}(\eta_{j,k}^{2}-\eta_{j,k-1}^{2})}{2}}  \mathcal{E}_{n} \bigg) \nonumber \\
& = \sum_{j=j_{k-1,+}+1}^{\lfloor j_{k,\star} \rfloor-\lfloor M'\log n\rfloor} \log  \bigg( 1+e^{-\frac{a_{j}(\eta_{j,k}^{2}-\eta_{j,k-1}^{2})}{2}}  \mathcal{E}_{n} \bigg) + \sum_{j=\lfloor j_{k,\star} \rfloor-\lfloor M'\log n\rfloor +1}^{\lfloor j_{k,\star} \rfloor} \log \bigg( 1+e^{-\frac{a_{j}(\eta_{j,k}^{2}-\eta_{j,k-1}^{2})}{2}}  \mathcal{E}_{n} \bigg) \nonumber \\
& = \bigO(n^{-10}) + \bigO\bigg( \frac{\log n}{n} \bigg)  = \bigO\bigg( \frac{\log n}{n} \bigg), \qquad \mbox{as } n \to + \infty, \label{basic estimate}
\end{align} 
where we recall that $M'$ is a large but fixed constant (independent of $n$). Thus we have $\widetilde{\mathsf{S}}_{n} = \mathcal{S}_{0} + \bigO ( \frac{\log n}{n})$ as $n  \to + \infty$, where
\begin{align}\label{def of Scal0}
\mathcal{S}_{0} = \sum_{j=j_{k-1,+}+1}^{\lfloor j_{k,\star} \rfloor} \log \Bigg\{ 1+e^{-\frac{a_{j}(\eta_{j,k}^{2}-\eta_{j,k-1}^{2})}{2}}  \frac{j/n-b r_{k-1}^{2b}}{br_{k}^{2b}-j/n}  \Bigg\}.
\end{align}
By changing the index of summation in \eqref{def of Scal0}, and using again \eqref{asymp etajk-etajkm1}, we get
\begin{align}
& \mathcal{S}_{0} = \sum_{j=0}^{\lfloor j_{k,\star} \rfloor - j_{k-1,+}-1} \log \Bigg\{ 1+\bigg( \frac{r_{k-1}}{r_{k}} \bigg)^{2(j+\theta_{k})} \frac{-j/n - \frac{\theta_{k}}{n}+\frac{j_{k,\star}}{n}-br_{k-1}^{2b}}{br_{k}^{2b}+j/n+\frac{\theta_{k}}{n}-\frac{j_{k,\star}}{n}}  \Bigg\} \nonumber \\
& = \sum_{j=0}^{\lfloor j_{k,\star} \rfloor - j_{k-1,+}-1} \log \Bigg\{ 1+\bigg( \frac{r_{k-1}}{r_{k}} \bigg)^{2(j+\theta_{k})} f_{0}(j/n)  \Bigg\} + \bigO \bigg( \frac{\log n}{n} \bigg), \qquad \mbox{as } n \to + \infty, \label{lol3}
\end{align}
where the error term has been estimated in a similar way as in \eqref{basic estimate}, and
\begin{align*}
f_{0}(x) := \frac{-x +t_{k}-br_{k-1}^{2b}}{br_{k}^{2b}+x-t_{k}}.
\end{align*}
To estimate the remaining sum in \eqref{lol3}, we split it into two parts as follows
\begin{align*}
& \sum_{j=0}^{\lfloor j_{k,\star} \rfloor - j_{k-1,+}-1} \log \Bigg\{ 1+\bigg( \frac{r_{k-1}}{r_{k}} \bigg)^{2(j+\theta_{k})} f_{0}(j/n)  \Bigg\} \\
& = \sum_{j=0}^{\lfloor M'\log n\rfloor} \log \Bigg\{ 1+\bigg( \frac{r_{k-1}}{r_{k}} \bigg)^{2(j+\theta_{k})} f_{0}(j/n)  \Bigg\} + \sum_{j=\lfloor M'\log n\rfloor+1}^{\lfloor j_{k,\star} \rfloor - j_{k-1,+}-1} \log \Bigg\{ 1+\bigg( \frac{r_{k-1}}{r_{k}} \bigg)^{2(j+\theta_{k})} f_{0}(j/n)  \Bigg\}.
\end{align*}
For the second part, we have
\begin{align*}
\sum_{j=\lfloor M'\log n\rfloor +1}^{\lfloor j_{k,\star} \rfloor - j_{k-1,+}-1} \log \Bigg\{ 1+\bigg( \frac{r_{k-1}}{r_{k}} \bigg)^{2(j+\theta_{k})} f_{0}(j/n)  \Bigg\} = \bigO (n^{-10}), \qquad \mbox{as } n \to + \infty,
\end{align*}
provided $M'$ is chosen large enough. For the first part, since $f_{0}$ is analytic in a neighborhood of $0$, as $n \to + \infty$ we have
\begin{align*}
& \sum_{j=0}^{\lfloor M'\log n\rfloor} \log \Bigg\{ 1+\bigg( \frac{r_{k-1}}{r_{k}} \bigg)^{2(j+\theta_{k})} f_{0}(j/n)  \Bigg\} = \sum_{j=0}^{\lfloor M'\log n\rfloor} \log \Bigg\{ 1+\bigg( \frac{r_{k-1}}{r_{k}} \bigg)^{2(j+\theta_{k})} (f_{0}(0) + \bigO(j/n)) \Bigg\} \\
& = \sum_{j=0}^{\lfloor M'\log n\rfloor} \log \Bigg\{ 1+\bigg( \frac{r_{k-1}}{r_{k}} \bigg)^{2(j+\theta_{k})} f_{0}(0) \Bigg\} + \bigO\bigg( \frac{(\log n)^{2}}{n} \bigg) \\
& = \sum_{j=0}^{+\infty} \log \Bigg\{ 1+\bigg( \frac{r_{k-1}}{r_{k}} \bigg)^{2(j+\theta_{k})} f_{0}(0) \Bigg\} + \bigO\bigg( \frac{(\log n)^{2}}{n} \bigg).
\end{align*}
Hence, we have just shown that
\begin{align}\label{lol28}
\widetilde{\mathsf{S}}_{n} = \sum_{j=0}^{+\infty} \log \Bigg\{ 1+\bigg( \frac{r_{k-1}}{r_{k}} \bigg)^{2(j+\theta_{k})} f_{0}(0) \Bigg\} + \bigO\bigg( \frac{(\log n)^{2}}{n} \bigg), \qquad \mbox{as } n \to + \infty.
\end{align}
By substituting \eqref{lol28} and the large $n$ asymptotics of $\mathsf{S}_{n}^{(1)}$, $\mathsf{S}_{n}^{(2)}$, $\mathsf{S}_{n}^{(3)}$ and $\mathsf{S}_{n}^{(4)}$ in \eqref{lol29}, we obtain the claim.
\end{proof}
The asymptotic analysis of the sums $S_{2k-1}^{(2)}$, $k=2,4,\ldots,2g$ is similar to that of the sums $S_{2k-1}^{(1)}$, $k=2,4,\ldots,2g$, so we omit the proof of the following lemma.
\begin{lemma}\label{lemma:S2km1 part 2}
Let $k \in \{2,4,\ldots,2g\}$. As $n \to + \infty$, we have
\begin{align*}
& S_{2k-1}^{(2)} = n^{2} \int_{t_{k}}^{B_{k}} f_{1,k}(x)dx - \frac{B_{k}-t_{k}}{2}n \log n + n \bigg( (\alpha-1+\theta_{k})f_{1,k}(t_{k}) \\
& -(\alpha+1-\theta_{k,-}^{(n,\epsilon)})f_{1,k}(B_{k}) + \frac{f_{1,k}(B_{k})+f_{1,k}(t_{k})}{2} + \int_{t_{k}}^{B_{k}}f_{2,k}(x)dx \bigg) - \frac{\log n}{2}(\theta_{k,-}^{(n,\epsilon)}-1+\theta_{k}) \\
& + \frac{1-6(\alpha + 1 - \theta_{k,-}^{(n,\epsilon)})+6(\alpha +1 - \theta_{k,-}^{(n,\epsilon)})^{2}}{12}(f_{1,k})'(B_{k}) - \frac{1+6(\alpha -1 + \theta_{k})+6(\alpha -1 + \theta_{k})^{2}}{12}(f_{1,k})'(t_{k}) \\
& -(\alpha + 1 - \theta_{k,-}^{(n,\epsilon)})f_{2,k}(B_{k}) + (\alpha-1+\theta_{k})f_{2,k}(t_{k}) + \frac{f_{2,k}(B_{k})+f_{2,k}(t_{k})}{2} \\
& + \int_{t_{k}}^{B_{k}}f_{3,k}(x)dx  + \sum_{j=0}^{+\infty} \log \bigg\{ 1 + \bigg( \frac{r_{k-1}}{r_{k}} \bigg)^{2(j+1-\theta_{k})} \frac{br_{k}^{2b}-t_{k}}{t_{k}-br_{k-1}^{2b}} \bigg\} +  \bigO\bigg( \frac{(\log n)^{2}}{n} \bigg),
\end{align*}
where $t_{k}$ is given in \eqref{def of t2k in proof} and $f_{1,k},f_{2,k},f_{3,k}, B_{k}, \theta_{k}, \theta_{k,-}^{(n,\epsilon)}$ are given in \eqref{lol10 start}--\eqref{lol10 end}.
\end{lemma}
Substituting the asymptotics of Lemmas \ref{lemma:S2km1 part 1} and \ref{lemma:S2km1 part 2} in \eqref{lol30}, we directly get the following.
\begin{lemma}\label{lemma: S2km1 k even}
Let $k \in \{2,4,\ldots,2g\}$. As $n \to + \infty$, we have
\begin{align*}
& S_{2k-1} = F_{1,k}^{(\epsilon)} n^{2}  + F_{2,k}^{(\epsilon)} n \log n + F_{3,k}^{(n,\epsilon)} n + F_{5,k}^{(n,\epsilon)} \log n + F_{6,k}^{(n,\epsilon)} + \widetilde{\Theta}_{k,n} + \bigO\bigg( \frac{(\log n)^{2}}{n} \bigg)
\end{align*}
where
\begin{align*}
& F_{1,k}^{(\epsilon)} = \int_{A_{k}}^{t_{k}} f_{1,k-1}(x)dx + \int_{t_{k}}^{B_{k}} f_{1,k}(x)dx, \\
& F_{2,k}^{(\epsilon)} = -\frac{B_{k}-A_{k}}{2}, \\
& F_{3,k}^{(n,\epsilon)} = (\alpha-1+\theta_{k-1,+}^{(n,\epsilon)})f_{1,k-1}(A_{k}) -(\alpha+\theta_{k})f_{1,k-1}(t_{k}) +(\alpha-1+\theta_{k})f_{1,k}(t_{k}) - (\alpha+1-\theta_{k,-}^{(n,\epsilon)})f_{1,k}(B_{k}) \\
& + \frac{f_{1,k-1}(A_{k})+f_{1,k-1}(t_{k})}{2} + \frac{f_{1,k}(t_{k})+f_{1,k}(B_{k})}{2} + \int_{A_{k}}^{t_{k}} f_{2,k-1}(x)dx + \int_{t_{k}}^{B_{k}}f_{2,k}(x)dx, \\
& F_{5,k}^{(n,\epsilon)} = \frac{1-\theta_{k-1,+}^{(n,\epsilon)}-\theta_{k,-}^{(n,\epsilon)}}{2}, \\
& F_{6,k}^{(n,\epsilon)} = - \frac{1+6(\alpha-1+\theta_{k-1,+}^{(n,\epsilon)})+6(\alpha-1+\theta_{k-1,+}^{(n,\epsilon)})^{2}}{12}f_{1,k-1}'(A_{k}) + \frac{1-6(\alpha+\theta_{k})+6(\alpha+\theta_{k})^{2}}{12}f_{1,k-1}'(t_{k}) \\
& - \frac{1+6(\alpha-1+\theta_{k})+6(\alpha-1+\theta_{k})^{2}}{12}f_{1,k}'(t_{k}) + \frac{1-6(\alpha+1-\theta_{k,-}^{(n,\epsilon)})+6(\alpha+1-\theta_{k,-}^{(n,\epsilon)})^{2}}{12}f_{1,k}'(B_{k}) \\
& +(\alpha-1+\theta_{k-1,+}^{(n,\epsilon)})f_{2,k-1}(A_{k}) - (\alpha+\theta_{k})f_{2,k-1}(t_{k}) + (\alpha-1+\theta_{k})f_{2,k}(t_{k}) - (\alpha+1-\theta_{k,-}^{(n,\epsilon)}))f_{2,k}(B_{k}) \\
& + \frac{f_{2,k-1}(A_{k})+f_{2,k-1}(t_{k})}{2} + \frac{f_{2,k}(t_{k})+f_{2,k}(B_{k})}{2}  + \int_{A_{k}}^{t_{k}}f_{3,k-1}(x)dx + \int_{t_{k}}^{B_{k}}f_{3,k}(x)dx, \\
& \widetilde{\Theta}_{k,n} = \sum_{j=0}^{+\infty} \log \bigg\{ 1 + \bigg( \frac{r_{k-1}}{r_{k}} \bigg)^{2(j+\theta_{k})} \frac{t_{k}-br_{k-1}^{2b}}{br_{k}^{2b}-t_{k}} \bigg\} + \sum_{j=0}^{+\infty} \log \bigg\{ 1 + \bigg( \frac{r_{k-1}}{r_{k}} \bigg)^{2(j+1-\theta_{k})} \frac{br_{k}^{2b}-t_{k}}{t_{k}-br_{k-1}^{2b}} \bigg\},
\end{align*}
and where $t_{k}$ is given in \eqref{def of t2k in proof} and $f_{1,k},f_{2,k},f_{3,k}, A_{k}, B_{k}, \theta_{k}, \theta_{k,+}^{(n,\epsilon)}, \theta_{k,-}^{(n,\epsilon)}$ are given in \eqref{lol10 start}--\eqref{lol10 end}.
\end{lemma}
\begin{lemma}\label{lemma:Explicit expression for the F coeff}
Let $k \in \{2,4,\ldots,2g\}$. The quantities $F_{1,k}^{(\epsilon)},F_{2,k}^{(\epsilon)},F_{3,k}^{(n,\epsilon)},F_{5,k}^{(n,\epsilon)},F_{6,k}^{(n,\epsilon)}$ appearing in Lemma \ref{lemma: S2km1 k even} can be rewritten more explicitly as follows:
\begin{align*}
& F_{1,k}^{(\epsilon)} = \frac{(r_{k}^{2b}-r_{k-1}^{2b})^{2}}{4\log(\frac{r_{k}}{r_{k-1}})} + \frac{br_{k-1}^{4b}}{(1-\epsilon)^{2}}\frac{1-4\epsilon - 2 \log(1-\epsilon)}{4} - \frac{br_{k}^{4b}}{(1+\epsilon)^{2}}\frac{1+4\epsilon - 2 \log(1+\epsilon)}{4}, \\
& F_{2,k}^{(\epsilon)} = -\frac{br_{k}^{2b}}{2(1+\epsilon)}+\frac{br_{k-1}^{2b}}{2(1-\epsilon)}, \\
& F_{3,k}^{(n,\epsilon)} = \frac{r_{k-1}^{2b}}{1-\epsilon}\bigg\{ \frac{2\alpha-1+2\theta_{k-1,+}^{(n,\epsilon)}}{2}(\epsilon + \log(1-\epsilon)) - \frac{b+2\alpha}{2} - b \log b + \frac{b}{2}\log(2\pi) - b^{2}\log(r_{k-1}) \\
&  -\frac{2\alpha-b}{2}\log(1-\epsilon) + b \epsilon \log \bigg( \frac{\epsilon b r_{k-1}^{2b}}{1-\epsilon} \bigg) \bigg\} + \frac{r_{k}^{2b}}{1+\epsilon}\bigg\{ \frac{2\alpha+1-2\theta_{k,-}^{(n,\epsilon)}}{2}(\epsilon - \log(1+\epsilon)) + \frac{b+2\alpha}{2} + b \log b  \\
& - \frac{b}{2}\log(2\pi) + b^{2}\log(r_{k}) +\frac{2\alpha-b}{2}\log(1+\epsilon) + b \epsilon \log \bigg( \frac{\epsilon b r_{k}^{2b}}{1+\epsilon} \bigg) \bigg\} + 2\alpha t_{k} \log \frac{r_{k-1}}{r_{k}} \\
& - \big( t_{k}-br_{k-1}^{2b} \big)\log(t_{k}-br_{k-1}^{2b}) - (br_{k}^{2b}-t_{k})\log(br_{k}^{2b}-t_{k}), \\
& F_{5,k}^{(n,\epsilon)} = \frac{1-\theta_{k-1,+}^{(n,\epsilon)}-\theta_{k,-}^{(n,\epsilon)}}{2}, \\
& F_{6,k}^{(n,\epsilon)} = \frac{1-3b+b^{2}+6(b-1)\theta_{k,-}^{(n,\epsilon)}+6(\theta_{k,-}^{(n,\epsilon)})^{2}}{12b} \log(1+\epsilon) -\frac{2b}{\epsilon} + \big(1-\theta_{k-1,+}^{(n,\epsilon)}-\theta_{k,-}^{(n,\epsilon)}\big)\log \epsilon \\
& -\frac{1+3b+b^{2}-6(1+b)\theta_{k-1,+}^{(n,\epsilon)}+6(\theta_{k-1,+}^{(n,\epsilon)})^{2}}{12b} \log(1-\epsilon) + \bigg( \frac{1}{2}-\alpha-\theta_{k-1,+}^{(n,\epsilon)} \bigg) \log \big( r_{k-1}^{b}\sqrt{2\pi} \big) \\
& + \bigg( \frac{1}{2}+\alpha-\theta_{k,-}^{(n,\epsilon)} \bigg) \log \big( r_{k}^{b}\sqrt{2\pi} \big) + \bigg( \frac{1+b^{2}+6b\alpha}{6}-\theta_{k}+\theta_{k}^{2} \bigg) \log \frac{r_{k-1}}{r_{k}} \\
& + \bigg( \theta_{k}-\frac{1}{2} \bigg)\log \bigg( \frac{t_{k}-br_{k-1}^{2b}}{br_{k}^{2b}-t_{k}} \bigg) + \frac{b^{2}r_{k}^{2b}}{br_{k}^{2b}-t_{k}} + \frac{b^{2}r_{k-1}^{2b}}{t_{k}-br_{k-1}^{2b}}.
\end{align*}
\end{lemma}
\begin{proof}
It suffices to substitute the quantities $f_{1,k},f_{2,k},f_{3,k}, A_{k}, B_{k}$ by their definitions \eqref{def of jkstar and tkstar}, \eqref{lol10 start}--\eqref{lol10 end}, and to perform some simple primitives.
\end{proof}

We now turn our attention to the sums $S_{2k}$, $k=1,\ldots,2g$. Their analysis is very different from the analysis of $S_{2k-1}$. We first make apparent the terms that are not exponentially small.
\begin{lemma}
Let $k \in \{1,3,\ldots,2g-1\}$. There exists $c>0$ such that
\begin{align}\label{lol32}
& S_{2k} = \sum_{j=j_{k,-}}^{j_{k,+}} \log \bigg( \frac{\gamma(a_{j},nr_{k}^{2b})}{\Gamma(a_{j})} \bigg) + \bigO(e^{-cn}), \qquad \mbox{as } n \to + \infty.
\end{align}
Let $k \in \{2,4,\ldots,2g\}$. There exists $c>0$ such that
\begin{align}\label{lol33}
& S_{2k} = \sum_{j=j_{k,-}}^{j_{k,+}} \log \bigg( 1- \frac{\gamma(a_{j},nr_{k}^{2b})}{\Gamma(a_{j})} \bigg) + \bigO(e^{-cn}), \qquad \mbox{as } n \to + \infty.
\end{align}
\end{lemma}
\begin{proof}
By definition of $j_{k,-}$, $j_{k,+}$ and $\lambda_{j,\ell}$ (see \eqref{def of aj lambdajl etajl} and \eqref{def of jk plus and minus}), for $j \in \{j_{k,-},\ldots,j_{k,+}\}$ we have 
\begin{align}\label{lol31}
(1-\epsilon) \frac{r_{\ell}^{2b}}{r_{k}^{2b}} \leq \lambda_{j,\ell} \leq (1+\epsilon) \frac{r_{\ell}^{2b}}{r_{k}^{2b}} \quad \mbox{ and } \quad (1-\epsilon) \frac{r_{\ell}^{2b}-r_{k}^{2b}}{r_{k}^{2b}} \leq \lambda_{j,\ell}-\lambda_{k,\ell} \leq (1+\epsilon) \frac{r_{\ell}^{2b}-r_{k}^{2b}}{r_{k}^{2b}}.
\end{align}
Since $\epsilon>0$ is fixed, the second part of \eqref{lol31} implies that for each $\ell \neq k$, $\lambda_{j,\ell}-\lambda_{k,\ell}$ remains bounded away from $0$ as $n \to + \infty$ uniformly for $j \in \{j_{k,-},\ldots,j_{k,+}\}$, and the first part of \eqref{lol31} combined with \eqref{cond on epsilon 1} implies that for all $j \in \{j_{k,-},\ldots,j_{k,+}\}$ we have
\begin{align*}
\begin{cases}
\lambda_{j,\ell} \in [1-\epsilon,1+\epsilon], & \mbox{if } \ell = k, \\
\lambda_{j,\ell} \leq (1+\epsilon) \frac{r_{\ell}^{2b}}{r_{k}^{2b}} < 1-\epsilon, & \mbox{if } \ell \leq k-1, \\
\lambda_{j,\ell} \geq (1-\epsilon) \frac{r_{\ell}^{2b}}{r_{k}^{2b}} > 1+\epsilon, & \mbox{if } \ell \geq k+1.
\end{cases}
\end{align*}
Thus, by \eqref{def of S2k} and Lemma \ref{lemma: asymp of gamma for lambda bounded away from 1} (i)--(ii), we have
\begin{align*}
S_{2k} & = \sum_{j=j_{k,-}}^{j_{k,+}} \log \bigg( \sum_{\ell=1}^{k-1} (-1)^{\ell+1} \bigO(e^{-\frac{a_{j}\eta_{j,\ell}^{2}}{2}}) +  (-1)^{k+1} \frac{\gamma(\tfrac{j+\alpha}{b},nr_{k}^{2b})}{\Gamma(\tfrac{j+\alpha}{b})} + \sum_{\ell=k+1}^{2g+1} (-1)^{\ell+1}\big(1+\bigO(e^{-\frac{a_{j}\eta_{j,\ell}^{2}}{2}})\big) \bigg) \\
& = \sum_{j=j_{k,-}}^{j_{k,+}} \log \bigg( (-1)^{k+1} \frac{\gamma(\tfrac{j+\alpha}{b},nr_{k}^{2b})}{\Gamma(\tfrac{j+\alpha}{b})} + \sum_{\ell=k+1}^{2g+1} (-1)^{\ell+1} \bigg) + \bigO(e^{-cn}),
\end{align*}
as $n \to + \infty$ for some constant $c>0$, and the claim follows.
\end{proof}
Let $M=n^{\frac{1}{12}}$. We now split the sums on the right-hand sides of \eqref{lol32} and \eqref{lol33} into three parts $S_{2k}^{(1)}$, $S_{2k}^{(2)}$, $S_{2k}^{(3)}$, which are defined as follows
\begin{align}\label{asymp prelim of S2kpvp}
& S_{2k}^{(v)} = \begin{cases}
\displaystyle \sum_{j:\lambda_{j,k}\in I_{v}}  \log \bigg( \frac{\gamma(a_{j},nr_{k}^{2b})}{\Gamma(a_{j})} \bigg), & \mbox{if } k \in \{1,3,\ldots,2g-1\}, \\
\displaystyle  \sum_{j:\lambda_{j,k}\in I_{v}}  \log \bigg( 1-\frac{\gamma(a_{j},nr_{k}^{2b})}{\Gamma(a_{j})} \bigg), & \mbox{if } k \in \{2,4,\ldots,2g\},
\end{cases} \quad v=1,2,3, 
\end{align}
where
\begin{align}\label{def of the three intervals}
I_{1} = [1-\epsilon,1-\tfrac{M}{\sqrt{n}}), \qquad I_{2} = [1-\tfrac{M}{\sqrt{n}},1+\tfrac{M}{\sqrt{n}}], \qquad I_{3} = (1+\tfrac{M}{\sqrt{n}},1+\epsilon].
\end{align}
With this notation, the asymptotics \eqref{lol32} and \eqref{lol33} can be rewritten as
\begin{align}\label{S2k splitting in three parts}
& S_{2k}=S_{2k}^{(1)} + S_{2k}^{(2)} + S_{2k}^{(3)} + \bigO(e^{-cn}), \qquad \mbox{as } n \to + \infty, \qquad k=1,2,\ldots,2g.
\end{align}
Define also 
\begin{align*}
g_{k,-} := \bigg\lceil \frac{bnr_{k}^{2b}}{1+\frac{M}{\sqrt{n}}}-\alpha \bigg\rceil, \qquad g_{k,+} := \bigg\lfloor \frac{bnr_{k}^{2b}}{1-\frac{M}{\sqrt{n}}}-\alpha \bigg\rfloor, \qquad k=1,2,\ldots,2g,
\end{align*}
so that (formally) we can write
\begin{align}\label{sums lambda j}
& \sum_{j:\lambda_{j,k}\in I_{3}} = \sum_{j=j_{k,-}}^{g_{k,-}-1}, \qquad \sum_{j:\lambda_{j,k}\in I_{2}} = \sum_{j= g_{k,-}}^{g_{k,+}}, \qquad \sum_{j:\lambda_{j,k}\in I_{1}} = \sum_{j= g_{k,+}+1}^{j_{k,+}}. 
\end{align}
The individual sums $S_{2k}^{(1)}$, $S_{2k}^{(2)}$ and $S_{2k}^{(3)}$ depend on this new parameter $M$, but their sum $S_{2k}^{(1)} + \smash{S_{2k}^{(2)}} + \smash{S_{2k}^{(3)}}$ does not. Note also that $S_{2k}^{(2)}$ is independent of the other parameter $\epsilon$, while $S_{2k}^{(1)}$ and $S_{2k}^{(3)}$ do depend on $\epsilon$. The analysis of $S_{2k}^{(2)}$ is very different from the one needed for $S_{2k}^{(1)}$ and $S_{2k}^{(3)}$. For $S_{2k}^{(1)}$ and $S_{2k}^{(3)}$, we will approximate several sums of the form $\sum_{j} f(j/n)$ for some functions $f$, while for $S_{2k}^{(2)}$, we will approximate several sums of the form $\sum_{j} h(M_{j,k})$ for some functions $h$, where $M_{j,k}:=\sqrt{n}(\lambda_{j,k}-1)$. As can be seen from \eqref{asymp prelim of S2kpvp} and \eqref{def of the three intervals}, the sum $S_{2k}^{(2)}$ involves the $j$'s for which 
\begin{align}\label{lol55}
\lambda_{j,k} \in I_{2}=[1-\tfrac{M}{\sqrt{n}},1+\tfrac{M}{\sqrt{n}}], \qquad \mbox{i.e. } M_{j,k} \in [-M,M].
\end{align}
Let us briefly comment on our choice of $M$. An essential difficulty in analyzing $S_{2k}^{(1)}$, $S_{2k}^{(2)}$, $S_{2k}^{(3)}$ is that all the functions $f$ and $h$ will blow up near certain points. To analyze $\smash{S_{2k}^{(2)}}$, it would be simpler to define $M$ as being, for example, of order $\log n$, but in this case the sums $\sum_{j} f(j/n)$ involve some $j/n$'s that are too close to the poles of $f$. On the other hand, if $M$ would be of order $\sqrt{n}$, then $\smash{S_{2k}^{(1)}}$ and $S_{2k}^{(3)}$ could be analyzed in essentially the same way as the sums $S_{2k-1}^{(1)}$ and $S_{2k-1}^{(2)}$ of Lemmas \ref{lemma:S2km1 part 1} and \ref{lemma:S2km1 part 2} above (and if $M=\epsilon \sqrt{n}$, then the sums $S_{2k}^{(1)}$ and $S_{2k}^{(3)}$ are even empty sums), but in this case the sums $\sum_{j} h(M_{j,k})$ involve some $M_{j,k}$'s that are too close to the poles of $h$. Thus we are tight up from both sides: $M$ of order $\log n$ is not large enough, and $M$ of order $\sqrt{n}$ is too large. The reason why we choose exactly $M=\smash{n^{\frac{1}{12}}}$ is very technical and will be discussed later.

We also mention that sums of the form $\sum_{j} h(M_{j,k})$ were already approximated in \cite{Charlier 2d jumps}, so we will be able to recycle some results from there. However, even for these sums, our situation presents an important extra difficulty compared with \cite{Charlier 2d jumps}, namely that in \cite{Charlier 2d jumps} the functions $h$ are bounded, while in our case they blow up near either $+ \infty$ or $-\infty$. 

We now introduce some new quantities that will appear in the large $n$ asymptotics of $S_{2k}^{(1)}$, $S_{2k}^{(2)}$ and $S_{2k}^{(3)}$. For $k \in \{1,2,\ldots,g\}$, define
\begin{align*}
& \theta_{k,-}^{(n,M)} := g_{k,-} - \bigg( \frac{bn r_{k}^{2b}}{1+\frac{M}{\sqrt{n}}} - \alpha \bigg) = \bigg\lceil \frac{bn r_{k}^{2b}}{1+\frac{M}{\sqrt{n}}} - \alpha \bigg\rceil - \bigg( \frac{bn r_{k}^{2b}}{1+\frac{M}{\sqrt{n}}} - \alpha \bigg), \\
& \theta_{k,+}^{(n,M)} := \bigg( \frac{bn r_{k}^{2b}}{1-\frac{M}{\sqrt{n}}} - \alpha \bigg) - g_{k,+} = \bigg( \frac{bn r_{k}^{2b}}{1-\frac{M}{\sqrt{n}}} - \alpha \bigg) - \bigg\lfloor \frac{bn r_{k}^{2b}}{1-\frac{M}{\sqrt{n}}} - \alpha \bigg\rfloor.
\end{align*}
Clearly, $\theta_{k,-}^{(n,M)},\theta_{k,+}^{(n,M)} \in [0,1)$. For what follows, it is useful to note that $M_{j,k}$ is decreasing as $j$ increases, and that $\sum_{j=g_{k,-}}^{g_{k,+}}1$ is of order $\frac{M}{\sqrt{n}}$ as $n \to + \infty$.

We start with a general lemma needed for the analysis of $S_{2k}^{(2)}$. 
\begin{lemma}(Adapted from \cite[Lemma 2.7]{Charlier 2d jumps})\label{lemma:Riemann sum}
Let $h \in C^{3}(\mathbb{R})$ and $k \in \{1,\ldots,2g\}$. As $n \to + \infty$, we have
\begin{align}
& \sum_{j=g_{k,-}}^{g_{k,+}}h(M_{j,k}) = br_{k}^{2b} \int_{-M}^{M} h(t) dt \; \sqrt{n} - 2 b r_{k}^{2b} \int_{-M}^{M} th(t) dt + \bigg( \frac{1}{2}-\theta_{k,-}^{(n,M)} \bigg)h(M)+ \bigg( \frac{1}{2}-\theta_{k,+}^{(n,M)} \bigg)h(-M) \nonumber \\
& + \frac{1}{\sqrt{n}}\bigg[ 3br_{k}^{2b} \int_{-M}^{M}t^{2}h(t)dt + \bigg( \frac{1}{12}+\frac{\theta_{k,-}^{(n,M)}(\theta_{k,-}^{(n,M)}-1)}{2} \bigg)\frac{h'(M)}{br_{k}^{2b}} - \bigg( \frac{1}{12}+\frac{\theta_{k,+}^{(n,M)}(\theta_{k,+}^{(n,M)}-1)}{2} \bigg)\frac{h'(-M)}{br_{k}^{2b}} \bigg] \nonumber \\
& + \bigO\Bigg(  \frac{1}{n^{3/2}} \sum_{j=g_{k,-}+1}^{g_{k,+}} \bigg( (1+|M_{j}|^{3}) \tilde{\mathfrak{m}}_{j,n}(h) + (1+M_{j}^{2})\tilde{\mathfrak{m}}_{j,n}(h') + (1+|M_{j}|) \tilde{\mathfrak{m}}_{j,n}(h'') + \tilde{\mathfrak{m}}_{j,n}(h''') \bigg)   \Bigg), \label{sum f asymp 2}
\end{align}
where, for $\tilde{h} \in C(\mathbb{R})$ and $j \in \{g_{k,-}+1,\ldots,g_{k,+}\}$, we define $\tilde{\mathfrak{m}}_{j,n}(\tilde{h}) := \max_{x \in [M_{j,k},M_{j-1,k}]}|\tilde{h}(x)|$.
\end{lemma}
\begin{remark}
Note that $\tilde{\mathfrak{m}}_{j,n}$ depends on $k$, although this is not indicated in the notation.
\end{remark}
\begin{remark}
If $|h|, |h'|, |h''|$ and $|h'''|$ are bounded, then the error term simplifies to $\bigO(M^{4}n^{-1})$, which agrees with \cite[Lemma 2.7]{Charlier 2d jumps}.
\end{remark}
\begin{proof}
This lemma was proved in \cite[Lemma 2.7]{Charlier 2d jumps} in the case where $|h|, |h'|, |h''|$ and $|h'''|$ are bounded. The more general case considered here only requires more careful estimates on the various error terms. For simplicity, we will write $M_{j}$ instead of $M_{j,k}$. By Taylor's theorem, we have
\begin{align}
 & \int_{M_{g_{k,+}}}^{M_{g_{k,-}}}h(t)dt = \sum_{j=g_{k,-}+1}^{g_{k,+}}\int_{M_{j}}^{M_{j-1}}h(t)dt = \sum_{j=g_{k,-}+1}^{g_{k,+}} \bigg\{ h(M_{j})(M_{j-1}-M_{j}) \nonumber \\
&  + h'(M_{j})\frac{(M_{j-1}-M_{j})^{2}}{2} + h''(M_{j})\frac{(M_{j-1}-M_{j})^{3}}{6} \bigg\} + \sum_{j=g_{k,-}+1}^{g_{k,+}}\int_{M_{j}}^{M_{j-1}}\frac{(x-M_{j})^{3}}{6}h'''(\xi_{j,n}(x))dt, \label{riemann sum 1}
\end{align}
for some $\xi_{j,n}(x) \in [M_{j},x]$. For the summand of the last term, we have the bound
\begin{align*}
\bigg| \int_{M_{j}}^{M_{j-1}}\frac{(x-M_{j})^{3}}{6}h'''(\xi_{j,n}(x))dt \bigg| \leq \frac{(M_{j-1}-M_{j})^{4}}{24}\tilde{\mathfrak{m}}_{j,n}(h''').
\end{align*}
Since the following asymptotics
\begin{align}
M_{j-1}-M_{j} = \frac{bn^{3/2}r_{k}^{2b}}{(j+\alpha)(j-1+\alpha)} = \frac{1}{br_{k}^{2b}n^{1/2}} + \frac{2M_{j}}{br_{k}^{2b}n}+\bigg( \frac{1}{b^{2}r_{k}^{4b}}+\frac{M_{j}^{2}}{br_{k}^{2b}} \bigg)\frac{1}{n^{3/2}}+\bigO(M_{j}n^{-2}), \label{diff of Mj}
\end{align}
hold as $n \to +\infty$ uniformly for $j \in \{g_{k,-}+1,\ldots,g_{k,+}\}$, by rearranging the terms in \eqref{riemann sum 1} we obtain
\begin{align}
& \sum_{j=g_{k,-}+1}^{g_{k,+}}h(M_{j}) = br_{k}^{2b} \int_{M_{g_{k,+}}}^{M_{g_{k,-}}} h(t) dt \; \sqrt{n} -\frac{1}{\sqrt{n}} \sum_{j=g_{k,-}+1}^{g_{k,+}}\bigg( 2 M_{j}h(M_{j}) + \frac{1}{2br_{k}^{2b}}h'(M_{j}) \bigg) \nonumber \\
& -\frac{1}{n}\sum_{j=g_{k,-}+1}^{g_{k,+}}\bigg( \frac{1}{br_{k}^{2b}} h(M_{j}) + M_{j}^{2}h(M_{j}) + \frac{2}{br_{k}^{2b}} M_{j}h'(M_{j}) + \frac{1}{6b^{2}r_{k}^{4b}}h''(M_{j}) \bigg) + \widetilde{\mathcal{E}}_{1} \label{sum f asymp}
\end{align}
as $n \to + \infty$, where $\widetilde{\mathcal{E}}_{1}$ satisfies
\begin{align*}
|\widetilde{\mathcal{E}}_{1}| \leq \frac{C}{n^{3/2}} \sum_{j=g_{k,-}+1}^{g_{k,+}} \bigg( |M_{j}| \tilde{\mathfrak{m}}_{j,n}(h) + (1+M_{j}^{2})\tilde{\mathfrak{m}}_{j,n}(h') + |M_{j}| \tilde{\mathfrak{m}}_{j,n}(h'') + \tilde{\mathfrak{m}}_{j,n}(h''') \bigg) 
\end{align*}
for all large enough $n$ and for a certain $C>0$. The sums appearing on the right-hand side of \eqref{sum f asymp} can be analyzed similarly. For example, we have
\begin{subequations}\label{lol1}
\begin{align}
& \frac{1}{\sqrt{n}}\sum_{j=g_{k,-}+1}^{g_{k,+}}M_{j}h(M_{j}) = br_{k}^{2b}\int_{M_{g_{k,+}}}^{M_{g_{k,-}}} th(t) dt  \nonumber \\
&  - \frac{1}{n}\sum_{j=g_{k,-}+1}^{g_{k,+}}\bigg( 2M_{j}^{2}h(M_{j}) + \frac{1}{2br_{k}^{2b}} h(M_{j}) + \frac{1}{2br_{k}^{2b}}M_{j}h'(M_{j}) \bigg) + \widetilde{\mathcal{E}}_{2}, \\
& \frac{1}{\sqrt{n}}\sum_{j=g_{k,-}+1}^{g_{k,+}}h'(M_{j}) = br_{k}^{2b} \Big( h(M_{g_{k,-}})-h(M_{g_{k,+}}) \Big)  \nonumber \\
& - \frac{1}{n}\sum_{j=g_{k,-}+1}^{g_{k,+}} \bigg( 2 M_{j}h'(M_{j}) + \frac{1}{2br_{k}^{2b}}h''(M_{j}) \bigg) + \widetilde{\mathcal{E}}_{3}, \\
& \sum_{j=g_{k,-}+1}^{g_{k,+}} M_{j}^{2}h(M_{j}) = br_{k}^{2b} \int_{M_{g_{k,+}}}^{M_{g_{k,-}}} t^{2}h(t) dt \; \sqrt{n} +\widetilde{\mathcal{E}}_{4}, \\
& \sum_{j=g_{k,-}+1}^{g_{k,+}}h''(M_{j}) = br_{k}^{2b} \Big( h'(M_{g_{k,-}})-h'(M_{g_{k,+}}) \Big) \sqrt{n} + \widetilde{\mathcal{E}}_{5},
\end{align}
\end{subequations}
where $|\widetilde{\mathcal{E}}_{2}|$, $|\widetilde{\mathcal{E}}_{3}|$, $|\widetilde{\mathcal{E}}_{4}|$, $|\widetilde{\mathcal{E}}_{5}|$ are all bounded by
\begin{align}\label{a bound}
\frac{C}{n^{3/2}} \sum_{j=g_{k,-}+1}^{g_{k,+}} \bigg( (1+|M_{j}|^{3}) \tilde{\mathfrak{m}}_{j,n}(h) + (1+M_{j}^{2})\tilde{\mathfrak{m}}_{j,n}(h') + (1+|M_{j}|) \tilde{\mathfrak{m}}_{j,n}(h'') + \tilde{\mathfrak{m}}_{j,n}(h''') \bigg)
\end{align}
for some $C>0$ and for all sufficiently large $n$. After substituting \eqref{lol1} in \eqref{sum f asymp} we obtain
\begin{align}
& \sum_{j=g_{k,-}}^{g_{k,+}}h(M_{j})  = br_{k}^{2b}\int_{M_{g_{k,+}}}^{M_{g_{k,-}}} h(t) dt \; \sqrt{n} - 2 b r_{k}^{2b} \int_{M_{g_{k,+}}}^{M_{g_{k,-}}} th(t) dt + \frac{h(M_{g_{k,-}})+h(M_{g_{k,+}})}{2} \nonumber \\
& + \frac{br_{k}^{2b}}{\sqrt{n}}\bigg( 3 \int_{M_{g_{k,+}}}^{M_{g_{k,-}}} t^{2}h(t) dt + \frac{h'(M_{g_{k,-}})-h'(M_{g_{k,+}})}{12b^{2}r_{k}^{4b}} \bigg) + \widetilde{\mathcal{E}}_{6}, \label{lol2}
\end{align}
where $|\widetilde{\mathcal{E}}_{6}|$ is also bounded by \eqref{a bound} for a certain $C>0$ and for all large enough $n$. We then use
\begin{align*}
& M_{g_{k,-}} = M - \frac{\theta_{k,-}^{(n,M)}}{br_{k}^{2b}\sqrt{n}} - \frac{2M \theta_{k,-}^{(n,M)}}{br_{k}^{2b}n} + \bigO(M^{2}n^{-\frac{3}{2}}), & & \mbox{as } n \to + \infty, \\
& M_{g_{k,+}} = -M + \frac{\theta_{k,+}^{(n,M)}}{b r_{k}^{2b}\sqrt{n}} - \frac{2M \theta_{k,+}^{(n,M)}}{br_{k}^{2b}n} + \bigO(M^{2}n^{-\frac{3}{2}}), & & \mbox{as } n \to + \infty,
\end{align*}
to obtain asymptotics for the integrals in \eqref{lol2}. The errors in these expansions are smaller than \eqref{a bound}. The claim then follows after a computation.
\end{proof}

\begin{lemma}\label{lemma:S2ktilde}
Let $k \in \{1,2,\ldots,2g\}$. As $n \to + \infty$, we have
\begin{align}\label{lol39}
& S_{2k}^{(2)} = G_{4,k}^{(M)}\sqrt{n} + G_{6,k}^{(M)} + G_{7,k}^{(M)} \frac{1}{\sqrt{n}} + \bigO(M^{9}n^{-1}),
\end{align}
where
\begin{align}
& G_{4,k}^{(M)} = br_{k}^{2b} \int_{-M}^{M}h_{0,k}(x)dx, \label{def of G4pMp} \\
& G_{6,k}^{(M)} = - 2br_{k}^{2b}\int_{-M}^{M}x h_{0,k}(x)dx + \bigg( \frac{1}{2}-\theta_{k,-}^{(n,M)} \bigg)h_{0,k}(M) + \bigg( \frac{1}{2}-\theta_{k,+}^{(n,M)} \bigg)h_{0,k}(-M) \nonumber \\
&  + br_{k}^{2b}\int_{-M}^{M}h_{1,k}(x)dx, \label{def of G6pMp} \\
& G_{7,k}^{(M)} = 3br_{k}^{2b}\int_{-M}^{M}x^{2}h_{0,k}(x)dx + \bigg( \frac{1}{12}+\frac{\theta_{k,-}^{(n,M)}(\theta_{k,-}^{(n,M)}-1)}{2} \bigg)\frac{h_{0,k}'(M)}{br_{k}^{2b}} \nonumber \\
& - \bigg( \frac{1}{12}+\frac{\theta_{k,+}^{(n,M)}(\theta_{k,+}^{(n,M)}-1)}{2} \bigg)\frac{h_{0,k}'(-M)}{br_{k}^{2b}} -2b r_{k}^{2b}\int_{-M}^{M}xh_{1,k}(x)dx \nonumber \\
& + \bigg( \frac{1}{2}-\theta_{k,-}^{(n,M)} \bigg)h_{1,k}(M) + \bigg( \frac{1}{2}-\theta_{k,+}^{(n,M)} \bigg)h_{1,k}(-M) + br_{k}^{2b}\int_{-M}^{M}h_{2,k}(x)dx, \label{def of G7pMp}
\end{align}
and
\begin{align*}
& h_{0,k}(x) = \begin{cases}
\log \big( \frac{1}{2}\mathrm{erfc} \big( -\frac{xr_{k}^{b}}{\sqrt{2}} \big) \big), & \mbox{if } k \in \{1,3,5,\ldots,2g-1\}, \\
\log \big( 1- \frac{1}{2}\mathrm{erfc} \big( -\frac{xr_{k}^{b}}{\sqrt{2}} \big) \big), & \mbox{if } k \in \{2,4,6,\ldots,2g\},
\end{cases} \\
& h_{1,k}(x) = \begin{cases}
\frac{e^{-\frac{x^{2}r_{k}^{2b}}{2}}}{\sqrt{2\pi}\frac{1}{2}\mathrm{erfc}\big(-\frac{x r_{k}^{b} }{\sqrt{2}}\big)}\Big( \frac{1}{3r_{k}^{b}}-\frac{5x^{2}r_{k}^{b}}{6} \Big), & \mbox{if } k \in \{1,3,5,\ldots,2g-1\}, \\[0.5cm]
\frac{-e^{-\frac{x^{2}r_{k}^{2b}}{2}}}{\sqrt{2\pi}\big(1-\frac{1}{2}\mathrm{erfc}\big(-\frac{x r_{k}^{b} }{\sqrt{2}}\big)\big)}\Big( \frac{1}{3r_{k}^{b}}-\frac{5x^{2}r_{k}^{b}}{6} \Big), & \mbox{if } k \in \{2,4,6,\ldots,2g\},
\end{cases} \\
& h_{2,k}(x) = \begin{cases}
\frac{e^{-\frac{x^{2}r_{k}^{2b}}{2}}}{\sqrt{2\pi}\mathrm{erfc}\big(-\frac{x r_{k}^{b} }{\sqrt{2}}\big)} \Big( \frac{-25}{36}x^{5}r_{k}^{3b} + \frac{73}{36}x^{3}r_{k}^{b}+\frac{x}{6r_{k}^{b}}\Big) - \frac{e^{- x^{2}r_{k}^{2b}}}{\pi \, \mathrm{erfc}^{2}\big(-\frac{x r_{k}^{b} }{\sqrt{2}}\big)} \Big( \frac{1}{3r_{k}^{b}}-\frac{5x^{2}r_{k}^{b}}{6} \Big)^{2}, \\
\frac{-e^{-\frac{x^{2}r_{k}^{2b}}{2}}}{2\sqrt{2\pi}\big(1-\frac{1}{2}\mathrm{erfc}\big(-\frac{x r_{k}^{b} }{\sqrt{2}}\big)\big)} \Big( \frac{-25}{36}x^{5}r_{k}^{3b} + \frac{73}{36}x^{3}r_{k}^{b}+\frac{x}{6r_{k}^{b}}\Big) - \frac{e^{- x^{2}r_{k}^{2b}}}{4\pi \big(1-\frac{1}{2}\mathrm{erfc}\big(-\frac{x r_{k}^{b} }{\sqrt{2}}\big)\big)^{2}} \Big( \frac{1}{3r_{k}^{b}}-\frac{5x^{2}r_{k}^{b}}{6} \Big)^{2}.
\end{cases}
\end{align*}
In the above equation for $h_{2,k}$, the first line reads for $k \in \{1,3,5,\ldots,2g-1\}$ and the second line reads for $k \in \{2,4,6,\ldots,2g\}$.
\end{lemma}
\begin{remark}
Note that the error term $\bigO(M^{9}n^{-1})$ above is indeed small as $n \to + \infty$, because $M=n^{\frac{1}{12}}$.
\end{remark}
\begin{proof}
We only do the proof for $k$ odd (the case of $k$ even is similar). For convenience, in this proof we will use $\lambda_{j}$, $\eta_{j}$ and $M_{j}$ in place of $\lambda_{j,k}$, $\eta_{j,k}$ and $M_{j,k}$. From \eqref{Temme exact formula}, \eqref{asymp prelim of S2kpvp} and \eqref{sums lambda j}, we see that
\begin{align}\label{lol38}
S_{2k}^{(2)} = \sum_{j:\lambda_{j,k}\in I_{2}}  \log \bigg( \frac{\gamma(a_{j},nr_{k}^{2b})}{\Gamma(\frac{j+\alpha}{b})} \bigg) = \sum_{j=g_{k,-}}^{g_{k,+}} \log \bigg( \frac{1}{2}\mathrm{erfc}\Big(-\eta_{j} \sqrt{a_{j}/2}\Big) - R_{a_{j}}(\eta_{j}) \bigg).
\end{align}
Recall from \eqref{lol55} that for all $j \in \{j:\lambda_{j}\in I_{2}\}$, we have 
\begin{align*}
1-\frac{M}{\sqrt{n}} \leq \lambda_{j} = \frac{bnr_{k}^{2b}}{j+\alpha} \leq 1+\frac{M}{\sqrt{n}}, \qquad -M \leq M_{j} \leq M.
\end{align*}
Hence, using \eqref{def of aj lambdajl etajl} we obtain 
\begin{align}
& \eta_{j} = (\lambda_{j}-1)\bigg( 1 - \frac{\lambda_{j}-1}{3} + \frac{7}{36}(\lambda_{j}-1)^{2} +\bigO((\lambda_{j}-1)^{3})\bigg) = \frac{M_{j}}{\sqrt{n}} - \frac{M_{j}^{2}}{3n} + \frac{7M_{j}^{3}}{36n^{3/2}} + \bigO\bigg(\frac{M^{4}}{n^{2}}\bigg), \nonumber \\
& -\eta_{j} \sqrt{a_{j}/2} = - \frac{M_{j} r_{k}^{b}}{\sqrt{2}} + \frac{5M_{j}^{2} r_{k}^{b}}{6\sqrt{2}\sqrt{n}} - \frac{53 M_{j}^{3} r_{k}^{b}}{72\sqrt{2}n} +\bigO(M^{4}n^{-\frac{3}{2}}), \label{asymp etaj and etajsqrtajover2} 
\end{align}
as $n \to + \infty$ uniformly for $j \in \{j:\lambda_{j}\in I_{2}\}$. By Taylor's theorem, for each $j \in \{j:\lambda_{j}\in I_{2}\}$ we have
\begin{align}
& \frac{1}{2}\mathrm{erfc}\Big(-\eta_{j} \sqrt{a_{j}/2}\Big) = \frac{1}{2}\mathrm{erfc}\Big(- \frac{M_{j} r_{k}^{b}}{\sqrt{2}}\Big)  + \frac{1}{2}\mathrm{erfc}'\Big(- \frac{M_{j} r_{k}^{b}}{\sqrt{2}}\Big) \bigg( -\eta_{j} \sqrt{a_{j}/2}+\frac{M_{j} r_{k}^{b}}{\sqrt{2}} \bigg) \nonumber \\
& + \frac{1}{4}\mathrm{erfc}''\Big(- \frac{M_{j} r_{k}^{b}}{\sqrt{2}}\Big) \bigg( -\eta_{j} \sqrt{a_{j}/2}+\frac{M_{j} r_{k}^{b}}{\sqrt{2}} \bigg)^{2} + \frac{1}{12}\mathrm{erfc}'''\big(\xi_{j}\big) \bigg( -\eta_{j} \sqrt{a_{j}/2}+\frac{M_{j} r_{k}^{b}}{\sqrt{2}} \bigg)^{3}, \label{lol34}
\end{align}
for a certain $\xi_{j} \in [- \frac{M_{j} r_{k}^{b}}{\sqrt{2}},-\eta_{j} \sqrt{a_{j}/2}]$. Using \eqref{def of erfc}, $\mathrm{erfc}'''(x) = \frac{4}{\sqrt{\pi}}(1-2x^{2})e^{-x^{2}}$ and \eqref{asymp etaj and etajsqrtajover2}, we infer that there exists a constant $C>0$ such that
\begin{align}
\Bigg| \frac{\frac{1}{12}\mathrm{erfc}'''\big(\xi_{j}\big) \big( -\eta_{j} \sqrt{a_{j}/2}+\frac{M_{j} r_{k}^{b}}{\sqrt{2}} \big)^{3}}{\frac{1}{2}\mathrm{erfc}\big(- \frac{M_{j} r_{k}^{b}}{\sqrt{2}}\big)} \Bigg| \leq C (1+M_{j}^{8})n^{-\frac{3}{2}} \label{lol35}
\end{align}
holds for all sufficiently large $n$ and all $j \in \{j:\lambda_{j}\in I_{2}\}$. Similarly, by Taylor's theorem, for each $j \in \{j:\lambda_{j}\in I_{2}\}$ we have
\begin{align}
R_{a_{j}}(\eta_{j}) = R_{a_{j}}(\tfrac{M_{j}}{\sqrt{n}}) + R_{a_{j}}'(\tfrac{M_{j}}{\sqrt{n}}) (\eta_{j}-\tfrac{M_{j}}{\sqrt{n}})+ \tfrac{1}{2}R_{a_{j}}''(\tilde{\xi}_{j}) (\eta_{j}-\tfrac{M_{j}}{\sqrt{n}})^{2}, \label{lol36}
\end{align}
for some $\tilde{\xi}_{j} \in [\eta_{j},\frac{M_{j}}{\sqrt{n}}]$. Furthermore, $R_{a}(\eta)$ is analytic with respect to $\lambda$ (see \cite[p. 285]{Temme}), in particular near $\lambda = 1$ (or $\eta=0$), and the expansion \eqref{asymp of Ra} holds in fact uniformly for $|\arg z | \leq 2\pi -\epsilon'$ for any $\epsilon'>0$ (see e.g. \cite[p. 325]{Paris}). It then follows from Cauchy's formula that \eqref{asymp of Ra} can be differentiated with respect to $\eta$ without increasing the error term. Thus, differentiating twice \eqref{asymp of Ra} we conclude that there exists $C>0$ such that
\begin{align}\label{lol37}
\Bigg|\frac{\frac{1}{2} R_{a_{j}}''(\tilde{\xi}_{j}) (\eta_{j}-\tfrac{M_{j}}{\sqrt{n}})^{2}}{\frac{1}{2}\mathrm{erfc}\big(- \frac{M_{j} r_{k}^{b}}{\sqrt{2}}\big)}\Bigg| \leq C (1+M_{j}^{6})n^{-\frac{3}{2}}
\end{align}
holds for all sufficiently large $n$ and all $j \in \{j:\lambda_{j}\in I_{2}\}$. Combining \eqref{lol38}, \eqref{asymp etaj and etajsqrtajover2}, \eqref{lol34}, \eqref{lol35}, \eqref{lol36} and \eqref{lol37} with \eqref{asymp of Ra} and \eqref{def of c0 and c1}, we obtain after a computation that
\begin{align}
S_{2k}^{(2)} & = \sum_{j=j_{k,-}}^{j_{k,+}} \log \bigg\{ \frac{1}{2}\mathrm{erfc}\Big(-\frac{M_{j} r_{k}^{b}}{\sqrt{2}}\Big)+\frac{1}{2}\mathrm{erfc}'\Big(-\frac{M_{j} r_{k}^{b}}{\sqrt{2}}\Big)\frac{5M_{j}^{2}r_{k}}{6\sqrt{2}\sqrt{n}} \nonumber \\
& + \frac{M_{j}^{3}}{288n}\bigg( 25 M_{j} r_{k}^{2b} \mathrm{erfc}''\Big( -\frac{M_{j} r_{k}^{b}}{\sqrt{2}} \Big) -53 \sqrt{2} r_{k}^{b}\mathrm{erfc}'\Big( -\frac{M_{j} r_{k}^{b}}{\sqrt{2}} \Big) \bigg) \nonumber \\
& +\frac{e^{-\frac{M_{j}^{2} r_{k}^{2b}}{2}}}{\sqrt{2\pi}}\bigg( \frac{1}{3 r_{k}^{b}\sqrt{n}} - \bigg[\frac{M_{j}}{12 r_{k}^{b}} + \frac{5M_{j}^{3} r_{k}^{b}}{18} \bigg]\frac{1}{n} \bigg) \bigg\} + \sum_{j=j_{k,-}}^{j_{k,+}}\bigO(M_{j}^{8}n^{-3/2}) \nonumber \\
& = \sum_{j=j_{k,-}}^{j_{k,+}} h_{0,k}(M_{j})+\frac{1}{\sqrt{n}} \sum_{j=j_{k,-}}^{j_{k,+}} h_{1,k}(M_{j})+\frac{1}{n} \sum_{j=j_{k,-}}^{j_{k,+}} h_{2,k}(M_{j}) + \bigO(M^{9}n^{-1}), \label{asymp of S2kp2p in proof}
\end{align}
as $n \to + \infty$. Each of these three sums can be expanded using Lemma \ref{lemma:Riemann sum}. The errors in these expansions can be estimated as follows. First, note that the function $e^{-\frac{x^{2}r_{k}^{2b}}{2}} \mathrm{erfc}\big(\frac{-x r_{k}^{b} }{\sqrt{2}}\big)^{-1}$ is exponentially small as $x \to + \infty$, and is bounded by a polynomial of degree $1$ as $x \to - \infty$. Hence, the functions $h_{0,k}(x)$, $h_{1,k}(x)$ and $h_{2,k}(x)$ also tend to $0$ exponentially fast as $x \to + \infty$, while as $x \to -\infty$ they are bounded by polynomials of degree $2$, $3$ and $6$, respectively. The derivatives of $h_{0,k}(x)$, $h_{1,k}(x)$ and $h_{2,k}(x)$ can be estimated similarly. Using Lemma \ref{lemma:Riemann sum}, we then find that the fourth term in the large $n$ asymptotics of $\sum_{j=j_{k,-}}^{j_{k,+}} h_{0,k}(M_{j})$ is 
\begin{align*}
& \bigO\bigg( \frac{1}{n^{3/2}} \sum_{j=g_{k,-}+1}^{g_{k,+}} \bigg( (1+|M_{j}|^{3}) \tilde{\mathfrak{m}}_{j,n}(h_{0,k}) + (1+M_{j}^{2})\tilde{\mathfrak{m}}_{j,n}(h_{0,k}') + (1+|M_{j}|) \tilde{\mathfrak{m}}_{j,n}(h_{0,k}'') + \tilde{\mathfrak{m}}_{j,n}(h_{0,k}''') \bigg) \\
& = \bigO\bigg(\frac{M^{6}}{n} \bigg), \qquad \mbox{as } n \to + \infty.
\end{align*}
Similarly, the third term in the asymptotics of $\frac{1}{\sqrt{n}} \sum_{j=j_{k,-}}^{j_{k,+}} h_{1,k}(M_{j})$ is $\bigO(\frac{M^{6}}{n})$, and the second term in the asymptotics of $\frac{1}{n} \sum_{j=j_{k,-}}^{j_{k,+}} h_{2,k}(M_{j})$ is $\bigO(\frac{M^{8}}{n})$. All these errors are, in particular, $\bigO(M^{9}n^{-1})$. Hence, by substituting the asymptotics of these three sums in \eqref{asymp of S2kp2p in proof}, we find the claim.
\end{proof}
The quantities $G_{4,k}^{(M)}$, $G_{6,k}^{(M)}$ and $G_{7,k}^{(M)}$ appearing in \eqref{lol39} depend quite complicatedly on $M$. The goal of the following lemma is to find more explicit asymptotics for $S_{2k}^{(2)}$. We can do that at the cost of introducing a new type of error terms. Indeed, the error $\bigO(M^{9}n^{-1})$ of \eqref{lol39} is an error that only restrict $M$ to be ``not too large". In Lemma \ref{lemma:S2ktilde large M} below, there is another kind of error term that restrict $M$ to be ``not too small". 
\begin{lemma}\label{lemma:S2ktilde large M}
Let $k \in \{1,3,\ldots,2g-1\}$. As $n \to + \infty$, we have
\begin{align*}
& S_{2k}^{(2)} = \widetilde{G}_{4,k}^{(M)} \sqrt{n} + \widetilde{G}_{6,k}^{(M)}  + \bigO\bigg(\frac{M^{5}}{\sqrt{n}}\bigg)+\bigO\bigg(\frac{\sqrt{n}}{M^{7}}\bigg),
\end{align*}
where
\begin{align}
& \widetilde{G}_{4,k}^{(M)} = -\frac{br_{k}^{4b}}{6}M^{3} - br_{k}^{2b} M \log M + br_{k}^{2b}\big( 1-\log(r_{k}^{b}\sqrt{2\pi}) \big)M \nonumber \\
& + br_{k}^{2b}\int_{-\infty}^{-1}\bigg( h_{0,k}(x)-\bigg[ -\frac{r_{k}^{2b}}{2}x^{2}-\log(-x)-\log(r_{k}^{b}\sqrt{2\pi}) \bigg] \bigg)dx + br_{k}^{2b} \int_{-1}^{+\infty}h_{0,k}(x)dx \nonumber \\
& + br_{k}^{2b}\Big( \log(r_{k}^{b}\sqrt{2\pi})-1+\frac{r_{k}^{2b}}{6} \Big) + \frac{b}{M} - \frac{5b}{6r_{k}^{2b}M^{3}} + \frac{37b}{15r_{k}^{4b}M^{5}}, \label{def of G4tilde for k odd} \\
& \widetilde{G}_{6,k}^{(M)} = - \frac{11}{24}br_{k}^{4b}M^{4}  - br_{k}^{2b} M^{2}\log M + r_{k}^{2b} \bigg( \frac{b}{4} - b \log (r_{k}^{b}\sqrt{2\pi}) + \frac{2\theta_{k,+}^{(n,M)}-1}{4} \bigg)M^{2} \nonumber \\
& + \frac{2\theta_{k,+}^{(n,M)}-1}{2} \log M - \frac{br_{k}^{2b}}{4}  -2br_{k}^{2b}\int_{-1}^{+\infty} xh_{0,k}(x)dx \nonumber \\
& -2br_{k}^{2b}\int_{-\infty}^{-1} \bigg\{ xh_{0,k}(x) - \bigg( -\frac{r_{k}^{2b}}{2}x^{3} - x \log(-x) - \log(r_{k}^{b}\sqrt{2\pi})x - \frac{1}{r_{k}^{2b}x} \bigg) \bigg\} dx \nonumber \\
& + br_{k}^{2b} \int_{-\infty}^{-1} \bigg\{ h_{1,k}(x)-\bigg( \frac{5}{6}r_{k}^{2b}x^{3}+\frac{x}{2}- \frac{2}{r_{k}^{2b}x} \bigg) \bigg\}dx + br_{k}^{2b} \int_{-1}^{+\infty} h_{1,k}(x)dx \nonumber \\
& + \bigg(br_{k}^{2b} + \frac{2\theta_{k,+}^{(n,M)}-1}{2} \bigg) \log(r_{k}^{b}\sqrt{2\pi}) + \frac{11}{24}br_{k}^{4b}. \label{def of G6tilde for k odd}
\end{align}
Let $k \in \{2,4,\ldots,2g\}$. As $n \to + \infty$, we have
\begin{align*}
& S_{2k}^{(2)} = \widetilde{G}_{4,k}^{(M)} \sqrt{n} + \widetilde{G}_{6,k}^{(M)} + \bigO\bigg(\frac{M^{5}}{\sqrt{n}}\bigg)+\bigO\bigg(\frac{\sqrt{n}}{M^{7}}\bigg),
\end{align*}
where
\begin{align*}
& \widetilde{G}_{4,k}^{(M)} = -\frac{br_{k}^{4b}}{6}M^{3} - br_{k}^{2b} M \log M + br_{k}^{2b}\big( 1-\log(r_{k}^{b}\sqrt{2\pi}) \big)M \\
& + br_{k}^{2b}\int_{-\infty}^{1} h_{0,k}(x) dx + br_{k}^{2b} \int_{1}^{+\infty} \bigg( h_{0,k}(x) -\bigg[ -\frac{r_{k}^{2b}}{2}x^{2}-\log(x)-\log(r_{k}^{b}\sqrt{2\pi}) \bigg]\bigg) dx \\
& + br_{k}^{2b}\Big( \log(r_{k}^{b}\sqrt{2\pi})-1+\frac{r_{k}^{2b}}{6} \Big) + \frac{b}{M} - \frac{5b}{6r_{k}^{2b}M^{3}} + \frac{37b}{15r_{k}^{4b}M^{5}}, \\
& \widetilde{G}_{6,k}^{(M)} = \frac{11}{24}br_{k}^{4b}M^{4} + br_{k}^{2b} M^{2}\log M + r_{k}^{2b} \bigg( -\frac{b}{4} + b \log (r_{k}^{b}\sqrt{2\pi}) + \frac{2\theta_{k,+}^{(n,M)}-1}{4} \bigg)M^{2} \\
& + \frac{2\theta_{k,+}^{(n,M)}-1}{2} \log M + \frac{br_{k}^{2b}}{4} -2br_{k}^{2b}\int_{-\infty}^{1}  xh_{0,k}(x)  dx \\
& -2br_{k}^{2b}\int_{1}^{+\infty} \bigg\{ xh_{0,k}(x) - \bigg( -\frac{r_{k}^{2b}}{2}x^{3} - x \log x - \log(r_{k}^{b}\sqrt{2\pi})x - \frac{1}{r_{k}^{2b}x} \bigg) \bigg\}dx \\
& + br_{k}^{2b} \int_{-\infty}^{1} h_{1,k}(x) dx + br_{k}^{2b} \int_{1}^{+\infty} \bigg\{ h_{1,k}(x) - \bigg( \frac{5}{6}r_{k}^{2b}x^{3}+\frac{x}{2}- \frac{2}{r_{k}^{2b}x} \bigg) \bigg\}dx \\
& + \bigg(-br_{k}^{2b} + \frac{2\theta_{k,+}^{(n,M)}-1}{2} \bigg) \log(r_{k}^{b}\sqrt{2\pi}) - \frac{11}{24}br_{k}^{4b}.
\end{align*}
\end{lemma}
\begin{remark}
With our choice $M=n^{\frac{1}{12}}$, both errors $\bigO\big(\frac{M^{5}}{\sqrt{n}}\big)$ and $\bigO\big(\frac{\sqrt{n}}{M^{7}}\big)$ are of the same order:
\begin{align*}
\frac{M^{5}}{\sqrt{n}} = \frac{\sqrt{n}}{M^{7}} = M^{-1} = n^{-\frac{1}{12}}.
\end{align*}
So $M=n^{\frac{1}{12}}$ is the choice that produces the best control of the total error. However this still does not really explain why we chose $M=n^{\frac{1}{12}}$. Indeed, in the above asymptotics one could have easily computed the next term of order $\frac{\sqrt{n}}{M^{7}}$ if this was needed. The real reason why we chose $M=n^{\frac{1}{12}}$ is because the sums $S_{2k}^{(1)}$ and $S_{2k}^{(3)}$, which are analyzed below, also contain a term of order $\frac{\sqrt{n}}{M^{7}}$ in their asymptotics, and this term is hard to compute explicitly. 
\end{remark}
\begin{proof}
We only do the proof for $k$ odd. As already mentioned in the proof of Lemma \ref{lemma:S2ktilde}, $h_{0,k}(x)$, $h_{1,k}(x)$ and $h_{2,k}(x)$ are exponentially small as $x \to + \infty$, and since $M=n^{\frac{1}{12}}$, this implies that there exists $c>0$ such that
\begin{align*}
& \int_{-1}^{M} x^{\ell}h_{j,k}(x)dx = \int_{-1}^{+\infty} x^{\ell}h_{j,k}(x)dx + \bigO(e^{-n^{c}}), \qquad \mbox{as } n \to + \infty, \quad j=0,1,2, \; \ell = 0,1,2.
\end{align*}
On the other hand, as $x \to -\infty$, we have
\begin{align}
& h_{0,k}(x) = -\frac{r_{k}^{2b}}{2}x^{2}-\log(-x)-\log\big( r_{k}^{b}\sqrt{2\pi} \big) - \frac{r_{k}^{-2b}}{x^{2}} + \frac{5r_{k}^{-4b}}{2x^{4}} - \frac{37r_{k}^{-6b}}{3x^{6}} + \bigO(x^{-8}), \label{lol40} \\
& h_{1,k}(x) = \frac{5}{6}r_{k}^{2b}x^{3}+\frac{x}{2}-\frac{2r_{k}^{-2b}}{x} + \bigO(x^{-3}), \\
& h_{2,k}(x) = \bigO(x^{4}). \label{lol41}
\end{align}
Using \eqref{lol40}--\eqref{lol41} and the definitions \eqref{def of G4pMp}--\eqref{def of G7pMp} of $G_{4,k}^{(M)}$, $G_{6,k}^{(M)}$, $G_{7,k}^{(M)}$, we obtain that
\begin{align*}
& G_{4,k}^{(M)}\sqrt{n} = \widetilde{G}_{4,k}^{(M)}\sqrt{n} + \bigO\bigg(\frac{\sqrt{n}}{M^{7}}\bigg), \\
& G_{6,k}^{(M)} = \widetilde{G}_{6,k}^{(M)} + \bigO\bigg( \frac{1}{M^{2}} \bigg), \\
& \frac{G_{7,k}^{(M)}}{\sqrt{n}} = \bigO\bigg(\frac{M^{5}}{\sqrt{n}}\bigg), 
\end{align*}
as $n \to + \infty$, and the claim follows.
\end{proof}
\begin{remark}
From the above proof, we see that $\frac{G_{7,k}^{(M)}}{\sqrt{n}} = \bigO\big(\frac{M^{5}}{\sqrt{n}}\big)$ as $n  \to + \infty$, and therefore $\smash{G_{7,k}^{(M)}}$ will not contribute at all in our final answer. It was however very important to compute $\smash{G_{7,k}^{(M)}}$ explicitly. Indeed, as can be seen from the statement of Lemma \ref{lemma:S2ktilde}, $h_{2,k}$ consists of two parts, and it is easy to check that each of these two parts is of order $x^{6}$ as $x \to -\infty$. Thus the fact that actually we have $h_{2,k}(x)=\bigO(x^{4})$ (see \eqref{lol41}) means that there are great cancellations in the asymptotic behavior of these two parts, and this is not something one could have detected in advance without computing explicitly $\smash{G_{7,k}^{(M)}}$ and $h_{2,k}$.
\end{remark}
Now we rewrite the integrals in $\widetilde{G}_{4,k}^{(M)}$ and $\widetilde{G}_{6,k}^{(M)}$ more elegantly. 
\begin{lemma}\label{lemma:S2ktilde large M and nice integrals}
The quantities $\widetilde{G}_{4,k}^{(M)}$ and $\widetilde{G}_{6,k}^{(M)}$ appearing in Lemma \ref{lemma:S2ktilde large M} can be rewritten more explicitly as follows. For $k \in \{1,3,\ldots,2g-1\}$, we have
\begin{align*}
& \widetilde{G}_{4,k}^{(M)} = -\frac{br_{k}^{4b}}{6}M^{3} - br_{k}^{2b} M \log M + br_{k}^{2b}\big( 1-\log(r_{k}^{b}\sqrt{2\pi}) \big)M + \sqrt{2}br_{k}^{b} \int_{-\infty}^{0}\log \bigg( \frac{1}{2}\mathrm{erfc}(y) \bigg)dy \\
& + \sqrt{2}br_{k}^{b}\int_{0}^{+\infty} \bigg[ \log \bigg( \frac{1}{2}\mathrm{erfc}(y) \bigg)+y^{2}+\log y + \log(2\sqrt{\pi}) \bigg]dy + \frac{b}{M} - \frac{5b}{6r_{k}^{2b}M^{3}} + \frac{37b}{15r_{k}^{4b}M^{5}},   \\
& \widetilde{G}_{6,k}^{(M)} = - \frac{11}{24}br_{k}^{4b}M^{4} - br_{k}^{2b} M^{2}\log M + r_{k}^{2b} \bigg( \frac{b}{4} - b \log (r_{k}^{b}\sqrt{2\pi}) + \frac{2\theta_{k,+}^{(n,M)}-1}{4} \bigg)M^{2}   \\
& + \frac{2\theta_{k,+}^{(n,M)}-1}{2} \log \Big(Mr_{k}^{b}\sqrt{2\pi}\Big)+2b \int_{-\infty}^{0}\bigg\{ 2y \log \bigg(\frac{1}{2}\mathrm{erfc}(y) \bigg) + \frac{e^{-y^{2}}}{\sqrt{\pi}\mathrm{erfc}(y)}\frac{1-5y^{2}}{3} \bigg\}dy \\
& +2b \int_{0}^{+\infty} \bigg\{ 2y \log \bigg(\frac{1}{2}\mathrm{erfc}(y) \bigg) + \frac{e^{-y^{2}}}{\sqrt{\pi}\mathrm{erfc}(y)}\frac{1-5y^{2}}{3} + \frac{11}{3}y^{3} + 2y\log y + \bigg( \frac{1}{2}+2\log(2\sqrt{\pi}) \bigg)y \bigg\}dy,
\end{align*}
and for $k \in \{2,4,\ldots,2g\}$, we have
\begin{align*}
& \widetilde{G}_{4,k}^{(M)} = -\frac{br_{k}^{4b}}{6}M^{3} - br_{k}^{2b} M \log M + br_{k}^{2b}\big( 1-\log(r_{k}^{b}\sqrt{2\pi}) \big)M + \sqrt{2}br_{k}^{b} \int_{-\infty}^{0}\log \bigg( \frac{1}{2}\mathrm{erfc}(y) \bigg)dy \\
& + \sqrt{2}br_{k}^{b}\int_{0}^{+\infty} \bigg[ \log \bigg( \frac{1}{2}\mathrm{erfc}(y) \bigg)+y^{2}+\log y + \log(2\sqrt{\pi}) \bigg]dy + \frac{b}{M} - \frac{5b}{6r_{k}^{2b}M^{3}} + \frac{37b}{15r_{k}^{4b}M^{5}},  \\
& \widetilde{G}_{6,k}^{(M)} =  \frac{11}{24}br_{k}^{4b}M^{4} + br_{k}^{2b} M^{2}\log M + r_{k}^{2b} \bigg( -\frac{b}{4} + b \log (r_{k}^{b}\sqrt{2\pi}) + \frac{2\theta_{k,-}^{(n,M)}-1}{4} \bigg)M^{2}   \\
& + \frac{2\theta_{k,-}^{(n,M)}-1}{2} \log \Big(Mr_{k}^{b}\sqrt{2\pi}\Big)-2b \int_{-\infty}^{0}\bigg\{ 2y \log \bigg(\frac{1}{2}\mathrm{erfc}(y) \bigg) + \frac{e^{-y^{2}}}{\sqrt{\pi}\mathrm{erfc}(y)}\frac{1-5y^{2}}{3} \bigg\}dy \\
& -2b \int_{0}^{+\infty} \bigg\{ 2y \log \bigg(\frac{1}{2}\mathrm{erfc}(y) \bigg) + \frac{e^{-y^{2}}}{\sqrt{\pi}\mathrm{erfc}(y)}\frac{1-5y^{2}}{3} + \frac{11}{3}y^{3} + 2y\log y + \bigg( \frac{1}{2}+2\log(2\sqrt{\pi}) \bigg)y \bigg\}dy.
\end{align*}
\end{lemma}
\begin{proof}
For $k \in \{1,3,\ldots,2g-1\}$, using the change of variables $y = -\frac{xr_{k}^{b}}{\sqrt{2}}$ and splitting the two integrals around $0$ instead of $-1$, we obtain
\begin{align*}
& \int_{-\infty}^{-1}\bigg( h_{0,k}(x)-\bigg[ -\frac{r_{k}^{2b}}{2}x^{2}-\log(-x)-\log(r_{k}^{b}\sqrt{2\pi}) \bigg] \bigg)dx +  \int_{-1}^{+\infty}h_{0,k}(x)dx = 1-\frac{r_{k}^{2b}}{6} - \log(r_{k}^{b}\sqrt{2\pi}) \\
& + \frac{\sqrt{2}}{r_{k}^{b}} \bigg\{ \int_{-\infty}^{0}\log \bigg( \frac{1}{2}\mathrm{erfc}(y) \bigg)dy + \int_{0}^{+\infty} \bigg[ \log \bigg( \frac{1}{2}\mathrm{erfc}(y) \bigg)+y^{2}+\log y + \log(2\sqrt{\pi}) \bigg]dy \bigg\}.
\end{align*}
It is also possible to split the integrals in $\widetilde{G}_{6,k}^{(M)}$ around $0$ by regrouping some integrands as follows
\begin{align*}
& -2\int_{-\infty}^{-1} \bigg\{ xh_{0,k}(x) - \bigg( -\frac{r_{k}^{2b}}{2}x^{3} - x \log(-x) - \log(r_{k}^{b}\sqrt{2\pi})x - \frac{1}{r_{k}^{2b}x} \bigg) \bigg\} dx -2\int_{-1}^{+\infty} xh_{0,k}(x)dx \\
& + \int_{-\infty}^{-1} \bigg\{ h_{1,k}(x)-\bigg( \frac{5}{6}r_{k}^{2b}x^{3}+\frac{x}{2}- \frac{2}{r_{k}^{2b}x} \bigg) \bigg\}dx +  \int_{-1}^{+\infty} h_{1,k}(x)dx = \frac{1}{4} - \frac{11r_{k}^{2b}}{24} -\log \big( r_{k}^{b}\sqrt{2\pi} \big) \\
& + \frac{2}{r_{k}^{2b}}\int_{-\infty}^{0} \bigg\{ 2y \log \bigg( \frac{1}{2}\mathrm{erfc}(y)\bigg) + \frac{e^{-y^{2}}}{\sqrt{\pi}\mathrm{erfc}(y)}\frac{1-5y^{2}}{3}  \bigg\}dy \\
& + \frac{2}{r_{k}^{2b}}\int_{0}^{+\infty} \bigg\{ 2y \log \bigg( \frac{1}{2}\mathrm{erfc}(y) \bigg) + \frac{e^{-y^{2}}}{\sqrt{\pi}\mathrm{erfc}(y)}\frac{1-5y^{2}}{3} + \frac{11}{3}y^{3} + 2y \log (y) + \bigg( \frac{1}{2}+2\log(2\sqrt{\pi}) \bigg)y  \bigg\}dy.
\end{align*}
We obtain the claim for $k$ odd after substituting these expressions in the definitions \eqref{def of G4tilde for k odd} and \eqref{def of G6tilde for k odd} of $\widetilde{G}_{4,k}^{(M)}$ and $\widetilde{G}_{6,k}^{(M)}$. The proof for $k \in \{2,4,\ldots,2g\}$ is similar (it requires furthermore the use of the functional relation $\mathrm{erfc}(y)+\mathrm{erfc}(-y)=2$), and we omit the details.
\end{proof}
Now we turn our attention to the sums $S_{2k}^{(3)}$ and $S_{2k}^{(1)}$. The analogue of these sums in \cite{Charlier 2d jumps} were relatively simple to analyze, see \cite[Lemmas 2.5 and 2.6]{Charlier 2d jumps}. In this paper, the sums $\{S_{2k}^{(3)}\}_{k \,\mathrm{odd}}$ and $\{S_{2k}^{(1)}\}_{k \,\mathrm{even}}$ are straightforward to analyze (and even simpler than in \cite[Lemmas 2.5 and 2.6]{Charlier 2d jumps}). However, the sums $\{S_{2k}^{(1)}\}_{k \,\mathrm{odd}}$ and $\{S_{2k}^{(3)}\}_{k \,\mathrm{even}}$ are challenging (their large $n$ asymptotics depend on both $\epsilon$ and $M$ in a complicated way). We start with the sums $\{S_{2k}^{(3)}\}_{k \,\mathrm{odd}}$ and $\{S_{2k}^{(1)}\}_{k \,\mathrm{even}}$.
\begin{lemma}\label{lemma:S2k p3p and p1p: the easy ones}
Let $k \in \{1,3,\ldots,2g-1\}$. As $n \to + \infty$, we have
\begin{align*}
S_{2k}^{(3)}=\bigO(n^{-10}).
\end{align*}
Let $k \in \{2,4,\ldots,2g\}$. As $n \to + \infty$, we have
\begin{align*}
S_{2k}^{(1)}=\bigO(n^{-10}).
\end{align*}
\end{lemma}
\begin{proof}
Let $k \in \{1,3,\ldots,2g-1\}$. Recall from \eqref{asymp prelim of S2kpvp} that 
\begin{align*}
S_{2k}^{(3)} = \sum_{j:\lambda_{j,k}\in I_{3}}  \log \bigg( \frac{\gamma(a_{j},nr_{k}^{2b})}{\Gamma(a_{j})} \bigg),
\end{align*}
and from \eqref{def of the three intervals} that $I_{3} = (1+\tfrac{M}{\sqrt{n}},1+\epsilon]$. We then infer, by \eqref{def of aj lambdajl etajl}, that there exists a constant $c>0$ such that $\sqrt{a_{j}}\eta_{j,k} \geq cM$ holds for all large $n$ and $j \in \{j:\lambda_{j,k}\in I_{3}\}$. By \eqref{Temme exact formula}, \eqref{asymp of Ra}, \eqref{large y asymp of erfc} and $\mathrm{erfc}(-y) = 2-\mathrm{erfc}(y)$, this implies
\begin{align*}
\frac{\gamma(a_{j},nr_{k}^{2b})}{\Gamma(a_{j})} = \frac{1}{2}\mathrm{erfc}\Big(-\eta_{j,k} \sqrt{a_{j}/2}\Big) - R_{a_{j}}(\eta_{j,k}) = 1+\bigO(e^{-\frac{c^{2}}{2}M^{2}}), \qquad \mbox{as } n \to + \infty
\end{align*}
uniformly for $j \in \{j:\lambda_{j,k}\in I_{3}\}$. Since $M=n^{\frac{1}{12}}$, the claim is proved for $k$ odd. The proof for $k$ even is similar and we omit it.
\end{proof}
We now focus on $\{S_{2k}^{(1)}\}_{k \,\mathrm{odd}}$ and $\{S_{2k}^{(3)}\}_{k \,\mathrm{even}}$.
\begin{lemma}\label{lemma: exact decomposition for the interpolating sums}
Let $k \in \{1,3,\ldots,2g-1\}$. We have
\begin{align}\label{lol52}
& S_{2k}^{(1)} = \sum_{j=g_{k,+}+1}^{j_{k,+}} \log \bigg\{ \frac{1}{2}\mathrm{erfc} \Big( -\frac{\eta_{j,k}}{\sqrt{2}}\sqrt{a_{j}} \Big) \bigg\} + \sum_{j=g_{k,+}+1}^{j_{k,+}} \log \Bigg\{ 1-\frac{R_{a_{j}}(\eta_{j,k})}{\frac{1}{2}\mathrm{erfc} \big( -\frac{\eta_{j,k}}{\sqrt{2}}\sqrt{a_{j}} \big)} \Bigg\},
\end{align}
where
\begin{align}
& -\frac{\eta_{j,k}}{\sqrt{2}} = \sqrt{\frac{br_{k}^{2b}}{j/n+\frac{\alpha}{n}}-1-\log \bigg( \frac{br_{k}^{2b}}{j/n+\frac{\alpha}{n}} \bigg)}, \qquad \sqrt{a_{j}} = \frac{\sqrt{n}}{\sqrt{b}}\sqrt{j/n+\frac{\alpha}{n}}, \label{lol42} \\
& R_{a_{j}}(\eta_{j,k}) = \frac{\exp\big( -\frac{a_{j}\eta_{j,k}^{2}}{2} \big) }{\sqrt{n}\sqrt{2\pi b^{-1}}\sqrt{j/n+\frac{\alpha}{n}}} \Bigg\{ c_{0}( \eta_{j,k} ) + \frac{b \; c_{1}( \eta_{j,k} )}{n(j/n+\frac{\alpha}{n})}  + \bigO(n^{-2}) \Bigg\}, \label{lol48}
\end{align}
and the last expansion holds as $n \to + \infty$ uniformly for $j \in \{g_{k,+}+1,\ldots,j_{k,+}\}$. We recall that the functions $c_{0}$ and $c_{1}$ are defined in \eqref{def of c0 and c1}.

\medskip \noindent Let $k \in \{2,4,\ldots,2g\}$. We have
\begin{align}\label{lol53}
& S_{2k}^{(3)} = \sum_{j=j_{k,-}}^{g_{k,-}-1} \log \bigg\{1- \frac{1}{2}\mathrm{erfc} \Big( -\frac{\eta_{j,k}}{\sqrt{2}}\sqrt{a_{j}} \Big) \bigg\} + \sum_{j=j_{k,-}}^{g_{k,-}-1} \log \Bigg\{ 1+\frac{R_{a_{j}}(\eta_{j,k})}{1- \frac{1}{2}\mathrm{erfc} \big( -\frac{\eta_{j,k}}{\sqrt{2}}\sqrt{a_{j}} \big)} \Bigg\},
\end{align}
where
\begin{align*}
& -\frac{\eta_{j,k}}{\sqrt{2}} = -\sqrt{\frac{br_{k}^{2b}}{j/n+\frac{\alpha}{n}}-1-\log \bigg( \frac{br_{k}^{2b}}{j/n+\frac{\alpha}{n}} \bigg)}, \qquad \sqrt{a_{j}} = \frac{\sqrt{n}}{\sqrt{b}}\sqrt{j/n+\frac{\alpha}{n}}, \\
& R_{a_{j}}(\eta_{j,k}) = \frac{\exp\Big( -\frac{a_{j}\eta_{j,k}^{2}}{2} \big) \Big)}{\sqrt{n}\sqrt{2\pi b^{-1}}\sqrt{j/n+\frac{\alpha}{n}}} \Bigg\{ c_{0}( \eta_{j,k} ) + \frac{b \; c_{1}( \eta_{j,k} )}{n(j/n+\frac{\alpha}{n})}  + \bigO(n^{-2}) \Bigg\},
\end{align*}
and the last expansion holds as $n \to + \infty$ uniformly for $j \in \{j_{k,-},\ldots,g_{k,-}-1\}$.
\end{lemma}
\begin{proof}
This follows from a direct application of Lemma \ref{lemma: uniform}. 
\end{proof}
The asymptotic analysis of $\{S_{2k}^{(1)}\}_{k \,\mathrm{odd}}$ and $\{S_{2k}^{(3)}\}_{k \,\mathrm{even}}$ is challenging partly because, as can be seen from the statement of Lemma \ref{lemma: exact decomposition for the interpolating sums}, there are four types of $n$-dependent parameters which vary at different speeds. Indeed, as $n \to + \infty$ and $j \in \{g_{k,+}+1,\ldots,j_{k,+}\}$, the quantities $\sqrt{a_{j}}$, $\eta_{j,k}$, $j/n$ and $\alpha/n$ are of orders $\sqrt{n}$, $j/n-br_{k}^{2b}$, $1$ and $\frac{1}{n}$ respectively. In particular, for $j$ close to $g_{k,+}+1$, $\eta_{j,k}$ is of order $\frac{M}{\sqrt{n}}$, while for $j$ close to $j_{k,+}$, it is of order $1$. In the next lemma, we obtain asymptotics for the right-hand sides of \eqref{lol52} and \eqref{lol53}. These asymptotics will then be evaluated more explicitly using Lemma \ref{lemma:Riemann sum NEW}. 

\begin{lemma}\label{lemma:preliminary expansions}
Let $k \in \{1,3,\ldots,2g-1\}$. As $n \to + \infty$, we have
\begin{align}
& \sum_{j=g_{k,+}+1}^{j_{k,+}} \log \bigg\{ \frac{1}{2}\mathrm{erfc} \Big( -\eta_{j,k}\sqrt{\tfrac{a_{j}}{2}} \Big) \bigg\} = n\sum_{j=g_{k,+}+1}^{j_{k,+}} \mathfrak{g}_{k,1}(j/n) + \log n \sum_{j=g_{k,+}+1}^{j_{k,+}} \mathfrak{g}_{k,2}(j/n) +\sum_{j=g_{k,+}+1}^{j_{k,+}} \mathfrak{g}_{k,3}(j/n) \nonumber \\
&  + \frac{1}{n}\sum_{j=g_{k,+}+1}^{j_{k,+}} \mathfrak{g}_{k,4}(j/n)+ \frac{1}{n^{2}}\sum_{j=g_{k,+}+1}^{j_{k,+}} \mathfrak{g}_{k,5}(j/n)+ \frac{1}{n^{3}}\sum_{j=g_{k,+}+1}^{j_{k,+}} \mathfrak{g}_{k,6}(j/n) + \bigO\bigg(\frac{\sqrt{n}}{M^{7}}\bigg), \label{lol46} \\
& \sum_{j=g_{k,+}+1}^{j_{k,+}} \log \Bigg\{ 1-\frac{R_{a_{j}}(\eta_{j,k})}{\frac{1}{2}\mathrm{erfc} \big( -\eta_{j,k}\sqrt{\frac{a_{j}}{2}} \big)} \Bigg\} = \sum_{j=g_{k,+}+1}^{j_{k,+}} \mathfrak{h}_{k,3}(j/n) + \frac{1}{n}\sum_{j=g_{k,+}+1}^{j_{k,+}} \mathfrak{h}_{k,4}(j/n)+ \bigO(M^{-2}), \label{lol47}
\end{align}
where
\begin{align}
& \mathfrak{g}_{k,1}(x) = - \frac{br_{k}^{2b}-x-x\log (\frac{br_{k}^{2b}}{x})}{b}, \label{def of gk1} \\
& \mathfrak{g}_{k,2}(x) = - \frac{1}{2}, \label{def of gk2} \\
& \mathfrak{g}_{k,3}(x) = \frac{1}{2}\log \bigg( \frac{b}{4\pi} \bigg) - \frac{1}{2}\log \bigg( br_{k}^{2b}-x-x\log \bigg( \frac{br_{k}^{2b}}{x} \bigg) \bigg) + \frac{\alpha}{b}\log \bigg( \frac{br_{k}^{2b}}{x} \bigg), \label{def of gk3} \\
& \mathfrak{g}_{k,4}(x) = -\frac{1}{2} \frac{b}{br_{k}^{2b}-x-x\log(\frac{br_{k}^{2b}}{x})} + \frac{1}{2}\frac{\alpha \log \big( \frac{br_{k}^{2b}}{x} \big)}{br_{k}^{2b}-x-x\log(\frac{b r_{k}^{2b}}{x})} - \frac{\alpha^{2}}{2bx}, \label{def of gk4} \\
& \mathfrak{g}_{k,5}(x) = \frac{5b^{2}}{8(br_{k}^{2b}-x-x\log(\frac{br_{k}^{2b}}{x}))^{2}} - \frac{b \alpha \log(\frac{br_{k}^{2b}}{x})}{2(br_{k}^{2b}-x-x\log(\frac{br_{k}^{2b}}{x}))^{2}} \nonumber \\
& \hspace{1.5cm} + \alpha^{2} \frac{-br_{k}^{2b}+x+x\log(\frac{br_{k}^{2b}}{x})+x \log^{2}(\frac{br_{k}^{2b}}{x})}{4x(br_{k}^{2b}-x-x\log(\frac{br_{k}^{2b}}{x}))^{2}} + \frac{\alpha^{3}}{6bx^{2}}, \label{def of gk5} \\
& \mathfrak{g}_{k,6}(x) = \frac{-37b^{3}}{24(br_{k}^{2b}-x-x\log(\frac{br_{k}^{2b}}{x}))^{3}} + \frac{5b^{2}\alpha \log(\frac{b r_{k}^{2b}}{x})}{4(br_{k}^{2b}-x-x\log(\frac{br_{k}^{2b}}{x}))^{3}} \nonumber \\
& + \frac{b \alpha^{2} \big( br_{k}^{2b}-x-x\log(\frac{br_{k}^{2b}}{x}) - 2x \log^{2}(\frac{br_{k}^{2b}}{x}) \big)}{4x(br_{k}^{2b}-x-x\log(\frac{br_{k}^{2b}}{x}))^{3}} \nonumber \\
& + \alpha^{3}\frac{(x-br_{k}^{2b})^{2}+5x(x-br_{k}^{2b})\log(\frac{br_{k}^{2b}}{x})+4x^{2}\log^{2}(\frac{br_{k}^{2b}}{x})+2x^{2}\log^{3}(\frac{br_{k}^{2b}}{x})}{12x^{2}(br_{k}^{2b}-x-x\log(\frac{br_{k}^{2b}}{x}))^{3}} - \frac{\alpha^{4}}{12bx^{3}}, \label{def of gk6} \\
& \mathfrak{h}_{k,3}(x) = \log \bigg( \sqrt{2x}\frac{\sqrt{br_{k}^{2b}-x-x\log(\frac{br_{k}^{2b}}{x})}}{|x-br_{k}^{2b}|} \bigg), \label{def of hk3} \\
& \mathfrak{h}_{k,4}(x) = \frac{-b(b^{2}r_{k}^{4b}+10b r_{k}^{2b}x+x^{2})}{12x(br_{k}^{2b}-x)^{2}} + \frac{br_{k}^{2b}\alpha}{x(br_{k}^{2b}-x)} + \frac{1}{2x}\frac{x b+(x-br_{k}^{2b})\alpha}{br_{k}^{2b}-x-x\log(\frac{br_{k}^{2b}}{x})}. \label{def of hk4}
\end{align}
Let $k \in \{2,4,\ldots,2g\}$. As $n \to + \infty$, we have
\begin{align*}
& \sum_{j=j_{k,-}}^{g_{k,-}-1} \log \bigg\{1- \frac{1}{2}\mathrm{erfc} \Big( -\eta_{j,k}\sqrt{\tfrac{a_{j}}{2}} \Big) \bigg\} = n\sum_{j=j_{k,-}}^{g_{k,-}-1} \mathfrak{g}_{k,1}(j/n) + \log n \sum_{j=j_{k,-}}^{g_{k,-}-1} \mathfrak{g}_{k,2}(j/n) +\sum_{j=j_{k,-}}^{g_{k,-}-1} \mathfrak{g}_{k,3}(j/n) \\
&  + \frac{1}{n}\sum_{j=j_{k,-}}^{g_{k,-}-1} \mathfrak{g}_{k,4}(j/n)+ \frac{1}{n^{2}}\sum_{j=j_{k,-}}^{g_{k,-}-1} \mathfrak{g}_{k,5}(j/n)+ \frac{1}{n^{3}}\sum_{j=j_{k,-}}^{g_{k,-}-1} \mathfrak{g}_{k,6}(j/n) + \bigO\bigg(\frac{\sqrt{n}}{M^{7}}\bigg), \\
& \sum_{j=j_{k,-}}^{g_{k,-}-1} \log \Bigg\{ 1+\frac{R_{a_{j}}(\eta_{j,k})}{1- \frac{1}{2}\mathrm{erfc} \big( -\eta_{j,k}\sqrt{\tfrac{a_{j}}{2}} \big)} \Bigg\} = \sum_{j=j_{k,-}}^{g_{k,-}-1} \mathfrak{h}_{k,3}(j/n) + \frac{1}{n}\sum_{j=j_{k,-}}^{g_{k,-}-1} \mathfrak{h}_{k,4}(j/n)+ \bigO(M^{-2}),
\end{align*}
where the functions $\mathfrak{g}_{k,1},\ldots,\mathfrak{g}_{k,6},\mathfrak{h}_{k,3}$ and $\mathfrak{h}_{k,4}$ are as in \eqref{def of gk1}-\eqref{def of hk4}.
\end{lemma}
\begin{remark}\label{remark: some functions have poles}
Using that $\mathfrak{g}_{4,k}$, $\mathfrak{g}_{5,k}$, $\mathfrak{g}_{6,k}$ and $\mathfrak{h}_{k,4}$ each have a pole at $x=br_{k}^{2b}$, of order $2$, $4$, $6$ and~$1$ respectively, we can easily show that the sums
\begin{align*}
\frac{1}{n}\sum_{j=g_{k,+}+1}^{j_{k,+}} \mathfrak{g}_{k,4}(j/n), \quad  \frac{1}{n^{2}}\sum_{j=g_{k,+}+1}^{j_{k,+}} \mathfrak{g}_{k,5}(j/n), \quad \frac{1}{n^{3}}\sum_{j=g_{k,+}+1}^{j_{k,+}} \mathfrak{g}_{k,6}(j/n), \quad \frac{1}{n}\sum_{j=g_{k,+}+1}^{j_{k,+}} \mathfrak{h}_{k,4}(j/n)
\end{align*}
are, as $n \to + \infty$, of order $\frac{\sqrt{n}}{M}$, $\frac{\sqrt{n}}{M^{3}}$, $\frac{\sqrt{n}}{M^{5}}$ and $\log n$, respectively. Since $M=n^{\frac{1}{12}}$, each of these sums is thus of order greater than $1$.
\end{remark}
\begin{proof}
Let $k \in \{1,3,\ldots,2g-1\}$, and define $\mathcal{F}_{k}(\tilde{\alpha})=\mathcal{F}_{k}(\tilde{\alpha};x)$ by 
\begin{align*}
\mathcal{F}_{k}(\tilde{\alpha}) = \frac{\sqrt{x+\tilde{\alpha}}}{\sqrt{b}}\sqrt{\frac{br_{k}^{2b}}{x+\tilde{\alpha}}-1-\log \bigg( \frac{br_{k}^{2b}}{x+\tilde{\alpha}} \bigg)}.
\end{align*}
By \eqref{lol42} we have $\mathcal{F}_{k}(\frac{\alpha}{n};\frac{j}{n}) = -\frac{\eta_{j,k}\sqrt{a_{j}}}{\sqrt{2}\sqrt{n}}$ for all $j \in \{g_{k,+}+1,\ldots,j_{k,+}\}$. For each $x \in [br_{k}^{2b},br_{k+1}^{2b}]$, using Taylor's theorem, we obtain
\begin{align*}
\mathcal{F}_{k}(\tilde{\alpha};x) = \sum_{\ell=0}^{4}\frac{\mathcal{F}_{k}^{(\ell)}(0;x)}{\ell!}\tilde{\alpha}^{\ell} + \frac{\mathcal{F}_{k}^{(5)}(\xi(\tilde{\alpha};x);x)}{5!}\tilde{\alpha}^{5},
\end{align*}
for some $\xi(\tilde{\alpha};x) \in (0,\tilde{\alpha})$ if $\tilde{\alpha}>0$ and $\xi(\tilde{\alpha};x) \in(\tilde{\alpha},0)$ if $\tilde{\alpha}<0$. The functions $\mathcal{F}_{k}^{(1)},\ldots,\mathcal{F}_{k}^{(5)}$ are explicitly computable, but since their expressions are rather long we do not write them down (we simply mention that $x \mapsto \mathcal{F}_{k}(0;x)$ has a simple zero as $x \searrow br_{k}^{2b}$, while the functions $x \mapsto \mathcal{F}_{k}^{(\ell)}(0;x)$ for $\ell \geq 1$ remain bounded as $x \searrow br_{k}^{2b}$). The function $\mathcal{F}_{k}^{(5)}$ satisfies the following: there exist $C>0$ and $\delta>0$ such that
\begin{align*}
\big|\mathcal{F}_{k}^{(5)}(\xi(\tilde{\alpha};x);x)\big| \leq C, \qquad \mbox{for all } |\tilde{\alpha}| \leq \delta \mbox{ and all } x \in [br_{k}^{2b},br_{k+1}^{2b}].
\end{align*}
We thus have
\begin{align*}
-\frac{\eta_{j,k}\sqrt{a_{j}}}{\sqrt{2}\sqrt{n}} = \sum_{\ell=0}^{4}\frac{\mathcal{F}_{k}^{(\ell)}(0;\frac{j}{n})}{\ell!}\frac{\alpha^{\ell}}{n^{\ell}} + \bigO(n^{-5}), \qquad \mbox{as } n \to + \infty
\end{align*}
uniformly for $j \in \{g_{k,+}+1,\ldots,j_{k,+}\}$. These asymptotics can be rewritten as 
\begin{align}\label{lol43}
& \hspace{-0.3cm} -\hspace{-0.05cm}\frac{\eta_{j,k}\sqrt{a_{j}}}{\sqrt{2}} = \sqrt{n}\mathcal{F}_{k}(0;\tfrac{j}{n}) + \hspace{-0.1cm} \sum_{\ell=1}^{4} \frac{\beta_{2\ell-1}}{(\sqrt{n}\mathcal{F}_{k}(0;\tfrac{j}{n}))^{2\ell-1}} + \bigO(n^{-\frac{9}{2}}), \; \beta_{2\ell-1} := \frac{\mathcal{F}_{k}^{(\ell)}(0;\frac{j}{n})}{\ell!}\alpha^{\ell}\mathcal{F}_{k}(0;\tfrac{j}{n})^{2\ell-1}
\end{align}
as $n \to + \infty$ uniformly for $j \in \{g_{k,+}+1,\ldots,j_{k,+}\}$. Since $x\mapsto \mathcal{F}_{k}(0;x)$ has a simple zero at $x=br_{k}^{2b}$, there exist constants $c_{1},c_{2},c_{1}',c_{2}'>0$ such that 
\begin{align}\label{lol44}
c_{1}'M \leq c_{1} \sqrt{n}\big( \tfrac{j}{n}-br_{k}^{2b} \big) \leq \sqrt{n}\mathcal{F}_{k}(0;\tfrac{j}{n}) \leq c_{2} \sqrt{n}\big( \tfrac{j}{n}-br_{k}^{2b} \big) \leq c_{2}'\sqrt{n}
\end{align}
for all large enough $n$ and all $j \in \{g_{k,+}+1,\ldots,j_{k,+}\}$. On the other hand, using \eqref{large y asymp of erfc} we obtain
\begin{align}
& \log \bigg( \frac{1}{2}\mathrm{erfc}\bigg( z + \frac{\beta_{1}}{z} + \frac{\beta_{3}}{z^{3}}+ \frac{\beta_{5}}{z^{5}}+ \frac{\beta_{7}}{z^{7}} + \frac{\beta_{9}}{z^{9}} \bigg) \bigg) = -z^{2} -\log(z) - \log(2\sqrt{\pi}) - 2\beta_{1} - \frac{\frac{1}{2}+\beta_{1}+\beta_{1}^{2}+2\beta_{3}}{z^{2}} \nonumber \\
& + \frac{\frac{5}{8}+\beta_{1}+\frac{\beta_{1}^{2}}{2}-\beta_{3}-2\beta_{1}\beta_{3}-2\beta_{5}}{z^{4}} + \frac{-\frac{37}{24}+\beta_{1}(-\frac{5}{2}+\beta_{3}-2\beta_{5})-\frac{3\beta_{1}^{2}}{2}-\frac{\beta_{1}^{3}}{3}+\beta_{3}-\beta_{3}^{2}-\beta_{5}-2\beta_{7}}{z^{6}} \nonumber  \\
& + \bigO(z^{-8}), \label{lol45}
\end{align}
as $z \to +\infty$ uniformly for $\beta_{1},\beta_{3},\ldots,\beta_{9}$ in compact subsets of $\mathbb{R}$. Combining \eqref{lol43}, \eqref{lol44} and \eqref{lol45} (with $z=\sqrt{n}\mathcal{F}_{k}(0;\tfrac{j}{n})$), and using that
\begin{align*}
\sum_{j=g_{k,+}+1}^{j_{k,+}} \frac{1}{(\sqrt{n}\mathcal{F}_{k}(0;\tfrac{j}{n}))^{8}} = \bigO\bigg(\frac{\sqrt{n}}{M^{7}}\bigg), \qquad \mbox{as } n \to + \infty,
\end{align*}
we find \eqref{lol46} after a long but straightforward computation. To prove \eqref{lol47}, we first use \eqref{lol48} to find
\begin{align}
& \frac{R_{a_{j}}(\eta_{j,k})}{\frac{1}{2}\mathrm{erfc} \big( -\eta_{j,k}\sqrt{\frac{a_{j}}{2}} \big)} \nonumber \\
& =  \frac{\exp\big(-\frac{a_{j}\eta_{j,k}^{2}}{2}\big)}{\sqrt{n}\mathcal{F}_{k}(0;\tfrac{j}{n})\frac{1}{2}\mathrm{erfc} \big( -\eta_{j,k}\sqrt{\frac{a_{j}}{2}} \big)} \frac{\sqrt{b}\mathcal{F}_{k}(0;\tfrac{j}{n})}{\sqrt{2\pi}\sqrt{j/n + \frac{\alpha}{n}}} \Bigg\{ c_{0}( \eta_{j,k} ) + \frac{b \; c_{1}( \eta_{j,k} )}{n(j/n+\frac{\alpha}{n})}  + \bigO(n^{-2}) \Bigg\}  \label{lol49}
\end{align}
as $n \to + \infty$ uniformly for $j \in \{g_{k,+}+1,\ldots,j_{k,+}\}$. Using again \eqref{large y asymp of erfc}, we obtain
\begin{align}\label{lol50}
\frac{\exp \big( -\big( z + \frac{\beta_{1}}{z} + \frac{\beta_{3}}{z^{3}}+ \frac{\beta_{5}}{z^{5}}+ \frac{\beta_{7}}{z^{7}} + \frac{\beta_{9}}{z^{9}} \big)^{2} \big)}{\frac{z}{2}\mathrm{erfc}\big( z + \frac{\beta_{1}}{z} + \frac{\beta_{3}}{z^{3}}+ \frac{\beta_{5}}{z^{5}}+ \frac{\beta_{7}}{z^{7}} + \frac{\beta_{9}}{z^{9}} \big)} = 2\sqrt{\pi} + \frac{\sqrt{\pi}(1+2\beta_{1})}{z^{2}} + \bigO(z^{-4}), \qquad \mbox{as } z \to + \infty
\end{align}
uniformly for $\beta_{1},\beta_{3},\ldots,\beta_{9}$ in compact subsets of $\mathbb{R}$. The first ratio on the right-hand side of \eqref{lol49} can then be expanded by combining \eqref{lol50} (with $z=\sqrt{n}\mathcal{F}_{k}(0;\tfrac{j}{n})$) and \eqref{lol43}. For the second part in \eqref{lol49}, since the coefficients $c_{0}(\eta)$ and $c_{1}(\eta)$ are analytic for $\eta \in \mathbb{R}$ (the singularities at $\eta=0$ in \eqref{def of c0 and c1} are removable), we have
\begin{align}\label{lol51}
\frac{\sqrt{b}\mathcal{F}_{k}(0;\tfrac{j}{n})}{\sqrt{2\pi}\sqrt{j/n + \frac{\alpha}{n}}} \Bigg\{ c_{0}( \eta_{j,k} ) + \frac{b \; c_{1}( \eta_{j,k} )}{n(j/n+\frac{\alpha}{n})} \Bigg\} = \mathcal{F}_{k}(0;\tfrac{j}{n}) \big( \mathcal{G}_{0}(\tfrac{j}{n})+\tfrac{1}{n}\mathcal{G}_{1}(\tfrac{j}{n})+\bigO(n^{-2})  \big)
\end{align}
for some explicit $\mathcal{G}_{0}$, $\mathcal{G}_{1}$ (which we do not write down) such that $\mathcal{G}_{0}(\tfrac{j}{n})$ and $\mathcal{G}_{0}(\tfrac{j}{n})$ remain of order~$1$ as $n \to + \infty$ uniformly for $j \in \{g_{k,+}+1,\ldots,j_{k,+}\}$. After a computation using \eqref{lol49}, \eqref{lol50} and \eqref{lol51}, we find
\begin{align*}
\sum_{j=g_{k,+}+1}^{j_{k,+}} \log \Bigg\{ 1-\frac{R_{a_{j}}(\eta_{j,k})}{\frac{1}{2}\mathrm{erfc} \big( -\eta_{j,k}\sqrt{\frac{a_{j}}{2}} \big)} \Bigg\} = \sum_{j=g_{k,+}+1}^{j_{k,+}} \bigg( \mathfrak{h}_{k,3}(j/n) + \frac{1}{n} \mathfrak{h}_{k,4}(j/n)+ \bigO\bigg(\frac{1}{n^{2}\mathcal{F}(0;\frac{j}{n})^{3}}\bigg) \bigg),
\end{align*}
as $n \to + \infty$. Since $x \mapsto \mathcal{F}_{k}(0;x)$ has a simple zero at $x=br_{k}^{2b}$, we have
\begin{align*}
\sum_{j=g_{k,+}+1}^{j_{k,+}} \frac{1}{n^{2}\mathcal{F}_{k}(0;\frac{j}{n})^{3}} \leq  \frac{C}{M^{2}}, \qquad \mbox{for a certain } C>0 \mbox{ and for all sufficiently large }n,
\end{align*}
and \eqref{lol47} follows. The proof for $k \in \{2,4,\ldots,2g\}$ is similar and we omit it.
\end{proof}
By applying Lemma \ref{lemma:Riemann sum NEW} with $f$ replaced by $\mathfrak{g}_{k,1},\ldots,\mathfrak{g}_{k,6},\mathfrak{h}_{k,3}$ and $\mathfrak{h}_{4,k}$, we can obtain the large $n$ asymptotics of the various sums appearing in the above Lemma \ref{lemma:preliminary expansions}. Note that, as already mentioned in Remark \ref{remark: some functions have poles}, the functions $\mathfrak{g}_{4,k}$, $\mathfrak{g}_{5,k}$, $\mathfrak{g}_{6,k}$ and $\mathfrak{h}_{k,4}$ have poles at $x=br_{k}^{2b}$. Nevertheless, we can still apply Lemma \ref{lemma:Riemann sum NEW} to obtain precise large $n$ asymptotics for 
\begin{align*}
\frac{1}{n}\sum_{j=g_{k,+}+1}^{j_{k,+}} \mathfrak{g}_{k,4}(j/n), \quad  \frac{1}{n^{2}}\sum_{j=g_{k,+}+1}^{j_{k,+}} \mathfrak{g}_{k,5}(j/n), \quad \frac{1}{n^{3}}\sum_{j=g_{k,+}+1}^{j_{k,+}} \mathfrak{g}_{k,6}(j/n), \quad \frac{1}{n}\sum_{j=g_{k,+}+1}^{j_{k,+}} \mathfrak{h}_{k,4}(j/n),
\end{align*}
see in particular Remark \ref{remark:we can apply the lemma for functions with a pole}. 
\begin{lemma}\label{lemma: expansion of various sums}
Let $k \in \{1,3,\ldots,2g-1\}$. As $n \to + \infty$, we have
\begin{align}
& n\sum_{j=g_{k,+}+1}^{j_{k,+}} \mathfrak{g}_{k,1}(j/n) = \frac{br_{k}^{4b}(2\epsilon + \epsilon^{2} + 2 \log(1-\epsilon))}{4(1-\epsilon)^{2}}n^{2} + \frac{(1-2\alpha - 2 \theta_{k,+}^{(n,\epsilon)})(\epsilon + \log(1-\epsilon))r_{k}^{2b}}{2(1-\epsilon)}n \nonumber \\
& + \frac{br_{k}^{4b}}{6}M^{3}\sqrt{n}+ \frac{11br_{k}^{4b}}{24}M^{4} + \frac{r_{k}^{2b}(1-2\alpha-2\theta_{k,+}^{(n,M)})}{4}M^{2} \nonumber \\
& + \frac{1-6(\alpha+\theta_{k,+}^{(n,\epsilon)})+6(\alpha+\theta_{k,+}^{(n,\epsilon)})^{2}}{12b}\log(1-\epsilon) + \bigO \bigg( \frac{M^{5}}{\sqrt{n}} \bigg), \label{lol56} \\
& \log n \sum_{j=g_{k,+}+1}^{j_{k,+}} \mathfrak{g}_{k,2}(j/n) = - \frac{br_{k}^{2b}\epsilon}{2(1-\epsilon)}n \log n + \frac{br_{k}^{2b}}{2}M \sqrt{n} \log n  \nonumber \\
& + \frac{br_{k}^{2b}M^{2}-\theta_{k,+}^{(n,M)}+\theta_{k,+}^{(n,\epsilon)}}{2}\log n + \bigO \bigg( \frac{M^{3}\log n}{\sqrt{n}} \bigg), \label{lol57} \\
& \sum_{j=g_{k,+}+1}^{j_{k,+}} \mathfrak{g}_{k,3}(j/n) = n \int_{\frac{g_{k,+}+1}{n}}^{\frac{br_{k}^{2b}}{1-\epsilon}}\mathfrak{g}_{k,3}(x)dx + \frac{1}{4}\log n - \frac{1}{2}\log M - \frac{1}{2}\log(r_{k}^{b}\sqrt{2\pi}) \nonumber \\
& + \frac{1-2\alpha-2\theta_{k,+}^{(n,\epsilon)}}{2}\bigg( \frac{\alpha}{b}\log(1-\epsilon) -  \log(2\sqrt{\pi}r_{k}^{b}) - \frac{1}{2} \log \bigg(\frac{-\epsilon - \log(1-\epsilon)}{1-\epsilon} \bigg) \bigg) + \bigO \bigg( \frac{M}{\sqrt{n}} \bigg), \label{lol58} \\
& \frac{1}{n}\sum_{j=g_{k,+}+1}^{j_{k,+}} \mathfrak{g}_{k,4}(j/n) = \int_{\frac{g_{k,+}+1}{n}}^{\frac{br_{k}^{2b}}{1-\epsilon}}\mathfrak{g}_{k,4}(x)dx + \bigO( M^{-2} ), \label{lol59} \\
& \frac{1}{n^{2}}\sum_{j=g_{k,+}+1}^{j_{k,+}} \mathfrak{g}_{k,5}(j/n) = \frac{5b\sqrt{n}}{6r_{k}^{2b}M^{3}}+\bigO(M^{-2}), \label{lol60} \\
& \frac{1}{n^{3}}\sum_{j=g_{k,+}+1}^{j_{k,+}} \mathfrak{g}_{k,6}(j/n) = \frac{-37b \sqrt{n}}{15 r_{k}^{4b}M^{5}} + \bigO(M^{-4}), \label{lol61} \\
& \sum_{j=g_{k,+}+1}^{j_{k,+}} \mathfrak{h}_{k,3}(j/n) = n \int_{\frac{g_{k,+}+1}{n}}^{\frac{br_{k}^{2b}}{1-\epsilon}}\mathfrak{h}_{k,3}(x)dx + \bigg( \frac{1}{2}-\alpha-\theta_{k,+}^{(n,\epsilon)} \bigg) \log \bigg( \frac{\sqrt{2}\sqrt{-\epsilon-\log(1-\epsilon)}}{\epsilon} \bigg) + \bigO\bigg( \frac{M}{\sqrt{n}} \bigg), \label{lol62} \\
& \frac{1}{n}\sum_{j=g_{k,+}+1}^{j_{k,+}} \mathfrak{h}_{k,4}(j/n) = \int_{\frac{g_{k,+}+1}{n}}^{\frac{br_{k}^{2b}}{1-\epsilon}}\mathfrak{h}_{k,4}(x)dx + \bigg( \frac{1}{M \sqrt{n}} \bigg). \label{lol63}
\end{align}
Let $k \in \{2,4,\ldots,2g\}$. As $n \to + \infty$, we have
\begin{align*}
& n\sum_{j=j_{k,-}}^{g_{k,-}-1} \mathfrak{g}_{k,1}(j/n) = \frac{br_{k}^{4b}(2\epsilon - \epsilon^{2} - 2 \log(1+\epsilon))}{4(1+\epsilon)^{2}}n^{2} - \frac{(1+2\alpha - 2 \theta_{k,-}^{(n,\epsilon)})(\epsilon - \log(1+\epsilon))r_{k}^{2b}}{2(1+\epsilon)}n \\
& + \frac{br_{k}^{4b}}{6}M^{3}\sqrt{n}- \frac{11br_{k}^{4b}}{24}M^{4} + \frac{r_{k}^{2b}(1+2\alpha-2\theta_{k,-}^{(n,M)})}{4}M^{2} \\
& + \frac{-1+6(\theta_{k,-}^{(n,\epsilon)}-\alpha)-6(\theta_{k,-}^{(n,\epsilon)}-\alpha)^{2}}{12b}\log(1+\epsilon) + \bigO \bigg( \frac{M^{5}}{\sqrt{n}} \bigg), \\
& \log n \sum_{j=j_{k,-}}^{g_{k,-}-1} \mathfrak{g}_{k,2}(j/n) = - \frac{br_{k}^{2b}\epsilon}{2(1+\epsilon)}n \log n + \frac{br_{k}^{2b}}{2}M \sqrt{n} \log n \\
& + \frac{-br_{k}^{2b}M^{2}-\theta_{k,-}^{(n,M)}+\theta_{k,-}^{(n,\epsilon)}}{2}\log n + \bigO \bigg( \frac{M^{3}\log n}{\sqrt{n}} \bigg), \\
& \sum_{j=j_{k,-}}^{g_{k,-}-1} \mathfrak{g}_{k,3}(j/n) = n \int_{\frac{br_{k}^{2b}}{1+\epsilon}}^{\frac{g_{k,-}-1}{n}}\mathfrak{g}_{k,3}(x)dx + \frac{1}{4}\log n - \frac{1}{2}\log M - \frac{1}{2}\log(r_{k}^{b}\sqrt{2\pi}) \\
& + \frac{1+2\alpha-2\theta_{k,-}^{(n,\epsilon)}}{2}\bigg( \frac{\alpha}{b}\log(1+\epsilon) -  \log(2\sqrt{\pi}r_{k}^{b}) - \frac{1}{2} \log \bigg(\frac{\epsilon - \log(1+\epsilon)}{1+\epsilon} \bigg) \bigg) + \bigO \bigg( \frac{M}{\sqrt{n}} \bigg), \\
& \frac{1}{n}\sum_{j=j_{k,-}}^{g_{k,-}-1} \mathfrak{g}_{k,4}(j/n) = \int_{\frac{br_{k}^{2b}}{1+\epsilon}}^{\frac{g_{k,-}-1}{n}}\mathfrak{g}_{k,4}(x)dx + \bigO ( M^{-2} ), \\
& \frac{1}{n^{2}}\sum_{j=j_{k,-}}^{g_{k,-}-1} \mathfrak{g}_{k,5}(j/n) = \frac{5b\sqrt{n}}{6r_{k}^{2b}M^{3}}+\bigO(M^{-2}), \\
& \frac{1}{n^{3}}\sum_{j=j_{k,-}}^{g_{k,-}-1} \mathfrak{g}_{k,6}(j/n) = \frac{-37b \sqrt{n}}{15 r_{k}^{4b}M^{5}} + \bigO(M^{-4}), \\
& \sum_{j=j_{k,-}}^{g_{k,-}-1} \mathfrak{h}_{k,3}(j/n) = n \int_{\frac{br_{k}^{2b}}{1+\epsilon}}^{\frac{g_{k,-}-1}{n}}\mathfrak{h}_{k,3}(x)dx + \bigg( \frac{1}{2}+\alpha-\theta_{k,-}^{(n,\epsilon)} \bigg) \log \bigg( \frac{\sqrt{2}\sqrt{\epsilon-\log(1+\epsilon)}}{\epsilon} \bigg) + \bigO\bigg( \frac{M}{\sqrt{n}} \bigg), \\
& \frac{1}{n}\sum_{j=j_{k,-}}^{g_{k,-}-1} \mathfrak{h}_{k,4}(j/n) = \int_{\frac{br_{k}^{2b}}{1+\epsilon}}^{\frac{g_{k,-}-1}{n}}\mathfrak{h}_{k,4}(x)dx + \bigg( \frac{1}{M \sqrt{n}} \bigg).
\end{align*}
\end{lemma}
\begin{proof}
We start with the proof of \eqref{lol56}. Using Lemma \ref{lemma:Riemann sum NEW} with $f$ replaced by $\mathfrak{g}_{k,1}$ and with 
\begin{align}\label{A a0 B b0 in proof}
A = \frac{g_{k,+}+1}{n} = \frac{br_{k}^{2b}}{1-\frac{M}{\sqrt{n}}} + \frac{1-\alpha-\theta_{k,+}^{(n,M)}}{n}, \quad a_{0}=0, \quad B=\frac{br_{k}^{2b}}{1-\epsilon}, \quad b_{0} = -\alpha-\theta_{k,+}^{(n,\epsilon)},
\end{align}
we obtain
\begin{align*}
 n \sum_{j=g_{k,+}+1}^{j_{k,+}} \mathfrak{g}_{k,1}(j/n) & = n^{2} \int_{A}^{B} \mathfrak{g}_{k,1}(x)dx + n\frac{\mathfrak{g}_{k,1}(A)+(1+2b_{0})\mathfrak{g}_{k,1}(B)}{2}  \nonumber \\
& + \frac{-\mathfrak{g}_{k,1}'(A)+(1+6b_{0}+6b_{0}^{2})\mathfrak{g}_{k,1}'(B)}{12}+ \bigO(n^{-1}).
\end{align*}
Since $\mathfrak{g}_{k,1}$ and all its derivatives remain bounded at $br_{k}^{2}$, $\int_{A}^{B} \mathfrak{g}_{k,1}(x)dx$, $\mathfrak{g}_{k,1}(A)$ and $\mathfrak{g}_{k,1}'(A)$ can be expanded by standard Taylor series. Using then the primitive $\int \mathfrak{g}_{k,1}(x)dx = \frac{3x^{2}}{4b}+\frac{x^{2}}{2b}\log(\frac{br_{k}^{2b}}{x})-xr_{k}^{2b}$, we find \eqref{lol56}. Since $\mathfrak{g}_{k,2}(x)=-\frac{1}{2}$, we have
\begin{align*}
\log n \sum_{j=g_{k,+}+1}^{j_{k,+}} \mathfrak{g}_{k,2}(j/n) = \frac{n \log n}{2} \bigg( \frac{br_{k}^{2b}}{1-\frac{M}{\sqrt{n}}} - \frac{br_{k}^{2b}}{1-\epsilon} \bigg) + \frac{\theta_{k,+}^{(n,\epsilon)}-\theta_{k,+}^{(n,M)}}{2} \log n,
\end{align*}
and a direct expansion yields \eqref{lol57}. We next turn to the proof of \eqref{lol58}. Using Lemma \ref{lemma:Riemann sum NEW} with $f$ replaced by $\mathfrak{g}_{k,3}$ and with $A,a_{0},B,b_{0}$ as in \eqref{A a0 B b0 in proof}, we get
\begin{align*}
\sum_{j=g_{k,+}+1}^{j_{k,+}} \mathfrak{g}_{k,3}(j/n) & = n \int_{A}^{B} \mathfrak{g}_{k,3}(x)dx + \frac{\mathfrak{g}_{k,3}(A)+(1+2b_{0})\mathfrak{g}_{k,3}(B)}{2} + \bigO\bigg(\frac{\log n}{n}\bigg),
\end{align*}
where we have used that $\frac{\mathfrak{g}_{k,3}'(A)}{n}=\bigO(\frac{\log n}{n})$ to estimate the error term. After expanding $\mathfrak{g}_{k,3}(A)$, we find \eqref{lol58}. The expansion \eqref{lol59} direct follows from Lemma \ref{lemma:Riemann sum NEW} with $f$ replaced by $\mathfrak{g}_{k,4}$ and with $A,a_{0},B,b_{0}$ as in \eqref{A a0 B b0 in proof}, and from the fact that $\frac{\mathfrak{g}_{k,4}(A)}{n} \lesssim \frac{1}{n(A-br_{k}^{2b})^{2}}=\bigO(\frac{1}{M^{2}})$. Next, applying Lemma \ref{lemma:Riemann sum NEW} with $f$ replaced by $\mathfrak{g}_{k,5}$ and with $A,a_{0},B,b_{0}$ as in \eqref{A a0 B b0 in proof}, we get
\begin{align*}
\frac{1}{n^{2}}\sum_{j=g_{k,+}+1}^{j_{k,+}} \mathfrak{g}_{k,5}(j/n) = \frac{1}{n} \int_{A}^{B} \mathfrak{g}_{k,5}(x)dx + \bigO\bigg(\frac{1}{M^{4}}\bigg), \qquad \mbox{as } n \to + \infty,
\end{align*}
where we have used $\frac{\mathfrak{g}_{k,5}(A)}{n^{2}} \lesssim \frac{1}{n^{2}(A-br_{k}^{2b})^{4}}=\bigO(\frac{1}{M^{4}})$ to estimate the error term. Since
\begin{align*}
\mathfrak{g}_{k,5}(x) = \frac{5 b^{4}r_{k}^{4b}}{2(x-br_{k}^{2b})^{4}} + \bigO\bigg( \frac{1}{(x-br_{k}^{2b})^{3}} \bigg), \qquad \mbox{as } x \to br_{k}^{2b},
\end{align*}
we have
\begin{align*}
\frac{1}{n} \int_{A}^{B} \mathfrak{g}_{k,5}(x)dx = \bigg[ \frac{5b^{4}r_{k}^{4b}}{6n(br_{k}^{2b}-x)^{3}} \bigg]_{A}^{B} + \bigO \bigg( \frac{1}{n(A-br_{k}^{2b})^{2}} \bigg) = \frac{5b \sqrt{n}}{6 r_{k}^{2b}M^{3}} + \bigO(M^{-2}), \qquad \mbox{as } n \to +\infty,
\end{align*}
and \eqref{lol60} follows. Similarly, using Lemma \ref{lemma:Riemann sum NEW} with $f$ replaced by $\mathfrak{g}_{k,6}$ and with $A,a_{0},B,b_{0}$ as in \eqref{A a0 B b0 in proof}, we get
\begin{align*}
\frac{1}{n^{3}}\sum_{j=g_{k,+}+1}^{j_{k,+}} \mathfrak{g}_{k,6}(j/n) = \frac{1}{n^{2}} \int_{A}^{B} \mathfrak{g}_{k,6}(x)dx + \bigO\bigg(\frac{1}{M^{6}}\bigg), \qquad \mbox{as } n \to + \infty.
\end{align*}
The expansion
\begin{align*}
\mathfrak{g}_{k,6}(x) = -\frac{37 b^{6}r_{k}^{6b}}{3(x-br_{k}^{2b})^{6}} + \bigO \big( (x-br_{k}^{2b})^{-5} \big), \qquad \mbox{as } x \to br_{k}^{2b},
\end{align*}
implies that 
\begin{align*}
\frac{1}{n^{2}} \int_{A}^{B} \mathfrak{g}_{k,6}(x)dx = -\frac{37b\sqrt{n}}{15 r_{k}^{4b}M^{5}} + \bigO(M^{-4}), \qquad \mbox{as } n \to +\infty,
\end{align*}
and \eqref{lol61} follows. Next, using Lemma \ref{lemma:Riemann sum NEW} with $A,a_{0},B,b_{0}$ as in \eqref{A a0 B b0 in proof} with $f$ replaced by $\mathfrak{h}_{k,3}$, we get
\begin{align*}
\sum_{j=g_{k,+}+1}^{j_{k,+}} \mathfrak{h}_{k,3}(j/n) = n \int_{A}^{B}\mathfrak{h}_{k,3}(x)dx + \frac{\mathfrak{h}_{k,3}(A)+(1+2b_{0})\mathfrak{h}_{k,3}(B)}{2} + \bigO(n^{-1}), \qquad \mbox{as } n \to +\infty,
\end{align*}
where we have used that $\mathfrak{h}_{k,3}'(A)=\bigO(1)$. After expanding $\mathfrak{h}_{k,3}(A)$ as $n \to \infty$ we get \eqref{lol62}. Since $\mathfrak{h}_{k,4}$ has a simple pole at $br_{k}^{2b}$, we have $\frac{\mathfrak{h}_{k,4}(A)}{n} = \bigO(\frac{1}{M\sqrt{n}})$ as $n \to +\infty$. Formula \eqref{lol63} now directly follows from Lemma \ref{lemma:Riemann sum NEW} with $A,a_{0},B,b_{0}$ as in \eqref{A a0 B b0 in proof} with $f$ replaced by $\mathfrak{h}_{k,4}$. This finishes the proof of the formulas corresponding to $k$ odd. The formulas corresponding to $k$ even can be proved in a similar way.
\end{proof}
We have not been able to find explicit primitives for any of the functions $\mathfrak{g}_{k,3}$, $\mathfrak{g}_{k,4}$, $\mathfrak{h}_{k,3}$ and $\mathfrak{h}_{k,4}$. However, quite remarkably the functions $\mathfrak{g}_{k,3}+\mathfrak{h}_{k,3}$ and $\mathfrak{g}_{k,4}+\mathfrak{h}_{k,4}$ admit easy primitives, namely
\begin{align*}
& \int \Big(\mathfrak{g}_{k,3}(x)+\mathfrak{h}_{k,3}(x)\Big)dx = \frac{x}{2}+\frac{\alpha x}{b} + \frac{x}{2}\log \bigg( \frac{bx}{2\pi} \bigg) + \frac{\alpha x}{b}\log \frac{b r_{k}^{2b}}{x} + (br_{k}^{2b}-x)\log |x-br_{k}^{2b}|, \\
& \int \Big(\mathfrak{g}_{k,4}(x)+\mathfrak{h}_{k,4}(x)\Big)dx = \frac{b^{2}r_{k}^{2b}}{x-br_{k}^{2b}} - \frac{b^{2}-6b \alpha + 6 \alpha^{2}}{12b}\log(x) - \alpha \log|x-br_{k}^{2b}|.
\end{align*}
Using these primitives, we obtain the following large $n$ asymptotics: for $k \in \{1,3,\ldots,2g-1\}$, we have
\begin{align}
& n \int_{\frac{g_{k,+}+1}{n}}^{\frac{br_{k}^{2b}}{1-\epsilon}}(\mathfrak{g}_{k,3}(x)+\mathfrak{h}_{k,3}(x))dx = \frac{\epsilon\big(b+2\alpha-2b \log \big( \frac{\epsilon}{1-\epsilon} r_{k}^{b}\sqrt{2\pi} \big)\big)+(2\alpha-b)\log(1-\epsilon)}{2(1-\epsilon)}r_{k}^{2b}n \nonumber \\
& - \frac{br_{k}^{2b}M}{2}\sqrt{n}\log n + br_{k}^{2b}M(\log (M r_{k}^{b}\sqrt{2\pi}) -1 )\sqrt{n} + \frac{2\alpha-b}{4}r_{k}^{2b}M^{2} \nonumber \\
&  + (1+br_{k}^{2b}M^{2}-\alpha - \theta_{k,+}^{(n,M)})\log \bigg( \frac{Mr_{k}^{b}\sqrt{2\pi}}{\sqrt{n}} \bigg) + \bigO\bigg( \frac{M^{3}\log n}{\sqrt{n}} \bigg), \label{lol11} \\
& \int_{\frac{g_{k,+}+1}{n}}^{\frac{br_{k}^{2b}}{1-\epsilon}}(\mathfrak{g}_{k,4}(x)+\mathfrak{h}_{k,4}(x))dx = \frac{b \sqrt{n}}{M} + \alpha \log \bigg( \frac{M}{\epsilon\sqrt{n}} \bigg) + \frac{b}{\epsilon} + \frac{b^{2}+6b \alpha+6\alpha^{2}}{12b}\log(1-\epsilon)+ \bigO(M^{-2}), \label{lol12}
\end{align}
and for $k \in \{2,4,\ldots,2g\}$, we have
\begin{align}
& n \int_{\frac{br_{k}^{2b}}{1+\epsilon}}^{\frac{g_{k,-}-1}{n}}(\mathfrak{g}_{k,3}(x)+\mathfrak{h}_{k,3}(x))dx = \frac{\epsilon\big(b+2\alpha-2b \log \big( \frac{\epsilon}{1+\epsilon} r_{k}^{b}\sqrt{2\pi} \big)\big)+(b-2\alpha)\log(1+\epsilon)}{2(1+\epsilon)}r_{k}^{2b}n \nonumber \\
& - \frac{br_{k}^{2b}M}{2}\sqrt{n}\log n + br_{k}^{2b}M(\log (M r_{k}^{b}\sqrt{2\pi}) -1 )\sqrt{n} + \frac{b-2\alpha}{4}r_{k}^{2b}M^{2} \nonumber \\
&  + (1-br_{k}^{2b}M^{2}+\alpha - \theta_{k,-}^{(n,M)})\log \bigg( \frac{Mr_{k}^{b}\sqrt{2\pi}}{\sqrt{n}} \bigg) + \bigO\bigg( \frac{M^{3}\log n}{\sqrt{n}} \bigg), \label{lol13} \\
& \int_{\frac{br_{k}^{2b}}{1+\epsilon}}^{\frac{g_{k,-}-1}{n}}(\mathfrak{g}_{k,4}(x)+\mathfrak{h}_{k,4}(x))dx = -\frac{b \sqrt{n}}{M} - \alpha \log \bigg( \frac{M}{\epsilon\sqrt{n}} \bigg) + \frac{b}{\epsilon} - \frac{b^{2}+6b \alpha+6\alpha^{2}}{12b}\log(1+\epsilon)+ \bigO(M^{-2}). \label{lol14}
\end{align}
Substituting the asymptotics of Lemma \ref{lemma: expansion of various sums} and \eqref{lol11}--\eqref{lol14} in Lemma \ref{lemma:preliminary expansions} and then in Lemma \ref{lemma: exact decomposition for the interpolating sums}, and simplifying, we obtain (after a long computation) the following explicit large $n$ asymptotics of $\{S_{2k}^{(1)}\}_{k \,\mathrm{odd}}$ and $\{S_{2k}^{(3)}\}_{k \,\mathrm{even}}$. 
\begin{lemma}\label{lemma: final asymp for S2kp1p}
Let $k \in \{1,3,\ldots,2g-1\}$. As $n \to + \infty$, we have
\begin{align*}
& S_{2k}^{(1)} = \frac{br_{k}^{4b}(2\epsilon +\epsilon^{2} + 2 \log(1-\epsilon))}{4(1-\epsilon)^{2}}n^{2} - \frac{br_{k}^{2b}\epsilon}{2(1-\epsilon)}n\log n + \frac{r_{k}^{2b}}{2(1-\epsilon)} \bigg\{ \Big( 1- 2\theta_{k,+}^{(n,\epsilon)}+b-2b \log \big( r_{k}^{b}\sqrt{2\pi} \big) \Big)\epsilon  \\
& + \Big( 1-b-2\theta_{k,+}^{(n,\epsilon)} \Big)\log(1-\epsilon) - 2b \epsilon \log \bigg(\frac{\epsilon}{1-\epsilon}\bigg) \bigg\} n \\
& + \bigg\{ \frac{br_{k}^{4b}}{6}M^{3} + br_{k}^{2b} \Big( \log(M) + \log(r_{k}^{b}\sqrt{2\pi}) -1 \Big)M-\frac{b}{M}+\frac{5b}{6r_{k}^{2b}M^{3}}-\frac{37b}{15r_{k}^{4b}M^{5}} \bigg\} \sqrt{n} \\
& + \frac{2\theta_{k,+}^{(n,\epsilon)}-1}{4} \log n + \frac{11}{24} br_{k}^{4b}M^{4} + br_{k}^{2b}M^{2}\log M + \bigg\{ \frac{1-b-2\theta_{k,+}^{(n,M)}}{4} + b \log(r_{k}^{b}\sqrt{2\pi}) \bigg\} r_{k}^{2b}M^{2} \\
& + \frac{1-2\theta_{k,+}^{(n,M)}}{2}\log M + \Big(\theta_{k,+}^{(n,\epsilon)}-\theta_{k,+}^{(n,M)}\Big)\log(r_{k}^{b}\sqrt{2\pi}) + \frac{2\theta_{k,+}^{(n,\epsilon)}-1}{2}\log \epsilon + \frac{b}{\epsilon} \\
& + \frac{1+3b+b^{2}-6(1+b)\theta_{k,+}^{(n,\epsilon)}+6(\theta_{k,+}^{(n,\epsilon)})^{2}}{12b} \log(1-\epsilon)  + \bigO\bigg(\frac{M^{5}}{\sqrt{n}}\bigg)+\bigO\bigg(\frac{\sqrt{n}}{M^{7}}\bigg).
\end{align*}
Let $k \in \{2,4,\ldots,2g\}$. As $n \to + \infty$, we have
\begin{align*}
& S_{2k}^{(3)} = \frac{br_{k}^{4b}(2\epsilon -\epsilon^{2} - 2 \log(1+\epsilon))}{4(1+\epsilon)^{2}}n^{2} - \frac{br_{k}^{2b}\epsilon}{2(1+\epsilon)}n\log n + \frac{r_{k}^{2b}}{2(1+\epsilon)} \bigg\{ \Big( 2\theta_{k,-}^{(n,\epsilon)}-1+b-2b \log \big( r_{k}^{b}\sqrt{2\pi} \big) \Big)\epsilon  \\
& + \Big( 1+b-2\theta_{k,-}^{(n,\epsilon)} \Big)\log(1+\epsilon) - 2b \epsilon \log \bigg(\frac{\epsilon}{1+\epsilon}\bigg) \bigg\} n \\
& + \bigg\{ \frac{br_{k}^{4b}}{6}M^{3} + br_{k}^{2b} \Big( \log(M) + \log(r_{k}^{b}\sqrt{2\pi}) -1 \Big)M-\frac{b}{M}+\frac{5b}{6r_{k}^{2b}M^{3}}-\frac{37b}{15r_{k}^{4b}M^{5}} \bigg\} \sqrt{n} \\
& + \frac{2\theta_{k,-}^{(n,\epsilon)}-1}{4} \log n - \frac{11}{24} br_{k}^{4b}M^{4} - br_{k}^{2b}M^{2}\log M + \bigg\{ \frac{1+b-2\theta_{k,-}^{(n,M)}}{4} - b \log(r_{k}^{b}\sqrt{2\pi}) \bigg\} r_{k}^{2b}M^{2} \\
& + \frac{1-2\theta_{k,-}^{(n,M)}}{2}\log M + \Big(\theta_{k,-}^{(n,\epsilon)}-\theta_{k,-}^{(n,M)}\Big)\log(r_{k}^{b}\sqrt{2\pi}) + \frac{2\theta_{k,-}^{(n,\epsilon)}-1}{2}\log \epsilon + \frac{b}{\epsilon} \\
& + \frac{-1+3b-b^{2}+6(1-b)\theta_{k,-}^{(n,\epsilon)}-6(\theta_{k,-}^{(n,\epsilon)})^{2}}{12b} \log(1+\epsilon)  + \bigO\bigg(\frac{M^{5}}{\sqrt{n}}\bigg)+\bigO\bigg(\frac{\sqrt{n}}{M^{7}}\bigg).
\end{align*}
\end{lemma}
Recall from \eqref{S2k splitting in three parts} that 
\begin{align*}
& S_{2k}=S_{2k}^{(1)} + S_{2k}^{(2)} + S_{2k}^{(3)} + \bigO(e^{-cn}), \qquad \mbox{as } n \to + \infty.
\end{align*}
By combining Lemmas \ref{lemma:S2ktilde large M}, \ref{lemma:S2ktilde large M and nice integrals}, \ref{lemma:S2k p3p and p1p: the easy ones} and \ref{lemma: final asymp for S2kp1p} and simplifying, we finally obtain (after another long computation) the large $n$ asymptotics of $S_{2k}$. 
\begin{lemma}\label{lemma: S2k FINAL all k}
Let $k \in \{1,2,\ldots,2g\}$. As $n \to + \infty$, we have
\begin{align*}
S_{2k} = E_{1,k}^{(\epsilon)} n^{2}  + E_{2,k}^{(\epsilon)} n \log n + E_{3,k}^{(n,\epsilon)} n + E_{4,k}\sqrt{n} + E_{5,k}^{(n,\epsilon)} \log n + E_{6,k}^{(n,\epsilon)} + \bigO\bigg(\frac{M^{5}}{\sqrt{n}} + \frac{\sqrt{n}}{M^{7}}\bigg),
\end{align*}
where, for $k \in \{1,3,\ldots,2g-1\}$, the coefficients $E_{1,k}^{(\epsilon)}$, $E_{2,k}^{(\epsilon)}$, $E_{3,k}^{(n,\epsilon)}$, $E_{4,k}$, $E_{5,k}^{(n,\epsilon)}$, $E_{6,k}^{(n,\epsilon)}$ are given by
\begin{align*}
& E_{1,k}^{(\epsilon)} = \frac{br_{k}^{4b}(2\epsilon +\epsilon^{2} + 2 \log(1-\epsilon))}{4(1-\epsilon)^{2}}, \\
& E_{2,k}^{(\epsilon)} = - \frac{br_{k}^{2b}\epsilon}{2(1-\epsilon)}, \\
& E_{3,k}^{(n,\epsilon)} = \frac{(1-b+2b\epsilon-2\theta_{k,+}^{(n,\epsilon)})\log(1-\epsilon) + \epsilon\big(1+b-2\theta_{k,+}^{(n,\epsilon)}-2b\log \big( \epsilon r_{k}^{b}\sqrt{2\pi} \big)\big)}{2(1-\epsilon)}r_{k}^{2b}, \\
& E_{4,k} = \sqrt{2}br_{k}^{b} \int_{-\infty}^{0}\log \bigg( \frac{1}{2}\mathrm{erfc}(y) \bigg)dy + \sqrt{2}br_{k}^{b}\int_{0}^{+\infty} \bigg[ \log \bigg( \frac{1}{2}\mathrm{erfc}(y) \bigg)+y^{2}+\log y + \log(2\sqrt{\pi}) \bigg]dy, \\
& E_{5,k}^{(n,\epsilon)} =  \frac{2\theta_{k,+}^{(n,\epsilon)}-1}{4}, \\
& E_{6,k}^{(n,\epsilon)} = \frac{1+3b+b^{2}-6(1+b)\theta_{k,+}^{(n,\epsilon)}+6(\theta_{k,+}^{(n,\epsilon)})^{2}}{12b}\log(1-\epsilon) + \frac{b}{\epsilon} + \frac{2\theta_{k,+}^{(n,\epsilon)}-1}{2}\log \Big(\epsilon r_{k}^{b}\sqrt{2\pi}\Big) \\
& +2b \int_{-\infty}^{0} \bigg\{ 2y\log \bigg( \frac{1}{2}\mathrm{erfc}(y)\bigg) + \frac{e^{-y^{2}}(1-5y^{2})}{3\sqrt{\pi}\mathrm{erfc}(y)} \bigg\}dy \\
& +2b \int_{0}^{+\infty} \bigg\{ 2y\log \bigg( \frac{1}{2}\mathrm{erfc}(y)\bigg) + \frac{e^{-y^{2}}(1-5y^{2})}{3\sqrt{\pi}\mathrm{erfc}(y)} + \frac{11}{3}y^{3} + 2y \log y + \bigg( \frac{1}{2} + 2 \log(2\sqrt{\pi}) \bigg)y \bigg\}dy,
\end{align*}
while for $k \in \{2,4,\ldots,2g\}$, the coefficients $E_{1,k}^{(\epsilon)}$, $E_{2,k}^{(\epsilon)}$, $E_{3,k}^{(n,\epsilon)}$, $E_{5,k}^{(n,\epsilon)}$, $E_{6,k}^{(n,\epsilon)}$ are given by
\begin{align*}
& E_{1,k}^{(\epsilon)} = \frac{br_{k}^{4b}(2\epsilon -\epsilon^{2} - 2 \log(1+\epsilon))}{4(1+\epsilon)^{2}}, \\
& E_{2,k}^{(\epsilon)} = - \frac{br_{k}^{2b}\epsilon}{2(1+\epsilon)}, \\
& E_{3,k}^{(n,\epsilon)} = \frac{(1+b+2b\epsilon-2\theta_{k,-}^{(n,\epsilon)})\log(1+\epsilon) + \epsilon\big(-1+b+2\theta_{k,-}^{(n,\epsilon)}-2b\log \big( \epsilon r_{k}^{b}\sqrt{2\pi} \big)\big)}{2(1+\epsilon)}r_{k}^{2b}, \\
& E_{4,k} = \sqrt{2}br_{k}^{b} \int_{-\infty}^{0}\log \bigg( \frac{1}{2}\mathrm{erfc}(y) \bigg)dy + \sqrt{2}br_{k}^{b}\int_{0}^{+\infty} \bigg[ \log \bigg( \frac{1}{2}\mathrm{erfc}(y) \bigg)+y^{2}+\log y + \log(2\sqrt{\pi}) \bigg]dy, \\
& E_{5,k}^{(n,\epsilon)} = \frac{2\theta_{k,-}^{(n,\epsilon)}-1}{4}, \\
& E_{6,k}^{(n,\epsilon)} = \frac{-1+3b-b^{2}+6(1-b)\theta_{k,-}^{(n,\epsilon)}-6(\theta_{k,-}^{(n,\epsilon)})^{2}}{12b}\log(1+\epsilon) + \frac{b}{\epsilon} + \frac{2\theta_{k,-}^{(n,\epsilon)}-1}{2}\log \Big(\epsilon r_{k}^{b}\sqrt{2\pi}\Big) \\
& -2b \int_{-\infty}^{0} \bigg\{ 2y\log \bigg( \frac{1}{2}\mathrm{erfc}(y)\bigg) + \frac{e^{-y^{2}}(1-5y^{2})}{3\sqrt{\pi}\mathrm{erfc}(y)} \bigg\}dy \\
& -2b \int_{0}^{+\infty} \bigg\{ 2y\log \bigg( \frac{1}{2}\mathrm{erfc}(y)\bigg) + \frac{e^{-y^{2}}(1-5y^{2})}{3\sqrt{\pi}\mathrm{erfc}(y)} + \frac{11}{3}y^{3} + 2y \log y + \bigg( \frac{1}{2} + 2 \log(2\sqrt{\pi}) \bigg)y \bigg\}dy.
\end{align*}
\end{lemma}
\begin{remark}
Recall that although $S_{2k}^{(1)}$, $S_{2k}^{(2)}$ and $S_{2k}^{(3)}$ depend on $M$, the sum $S_{2k}$ is independent of $M$. As can be seen from the above, all the coefficients $E_{1,k}^{(\epsilon)}$, $E_{2,k}^{(\epsilon)}$, $E_{3,k}^{(n,\epsilon)}$, $E_{5,k}^{(n,\epsilon)}$, $E_{6,k}^{(n,\epsilon)}$ are independent of $M$, as it must.
\end{remark}
For $x \in \mathbb{R}$, $\rho \in (0,1)$ and $a>0$, define
\begin{align}
\Theta(x;\rho,a) & = x(x-1)\log (\rho) + x \log(a) \nonumber \\
& + \sum_{j=0}^{+\infty} \log \bigg( 1+a\, \rho^{2(j+x)} \bigg) + \sum_{j=0}^{+\infty} \log \bigg( 1+a^{-1} \rho^{2(j+1-x)} \bigg). \label{def of Theta}
\end{align}
By shifting the indices of summation, it can be checked that 
\begin{align*}
\Theta(x+1;\rho,a)=\Theta(x;\rho,a), \qquad \mbox{for all } x \in \mathbb{R}, \; \rho \in (0,1), \; a>0,
\end{align*}
i.e. $x \mapsto \Theta(x;\rho,a)$ is periodic of period $1$. To complete the proof of Theorem \ref{thm:main thm} we will need the following lemma.
\begin{lemma}\label{lemma: Theta is expressible in terms of Jacobi}
We have
\begin{align*}
\Theta(x;\rho,a) = \frac{1}{2}\log \bigg( \frac{\pi a \rho^{-\frac{1}{2}}}{\log(\rho^{-1})} \bigg) + \frac{(\log a)^{2}}{4\log(\rho^{-1})} - \sum_{j=1}^{+\infty} \log(1-\rho^{2j}) + \log \theta \bigg( x + \frac{\log(a \rho)}{2\log(\rho)} \bigg| \frac{\pi i}{\log(\rho^{-1})} \bigg),
\end{align*}
where $\theta$ is the Jacobi theta function given by \eqref{def of Jacobi theta}.
\end{lemma}
\begin{proof}
The statement follows from two remarkable identities of the Jacobi theta function. First, using the Jacobi triple product formula (see e.g. \cite[Eq 20.5.3]{NIST})
\begin{align}\label{triple product formula}
\theta(z|\tau) = \prod_{\ell=1}^{+\infty} (1-e^{2 i \pi \tau \ell})(1+2 e^{i \pi \tau (2\ell-1)}\cos(2\pi z)+e^{i \pi \tau (4\ell-2)}),
\end{align}
we obtain
\begin{align}
\Theta(x,\rho,a) & = x(x-1)\log(\rho) + x \log(a) - \sum_{j=1}^{+\infty} \log(1-\rho^{2j}) \nonumber \\
& + \log \theta \bigg( \frac{(2x-1)\log (\rho)+\log(a)}{2\pi i} \bigg| \frac{\log(\rho^{-1})}{\pi}i \bigg). \label{first identity for Theta}
\end{align}
The claim then follows from a computation using the following Jacobi imaginary transformation (see e.g. \cite[Eq (20.7.32)]{NIST})
\begin{align*}
(-i \tau)^{1/2}\theta(z|\tau) = e^{i \pi \tau' z^{2}}\theta(z \tau'|\tau'), \qquad \mbox{ where } \tau' = -\frac{1}{\tau}.
\end{align*}
\end{proof}

We now finish the proof of Theorem \ref{thm:main thm}.

\begin{proof}[Proof of Theorem \ref{thm:main thm}]
Combining \eqref{log Dn as a sum of sums} with Lemmas \ref{lemma: S0}, \ref{lemma: S2km1 k odd}, \ref{lemma: S2km1 k even} and \ref{lemma: S2k FINAL all k}, we obtain
\begin{align*}
& \log \mathcal{P}_{n} = S_{0}+\sum_{k=1,3,...}^{2g+1}S_{2k-1}+\sum_{k=2,4,...}^{2g}S_{2k-1}+\sum_{k=1}^{2g}S_{2k} \\
& = \bigO(e^{-cn})+\sum_{k=1,3,...}^{2g+1}\bigO(e^{-cn}) \\
& + \sum_{k=2,4,...}^{2g} \bigg\{ F_{1,k}^{(\epsilon)} n^{2}  + F_{2,k}^{(\epsilon)} n \log n + F_{3,k}^{(n,\epsilon)} n + F_{5,k}^{(n,\epsilon)} \log n + F_{6,k}^{(n,\epsilon)} + \widetilde{\Theta}_{k,n} + \bigO\bigg( \frac{(\log n)^{2}}{n} \bigg) \bigg\} \\
& + \sum_{k=1}^{2g} \bigg\{ E_{1,k}^{(\epsilon)} n^{2}  + E_{2,k}^{(\epsilon)} n \log n + E_{3,k}^{(n,\epsilon)} n + E_{4,k}\sqrt{n} + E_{5,k}^{(n,\epsilon)} \log n + E_{6,k}^{(n,\epsilon)} + \bigO\bigg(\frac{M^{5}}{\sqrt{n}} + \frac{\sqrt{n}}{M^{7}}\bigg) \bigg\}
\end{align*}
as $n \to +\infty$, for a certain constant $c>0$. Recall that $M=n^{-\frac{1}{12}}$, so that $\frac{M^{5}}{\sqrt{n}} = \frac{\sqrt{n}}{M^{7}} = n^{-\frac{1}{12}}$. Let $C_{1},\ldots,C_{6},\mathcal{F}_{n}$ be the quantities defined in the statement of Theorem \ref{thm:main thm}. Using the formulas of Lemmas \ref{lemma:Explicit expression for the F coeff} and \ref{lemma: S2k FINAL all k}, we get
\begin{align*}
& \sum_{k=2,4,...}^{2g} F_{1,k}^{(\epsilon)} + \sum_{k=1}^{2g}  E_{1,k}^{(\epsilon)} \\
& = \sum_{k=2,4,...}^{2g} \bigg( \frac{(r_{k}^{2b}-r_{k-1}^{2b})^{2}}{4\log(\frac{r_{k}}{r_{k-1}})} + \frac{br_{k-1}^{4b}}{(1-\epsilon)^{2}}\frac{1-4\epsilon - 2 \log(1-\epsilon)}{4} - \frac{br_{k}^{4b}}{(1+\epsilon)^{2}}\frac{1+4\epsilon - 2 \log(1+\epsilon)}{4}  \bigg) \\
& + \sum_{k=1,3,...}^{2g-1} \frac{br_{k}^{4b}(2\epsilon +\epsilon^{2} + 2 \log(1-\epsilon))}{4(1-\epsilon)^{2}} + \sum_{k=2,4,...}^{2g} \frac{br_{k}^{4b}(2\epsilon -\epsilon^{2} - 2 \log(1+\epsilon))}{4(1+\epsilon)^{2}} \\
& = \sum_{k=2,4,...}^{2g} \frac{(r_{k}^{2b}-r_{k-1}^{2b})^{2}}{4\log(\frac{r_{k}}{r_{k-1}})} + \frac{b}{4} \sum_{k=1}^{2g} (-1)^{k+1}r_{k}^{4b} = C_{1}.
\end{align*}
Similarly,
\begin{align*}
\sum_{k=2,4,...}^{2g} F_{2,k}^{(\epsilon)} + \sum_{k=1}^{2g}  E_{2,k}^{(\epsilon)} = \sum_{k=2,4,...}^{2g} \bigg( \frac{br_{k-1}^{2b}}{2(1-\epsilon)} -\frac{br_{k}^{2b}}{2(1+\epsilon)} \bigg)  - \sum_{k=1,3,...}^{2g-1}  \frac{br_{k}^{2b}\epsilon}{2(1-\epsilon)} - \sum_{k=2,4,...}^{2g} \frac{br_{k}^{2b}\epsilon}{2(1+\epsilon)} = C_{2}.
\end{align*}
Using again Lemmas \ref{lemma:Explicit expression for the F coeff} and \ref{lemma: S2k FINAL all k}, we obtain after a long computation that
\begin{align*}
\sum_{k=2,4,...}^{2g} F_{3,k}^{(n,\epsilon)} + \sum_{k=1}^{2g}  E_{3,k}^{(n,\epsilon)} = C_{3} \quad \mbox{ and } \quad \sum_{k=2,4,...}^{2g} F_{5,k}^{(n,\epsilon)} + \sum_{k=1}^{2g}  E_{5,k}^{(n,\epsilon)} = C_{5}.
\end{align*}
It is also readily checked that $\sum_{k=1}^{2g}  E_{4,k}=C_{4}$. From \eqref{first identity for Theta} and Lemma \ref{lemma:Explicit expression for the F coeff}, we infer that
\begin{align*}
\sum_{k=2,4,...}^{2g} \widetilde{\Theta}_{k,n} = \sum_{k=1}^{g} \bigg\{ & \Theta\bigg( \theta_{2k},\frac{r_{2k-1}}{r_{2k}},\frac{t_{2k}-br_{2k-1}^{2b}}{br_{2k}^{2b}-t_{2k}} \bigg)  \\
& + \theta_{2k}(\theta_{2k}-1)\log \Big( \frac{r_{2k}}{r_{2k-1}} \Big) + \theta_{2k} \log \bigg(\frac{br_{2k}^{2b}-t_{2k}}{t_{2k}-br_{2k-1}^{2b}} \bigg) \bigg\}.
\end{align*}
Furthermore, by Lemma \ref{lemma: Theta is expressible in terms of Jacobi}, $\theta_{k} = j_{k,\star}-\lfloor j_{k,\star} \rfloor$, and $j_{k,\star} = n t_{k} -\alpha$,
\begin{align*}
\Theta\bigg( \theta_{2k},\frac{r_{2k-1}}{r_{2k}},\frac{t_{2k}-br_{2k-1}^{2b}}{br_{2k}^{2b}-t_{2k}} \bigg) & = \frac{\log \pi}{2} - \frac{1}{2} \log \bigg( \frac{br_{2k}^{2b}-t_{2k}}{t_{2k}-br_{2k-1}^{2b}} \bigg) + \frac{1}{4} \log \bigg( \frac{r_{2k}}{r_{2k-1}} \bigg) - \frac{1}{2} \log \log \bigg( \frac{r_{2k}}{r_{2k-1}} \bigg) \\
&  + \frac{\big[ \log \big(\frac{br_{2k}^{2b}-t_{2k}}{t_{2k}-br_{2k-1}^{2b}} \big) \big]^{2}}{4 \log \big( \frac{r_{2k}}{r_{2k-1}} \big)} - \sum_{j=1}^{+\infty} \log \bigg( 1-\bigg( \frac{r_{2k-1}}{r_{2k}} \bigg)^{2j} \bigg) \\
& + \log \theta \Bigg(t_{2k}n + \frac{1}{2} - \alpha + \frac{\log \big(\frac{br_{2k}^{2b}-t_{2k}}{t_{2k}-br_{2k-1}^{2b}} \big)}{2 \log \big( \frac{r_{2k}}{r_{2k-1}} \big)} \Bigg| \frac{\pi i}{\log(\frac{r_{2k}}{r_{2k-1}})} \Bigg),
\end{align*}
where we have also used the fact that $\theta(x+1|\tau)=\theta(x|\tau)$. Combining the above two equations yields
\begin{align}
& \sum_{k=2,4,...}^{2g} \widetilde{\Theta}_{k,n} = \mathcal{F}_{n} + \frac{g}{2}\log(\pi) + \sum_{j=1}^{g} \bigg\{ \bigg(\frac{1}{4}+\theta_{2k}^{2}-\theta_{2k}\bigg) \log \bigg( \frac{r_{2k}}{r_{2k-1}} \bigg) - \frac{1}{2} \log \log \bigg( \frac{r_{2k}}{r_{2k-1}} \bigg) \nonumber \\
& + \frac{\big[ \log \big(\frac{br_{2k}^{2b}-t_{2k}}{t_{2k}-br_{2k-1}^{2b}} \big) \big]^{2}}{4 \log \big( \frac{r_{2k}}{r_{2k-1}} \big)} - \sum_{j=1}^{+\infty} \log \bigg( 1-\bigg( \frac{r_{2k-1}}{r_{2k}} \bigg)^{2j} \bigg) + \bigg( \theta_{2k} - \frac{1}{2}\bigg) \log \bigg( \frac{br_{2k}^{2b}-t_{2k}}{t_{2k}-br_{2k-1}^{2b}} \bigg) \bigg\}. \label{lol64}
\end{align}
On the other hand, using Lemmas \ref{lemma:Explicit expression for the F coeff} and \ref{lemma: S2k FINAL all k}, we obtain (after a lot of cancellations)
\begin{align}
\sum_{k=2,4,...}^{2g} F_{6,k}^{(n,\epsilon)} + \sum_{k=1}^{2g} E_{6,k}^{(n,\epsilon)} = \sum_{k=1}^{g} \bigg\{ & \bigg( \theta_{2k}-\theta_{2k}^{2}-\frac{1+b^{2}}{6} \bigg) \log \bigg( \frac{r_{2k}}{r_{2k-1}} \bigg) + \frac{b^{2}r_{2k}^{2b}}{br_{2k}^{2b}-t_{2k}} \nonumber \\
& + \frac{b^{2}r_{2k-1}^{2b}}{t_{2k}-br_{2k-1}^{2b}} + \bigg( \frac{1}{2}-\theta_{2k} \bigg) \log \bigg( \frac{br_{2k}^{2b}-t_{2k}}{t_{2k}-br_{2k-1}^{2b}} \bigg) \bigg\}. \label{lol65}
\end{align}
By combining \eqref{lol64} and \eqref{lol65}, we finally obtain
\begin{align*}
\sum_{k=2,4,...}^{2g} (F_{6,k}^{(n,\epsilon)} +\widetilde{\Theta}_{k,n} ) + \sum_{k=1}^{2g}  E_{6,k}^{(n,\epsilon)} & = C_{6} + \mathcal{F}_{n}.
\end{align*}
This finishes the proof of Theorem \ref{thm:main thm}. 
\end{proof}
\section{Proof of Theorem \ref{thm:main thm 2}: the case $r_{2g}=+\infty$}\label{section:proof 2}
As in Section \ref{section:proof}, we start with \eqref{exact formula for log Pn}, but now we split $\log \mathcal{P}_{n}$ into $4g$ parts
\begin{align}\label{log Dn as a sum of sums 2}
\log \mathcal{P}_{n} = S_{0} + \sum_{k=1}^{2g-1}(S_{2k-1}+S_{2k}) + S_{4g-1},
\end{align}
with 
\begin{align*}
& S_{0} = \sum_{j=1}^{M'} \log \bigg( \sum_{\ell=1}^{2g+1} (-1)^{\ell+1}\frac{\gamma(\tfrac{j+\alpha}{b},nr_{\ell}^{2b})}{\Gamma(\tfrac{j+\alpha}{b})} \bigg),   \\
& S_{2k-1} = \sum_{j=j_{k-1,+}+1}^{j_{k,-}-1} \hspace{-0.3cm} \log \bigg( \sum_{\ell=1}^{2g+1} (-1)^{\ell+1}\frac{\gamma(\tfrac{j+\alpha}{b},nr_{\ell}^{2b})}{\Gamma(\tfrac{j+\alpha}{b})} \bigg), & & k=1,\ldots,2g-1,  \\
& S_{2k} = \sum_{j=j_{k,-}}^{j_{k,+}} \log \bigg( \sum_{\ell=1}^{2g+1} (-1)^{\ell+1}\frac{\gamma(\tfrac{j+\alpha}{b},nr_{\ell}^{2b})}{\Gamma(\tfrac{j+\alpha}{b})} \bigg), & & k=1,\ldots,2g-1,
\end{align*}
and 
\begin{align*}
S_{4g-1} = \sum_{j=j_{2g-1,+}+1}^{n} \hspace{-0.3cm} \log \bigg( \sum_{\ell=1}^{2g+1} (-1)^{\ell+1}\frac{\gamma(\tfrac{j+\alpha}{b},nr_{\ell}^{2b})}{\Gamma(\tfrac{j+\alpha}{b})} \bigg).
\end{align*}
The sums $S_{0},S_{1},\ldots,S_{4g-2}$ can be analyzed exactly as in Section \ref{section:proof}. For the large $n$ asymptotics of these sums, see Lemma \ref{lemma: S0} for $S_{0}$, Lemma \ref{lemma: S2km1 k odd} for $S_{2k-1}$ with $k \in \{1,3,\ldots,2g-1\}$, Lemma \ref{lemma: S2km1 k even} for $S_{2k-1}$ with $k \in \{2,4,\ldots,2g-2\}$, and Lemma \ref{lemma: S2k FINAL all k} for $S_{2k}$ with $k \in \{1,2,\ldots,2g-1\}$. Thus it only remains to determine the large $n$ asymptotics of $S_{4g-1}$ in this section. These asymptotics are stated in the following lemma. 
\begin{lemma}\label{lemma: S2km1 k =2g}
Let $k = 2g$. As $n \to + \infty$, we have
\begin{align*}
& S_{2k-1} = F_{1,k}^{(\epsilon)} n^{2}  + F_{2,k}^{(\epsilon)} n \log n + F_{3,k}^{(n,\epsilon)} n + F_{5,k}^{(n,\epsilon)} \log n + F_{6,k}^{(n,\epsilon)} + \bigO\bigg( \frac{\log n}{n} \bigg),
\end{align*}
where
\begin{align*}
& F_{1,k}^{(\epsilon)} = \frac{br_{k-1}^{4b}}{(1-\epsilon)^{2}}\frac{1-4\epsilon - 2 \log(1-\epsilon)}{4} + \frac{3}{4b} + \frac{1}{2b}\log(br_{k-1}^{2b}) -r_{k-1}^{2b}, \\
& F_{2,k}^{(\epsilon)} = \frac{br_{k-1}^{2b}}{2(1-\epsilon)}-\frac{1}{2}, \\
& F_{3,k}^{(n,\epsilon)} = \frac{r_{k-1}^{2b}}{1-\epsilon}\bigg\{ \frac{2\alpha-1+2\theta_{k-1,+}^{(n,\epsilon)}}{2}(\epsilon + \log(1-\epsilon)) - \frac{b+2\alpha}{2} - b \log b + \frac{b}{2}\log(2\pi) - b^{2}\log(r_{k-1}) \\
&  -\frac{2\alpha-b}{2}\log(1-\epsilon) + b \epsilon \log \bigg( \frac{\epsilon b r_{k-1}^{2b}}{1-\epsilon} \bigg) \bigg\} + \frac{b+2\alpha+1}{2b} - \frac{r_{k-1}^{2b}}{2}+\frac{1}{2}\log \bigg( \frac{b}{2\pi} \bigg) + \frac{1+2\alpha}{2b}\log \big( br_{k-1}^{2b} \big) \\
& -(1-br_{k-1}^{2b})\log \big( 1-br_{k-1}^{2b} \big), \\
& F_{5,k}^{(n,\epsilon)} = -\frac{\theta_{k-1,+}^{(n,\epsilon)}+\alpha}{2}, \\
& F_{6,k}^{(n,\epsilon)} = -\frac{1+3b+b^{2}-6(1+b)\theta_{k-1,+}^{(n,\epsilon)}+6(\theta_{k-1,+}^{(n,\epsilon)})^{2}}{12b} \log(1-\epsilon) -\frac{b}{\epsilon} + \bigg(\frac{1}{2}-\theta_{k-1,+}^{(n,\epsilon)}\bigg)\log \epsilon \\
&  + \bigg( \frac{1}{2}-\alpha-\theta_{k-1,+}^{(n,\epsilon)} \bigg) \log \big( r_{k-1}^{b}\sqrt{2\pi} \big) + \frac{1}{4}\log \bigg( \frac{b}{4\pi} \bigg) - \frac{1+2\alpha}{2}\log \big( 1-br_{k-1}^{2b} \big) + \frac{b^{2}r_{k-1}^{2b}}{1-br_{k-1}^{2b}} \\
& +b + \frac{b^{2}+6b\alpha + 6\alpha^{2}+6\alpha+1}{12b}\log \big( br_{k-1}^{2b} \big).
\end{align*}
\end{lemma}
\begin{proof}
In the same way as in Lemma \ref{lemma: S2km1 k even splitting}, as $n \to + \infty$ we find
\begin{align*}
S_{2k-1} = S_{2k-1}^{(1)}+\bigO(e^{-cn}), \qquad \mbox{where} \qquad S_{2k-1}^{(1)} = \sum_{j=j_{k-1,+}+1}^{n} \log \bigg( \frac{\gamma(\tfrac{j+\alpha}{b},nr_{k-1}^{2b})}{\Gamma(\tfrac{j+\alpha}{b})} \bigg).
\end{align*}
The large $n$ asymptotics of $S_{2k-1}^{(1)}$ can be obtained in a similar (and simpler, because there are no theta functions) way than in Lemma \ref{lemma:S2km1 part 1} using Lemma \ref{lemma:Riemann sum NEW}. We omit further details.
\end{proof}
By substituting the asymptotics of Lemmas \ref{lemma: S0}, \ref{lemma: S2km1 k odd}, \ref{lemma: S2km1 k even}, \ref{lemma: S2k FINAL all k} and \ref{lemma: S2km1 k =2g} in \eqref{log Dn as a sum of sums 2}, and then simplifying, we obtain the statement of Theorem \ref{thm:main thm 2}. 

\section{Proof of Theorem \ref{thm:main thm 3}: the case $r_{1}=0$}\label{section:proof 3}
We use again \eqref{exact formula for log Pn}, but now we split $\log \mathcal{P}_{n}$ into $4g-1$ parts as follows
\begin{align}\label{log Dn as a sum of sums 3}
\log \mathcal{P}_{n} = S_{3} + S_{4} + \sum_{k=3}^{2g}(S_{2k-1}+S_{2k}) + S_{4g+1}
\end{align}
with 
\begin{align*}
& S_{2k-1} = \sum_{j=j_{k-1,+}+1}^{j_{k,-}-1} \hspace{-0.3cm} \log \bigg( \sum_{\ell=1}^{2g+1} (-1)^{\ell+1}\frac{\gamma(\tfrac{j+\alpha}{b},nr_{\ell}^{2b})}{\Gamma(\tfrac{j+\alpha}{b})} \bigg), & & k=3,\ldots,2g+1,  \\
& S_{2k} = \sum_{j=j_{k,-}}^{j_{k,+}} \log \bigg( \sum_{\ell=1}^{2g+1} (-1)^{\ell+1}\frac{\gamma(\tfrac{j+\alpha}{b},nr_{\ell}^{2b})}{\Gamma(\tfrac{j+\alpha}{b})} \bigg), & & k=2,\ldots,2g,
\end{align*}
and 
\begin{align*}
S_{3} = \sum_{j=1}^{j_{2,-}-1} \log \bigg( \sum_{\ell=1}^{2g+1} (-1)^{\ell+1} \frac{\gamma(\tfrac{j+\alpha}{b},nr_{\ell}^{2b})}{\Gamma(\tfrac{j+\alpha}{b})} \bigg).
\end{align*}
The sums $S_{4},S_{5},\ldots,S_{4g+1}$ can be analyzed exactly as in Section \ref{section:proof}. Their large $n$ asymptotics is given by Lemma \ref{lemma: S2km1 k odd} for $S_{2k-1}$ with $k \in \{3,5,\ldots,2g+1\}$, Lemma \ref{lemma: S2km1 k even} for $S_{2k-1}$ with $k \in \{4,6,\ldots,2g\}$, and Lemma \ref{lemma: S2k FINAL all k} for $S_{2k}$ with $k \in \{2,3,\ldots,2g\}$.  Thus it only remains to analyze $S_{3}$ in this section. This analysis is different than in the previous Sections \ref{section:proof} and \ref{section:proof 2} and requires the asymptotics of $\gamma(a,z)$ as $z\to +\infty$ uniformly for $\frac{a}{z} \in [0,\frac{1}{1+\epsilon/2}]$. These asymptotics are not covered by Lemma \ref{lemma: uniform}, but are also known in the literature, see e.g. \cite{Paris}. 
\begin{lemma}(Taken from \cite[Section 4]{Paris}) \label{lemma: Paris}
As $z \to + \infty$ and simultaneously $\frac{z-a}{\sqrt{z}} \to + \infty$, we have
\begin{align*}
\frac{\gamma(a,z)}{\Gamma(a)} = 1-\frac{z^{a}e^{-z}}{\Gamma(a)}\bigg( \frac{1}{z-a} - \frac{z}{(z-a)^{3}} + \bigO(z^{-3}) \bigg).
\end{align*}
\end{lemma}
We are now in a position to obtain the large $n$ asymptotics of $S_{3}$. 
\begin{lemma}\label{lemma: S2km1 k =2}
Let $k = 2$. As $n \to + \infty$, we have
\begin{align*}
& S_{2k-1} = F_{1,k}^{(\epsilon)} n^{2}  + F_{2,k}^{(\epsilon)} n \log n + F_{3,k}^{(n,\epsilon)} n + F_{5,k}^{(n,\epsilon)} \log n + F_{6,k}^{(n,\epsilon)} + \bigO\bigg( \frac{\log n}{n} \bigg),
\end{align*}
where
\begin{align*}
& F_{1,k}^{(\epsilon)} = - \frac{br_{k}^{4b}}{(1+\epsilon)^{2}}\frac{1+4\epsilon - 2 \log(1+\epsilon)}{4}, \\
& F_{2,k}^{(\epsilon)} = -\frac{br_{k}^{2b}}{2(1+\epsilon)}, \\
& F_{3,k}^{(n,\epsilon)} =  \frac{r_{k}^{2b}}{1+\epsilon}\bigg\{ \big( 1+\alpha - \theta_{k,-}^{(n,\epsilon)} \big)\epsilon - b^{2} \log(r_{k}) + \alpha + \frac{1}{2} + \frac{b}{2} + b \epsilon \log \bigg( \frac{\epsilon}{1+\epsilon} \bigg) - \frac{b}{2}\log(2\pi) \\
& + \frac{2\theta_{k,-}^{(n,\epsilon)}-1-b}{2}\log(1+\epsilon) \bigg\}, \\
& F_{5,k}^{(n,\epsilon)} = -\frac{1+b^{2}+6\alpha+6\alpha^{2}-3b(3+4\alpha)}{12b}-\frac{\theta_{k,-}^{(n,\epsilon)}}{2}, \\
& F_{6,k}^{(n,\epsilon)} = \frac{1-3b+b^{2}+6(b-1)\theta_{k,-}^{(n,\epsilon)}+6(\theta_{k,-}^{(n,\epsilon)})^{2}}{12b} \log(1+\epsilon) -\frac{b}{\epsilon} + \frac{1-2\theta_{k,-}^{(n,\epsilon)}}{2} \log \epsilon \\
& + \bigg( 2b(1+\alpha)-\alpha-\alpha^{2}-\frac{1+3b+b^{2}}{6} \bigg) \log(r_{k}) + \frac{\alpha+1}{2}\log(2\pi) - \theta_{k,-}^{(n,\epsilon)} \log \big( r_{k}^{b}\sqrt{2\pi} \big) \\
& - \frac{1-3b+b^{2}+6\alpha-6b\alpha + 6 \alpha^{2}}{12b}\log(b) - \mathcal{G}(b,\alpha),
\end{align*}
where $\mathcal{G}(b,\alpha)$ is defined in \eqref{lol17}.
\end{lemma}
\begin{proof}
In a similar way as in Lemma \ref{lemma: S2km1 k even splitting}, as $n \to + \infty$ we find
\begin{align*}
S_{2k-1} = S_{2k-1}^{(2)}+\bigO(e^{-cn}), \qquad \mbox{where} \qquad S_{2k-1}^{(2)} = \sum_{j=1}^{j_{k,-}-1} \log \bigg( 1- \frac{\gamma(\tfrac{j+\alpha}{b},nr_{k-1}^{2b})}{\Gamma(\tfrac{j+\alpha}{b})} \bigg).
\end{align*}
Using Lemma \ref{lemma: Paris}, we conclude that as $n \to + \infty$, 
\begin{align}\label{lol54}
& S_{2k-1} = -\sum_{j=1}^{j_{k,-}-1} \log \Gamma \big( \tfrac{j+\alpha}{b} \big) +  \sum_{j=1}^{j_{k,-}-1} \bigg\{\frac{j/n}{b}n \log n + \Big( 2 \log (r_{k}) j/n - r_{k}^{2b} \Big) n  + \frac{\alpha-b}{b}\log n  \\
& +   2\alpha \log r_{k} - \log \bigg(\frac{br_{k}^{2b}-j/n}{b}\bigg)  + \frac{1}{n} \frac{- \alpha j/n - b(b-\alpha)r_{k}^{2b}}{(j/n-br_{k}^{2b})^{2}} \bigg\} + \bigO(n^{-1}).
\end{align}
Using \cite[formula 5.11.1]{NIST}, namely
\begin{align*}
\log \Gamma(z) = z \log z - z - \frac{\log z}{2} + \frac{\log 2\pi}{2} + \frac{1}{12 z} + \bigO(z^{-3}), \qquad \mbox{as } z \to + \infty,
\end{align*}
we obtain
\begin{align*}
& \sum_{j=1}^{j_{k,-}-1} \log \Gamma \big( \tfrac{j+\alpha}{b} \big) = \frac{br_{k}^{4b}}{2(1+\epsilon)^{2}} n^{2}\log n - \frac{br_{k}^{4b}}{4(1+\epsilon)^{2}} \bigg( 3-2\log \bigg( \frac{r_{k}^{2b}}{1+\epsilon} \bigg) \bigg) n^{2}  \\
& + \frac{2\theta_{k,-}^{(n,\epsilon)}-1-b}{2(1+\epsilon)}r_{k}^{2b} n \log n + \frac{b \log(2\pi) + (2\theta_{k,-}^{(n,\epsilon)}-1-b)(\log(\frac{r_{k}^{2b}}{1+\epsilon})-1)}{2(1+\epsilon)}r_{k}^{2b}n  \\
& + \frac{1+3b+b^{2}-6(1+b)\theta_{k,-}^{(n,\epsilon)}+6(\theta_{k,-}^{(n,\epsilon)})^{2}}{12b} \log \bigg( \frac{nr_{k}^{2b}}{1+\epsilon} \bigg) + \frac{\theta_{k,-}^{(n,\epsilon)}-\alpha-1}{2}\log(2\pi) \\
& + \frac{1-3b+b^{2}+6\alpha-6b\alpha+6\alpha^{2}}{12b}\log b + \mathcal{G}(b,\alpha) + \bigO(n^{-1}), \qquad \mbox{as } n \to + \infty.
\end{align*}
The second sum on the right-hand side of \eqref{lol54} can be expanded explicitly using Lemma \ref{lemma:Riemann sum NEW}, and after a long computation we obtain
\begin{align*}
& \sum_{j=1}^{j_{k,-}-1} \bigg\{\frac{j/n}{b}n \log n + \Big( 2 \log (r_{k}) j/n - r_{k}^{2b} \Big) n  + \frac{\alpha-b}{b}\log n +   2\alpha \log r_{k} - \log \bigg(\frac{br_{k}^{2b}-j/n}{b}\bigg) \\
&   + \frac{1}{n} \frac{- \alpha j/n - b(b-\alpha)r_{k}^{2b}}{(j/n-br_{k}^{2b})^{2}} \bigg\} = \frac{br_{k}^{4b}}{2(1+\epsilon)^{2}}n^{2}\log n + \bigg( \frac{b^{2}r_{k}^{4b}}{(1+\epsilon)^{2}}\log(r_{k}) - \frac{br_{k}^{4b}}{1+\epsilon} \bigg)n^{2} \\
& \frac{r_{k}^{2b}}{1+\epsilon}\frac{2\theta_{k,-}^{(n,\epsilon)}-1-2b}{2}n \log n + r_{k}^{2b} \bigg\{ \bigg(  \frac{b(2\theta_{k,-}^{(n,\epsilon)}-1)}{1+\epsilon}-2b^{2} \bigg) \log r_{k} -\theta_{k,-}^{(n,\epsilon)} + \alpha + 1 + \frac{b}{1+\epsilon} \\
& + \frac{\epsilon b}{1+\epsilon} \log \bigg( \frac{\epsilon r_{k}^{2b}}{1+\epsilon} \bigg) \bigg\} n + \frac{(2b-\alpha-\theta_{k,-}^{(n,\epsilon)})(1+\alpha-\theta_{k,-}^{(n,\epsilon)})}{2b}\log n + \Big( 2b(1+\alpha)-\alpha-\alpha^{2} \\
& -(1+2b)\theta_{k,-}^{(n,\epsilon)}+(\theta_{k,-}^{(n,\epsilon)})^{2} \Big) \log r_{k} + \frac{1-2\theta_{k,-}^{(n,\epsilon)}}{2}\log \epsilon -\frac{b}{\epsilon}.
\end{align*}
Substituting the last two asymptotic formulas in \eqref{lol54} and simplifying, we obtain the claim.
\end{proof}
By combining the asymptotics of Lemmas \ref{lemma: S2km1 k odd}, \ref{lemma: S2km1 k even}, \ref{lemma: S2k FINAL all k} and \ref{lemma: S2km1 k =2g} with \eqref{log Dn as a sum of sums 3}, and then simplifying, we obtain the statement of Theorem \ref{thm:main thm 3}.
\section{Proof of Theorem \ref{thm:main thm 4}: the case $r_{2g}=+\infty$ and $r_{1}=0$}
We split $\log \mathcal{P}_{n}$ into $4g-3$ parts
\begin{align}\label{log Dn as a sum of sums 4}
\log \mathcal{P}_{n} = S_{3} + S_{4} + \sum_{k=3}^{2g-1}(S_{2k-1}+S_{2k}) + S_{4g-1}
\end{align}
with 
\begin{align*}
& S_{2k-1} = \sum_{j=j_{k-1,+}+1}^{j_{k,-}-1} \hspace{-0.3cm} \log \bigg( \sum_{\ell=1}^{2g+1} (-1)^{\ell+1}\frac{\gamma(\tfrac{j+\alpha}{b},nr_{\ell}^{2b})}{\Gamma(\tfrac{j+\alpha}{b})} \bigg), & & k=3,\ldots,2g-1,  \\
& S_{2k} = \sum_{j=j_{k,-}}^{j_{k,+}} \log \bigg( \sum_{\ell=1}^{2g+1} (-1)^{\ell+1}\frac{\gamma(\tfrac{j+\alpha}{b},nr_{\ell}^{2b})}{\Gamma(\tfrac{j+\alpha}{b})} \bigg), & & k=2,\ldots,2g-1,
\end{align*}
and 
\begin{align*}
S_{3} = \sum_{j=1}^{j_{2,-}-1} \log \bigg( \sum_{\ell=1}^{2g+1} (-1)^{\ell+1} \frac{\gamma(\tfrac{j+\alpha}{b},nr_{\ell}^{2b})}{\Gamma(\tfrac{j+\alpha}{b})} \bigg), \quad S_{4g-1} = \hspace{-0.1cm} \sum_{j=j_{2g-1,+}+1}^{n} \hspace{-0.3cm} \log \bigg( \sum_{\ell=1}^{2g+1} (-1)^{\ell+1}\frac{\gamma(\tfrac{j+\alpha}{b},nr_{\ell}^{2b})}{\Gamma(\tfrac{j+\alpha}{b})} \bigg).
\end{align*}
The sums $S_{4},S_{5},\ldots,S_{4g-2}$ can be analyzed exactly as in Section \ref{section:proof}, $S_{4g-1}$ can be analyzed as in Section \ref{section:proof 2}, and $S_{3}$ as in Section \ref{section:proof 3}. More precisely, their large $n$ asymptotics are given by Lemma \ref{lemma: S2km1 k odd} for $S_{2k-1}$ with $k \in \{3,5,\ldots,2g-1\}$, Lemma \ref{lemma: S2km1 k even} for $S_{2k-1}$ with $k \in \{4,6,\ldots,2g-2\}$, Lemma \ref{lemma: S2k FINAL all k} for $S_{2k}$ with $k \in \{2,3,\ldots,2g-1\}$, Lemma \ref{lemma: S2km1 k =2g} for $S_{4g-1}$, and Lemma \ref{lemma: S2km1 k =2} for $S_{3}$. Substituting all these asymptotics in \eqref{log Dn as a sum of sums 4} and simplifying, we obtain the asymptotic formula of Theorem \ref{thm:main thm 4}.
\paragraph{Acknowledgements.} The author owes a great debt of gratitude to Oleg Lisovyy for pointing out that $\Theta$ in \eqref{def of Theta} is related to the Jacobi theta function as in \eqref{first identity for Theta} via the Jacobi triple product formula. The author is also grateful to Yacin Ameur, Sung-Soo Byun and Gr\'{e}gory Schehr for useful remarks, and to the referee for a careful reading. Support is acknowledged from the European Research Council, Grant Agreement No. 682537, the Swedish Research Council, Grant No. 2015-05430, the Ruth and Nils-Erik Stenb\"ack Foundation, and the Novo Nordisk Fonden Project Grant 0064428.

\end{document}